\documentclass[sigplan]{acmart}

\usepackage{booktabs}   %% For formal tables:
\usepackage{subcaption} %% For complex figures with subfigures/subcaptions
\usepackage{tablefootnote}

\usepackage{amsmath,amssymb} 
\usepackage{url}
\usepackage{xspace}
\usepackage{enumitem}

\usepackage{varwidth,algorithm,algorithmicx}
\usepackage[noend]{algpseudocode}

\usepackage{tikz}
\usetikzlibrary{arrows,snakes,backgrounds,automata,shapes,positioning,chains,calc}

\usepackage[inference]{semantic}

\makeatletter
\newcommand{\xRightarrow}[2][]{\ext@arrow 0359\Rightarrowfill@{#1}{#2}}
\makeatother

\newcommand{\kwswap}{\textsf{\textbf{swap}}}

\newcommand{\kwskip}{\textsf{\textbf{skip}}}
\newcommand{\kwdo}{\textsf{\textbf{do}}}
\newcommand{\kwwhile}{\textsf{\textbf{while}}}
\newcommand{\kwend}{\textsf{\textbf{end}}}
\newcommand{\kwif}{\textsf{\textbf{if}}}
\newcommand{\kwthen}{\textsf{\textbf{then}}}
\newcommand{\kwelse}{\textsf{\textbf{else}}}
\newcommand{\kwreturn}{\textsf{\textbf{return}}}
\newcommand{\kwthread}{\textsf{\textbf{thread}}}

\algnewcommand\Swap{\kwswap}
\algnewcommand\Skip{\kwskip}
\algnewcommand\Thread{\kwthread}

\algrenewcommand\algorithmicend{\kwend}
\algrenewcommand\algorithmicdo{\kwdo}
\algrenewcommand\algorithmicwhile{\kwwhile}
\algrenewcommand\algorithmicfor{\textsf{\textbf{for}}}
\algrenewcommand\algorithmicforall{\textsf{\textbf{for all}}}
\algrenewcommand\algorithmicloop{\textsf{\textbf{loop}}}
\algrenewcommand\algorithmicrepeat{\textsf{\textbf{repeat}}}
\algrenewcommand\algorithmicuntil{\textsf{\textbf{until}}}
\algrenewcommand\algorithmicprocedure{\textsf{\textbf{procedure}}}
\algrenewcommand\algorithmicfunction{\textsf{\textbf{function}}}
\algrenewcommand\algorithmicif{\kwif}
\algrenewcommand\algorithmicthen{\kwthen}
\algrenewcommand\algorithmicelse{\kwelse}
\algrenewcommand\algorithmicreturn{\kwreturn}

\algblockdefx{Thread}{EndThread}%
[1]{\kwthread \xspace #1}%
{\algorithmicend}

\algblockdefx{MyWhile}{EndMyWhile}%
[1]{\kwwhile \xspace #1}%
{\algorithmicend}

\makeatletter
\ifthenelse{\equal{\ALG@noend}{t}}%
  {\algtext*{EndMyWhile}}
  {}%
\makeatother
\makeatletter
\ifthenelse{\equal{\ALG@noend}{t}}%
  {\algtext*{EndThread}}
  {}%
\makeatother

\usepackage{caption, subcaption}
\newcommand{\dom}{\operatorname{\mathbf{dom}}}
\newcommand{\ran}{\operatorname{\mathbf{ran}}}
\newcommand{\OMIT}[1]{}

\setcopyright{none}

\acmYear{2019}
\acmISBN{} % \acmISBN{978-x-xxxx-xxxx-x/YY/MM}
\acmDOI{} % \acmDOI{10.1145/nnnnnnn.nnnnnnn}
\startPage{1}

\begin{document}

%% Title information

\title[Verifying C11 Programs Operationally]{Verifying C11 Programs Operationally}

%% Author with single affiliation.
\author{Simon Doherty}
% \authornote{with author1 note}          %% \authornote is optional;
%                                         %% can be repeated if necessary
% \orcid{nnnn-nnnn-nnnn-nnnn}             %% \orcid is optional
\affiliation{
  \department{Department of Computer Science}              %% \department is recommended
  \institution{University of Sheffield}            %% \institution is required
  % \country{UK}                    %% \country is recommended
}
\email{s.doherty@sheffield.ac.uk}          %% \email is recommended

\author{Brijesh Dongol}
%\authornote{with author1 note}          %% \authornote is optional;
                                        %% can be repeated if necessary
\orcid{nnnn-nnnn-nnnn-nnnn}             %% \orcid is optional
\affiliation{
  \department{Department of Computer Science}              %% \department is recommended
  \institution{University of Surrey}            %% \institution is required
  % \country{UK}                    %% \country is recommended
}
\email{b.dongol@surrey.ac.uk}          %% \email is recommended

%% Author with single affiliation.
\author{Heike Wehrheim}
% \authornote{with author1 note}          %% \authornote is optional;
%                                         %% can be repeated if necessary
% \orcid{nnnn-nnnn-nnnn-nnnn}             %% \orcid is optional
\affiliation{
  \department{Department of Computer Science}              %% \department is recommended
  \institution{University of Paderborn}            %% \institution is required
  % \country{Germany}                    %% \country is recommended
}
\email{wehrheim@upb.de}          %% \email is recommended

%% Author with single affiliation.
\author{John Derrick}
% \authornote{with author1 note}          %% \authornote is optional;
%                                         %% can be repeated if necessary
% \orcid{nnnn-nnnn-nnnn-nnnn}             %% \orcid is optional
\affiliation{
  \department{Department of Computer Science}              %% \department is recommended
  \institution{University of Sheffield}            %% \institution is required
  % \country{UK}                    %% \country is recommended
}
\email{j.derrick@sheffield.ac.uk}          %% \email is recommended

% \titlerunning{Operational Semantics for C11}
% \authorrunning{S. Doherty et al}

%% 2012 ACM Computing Classification System (CSS) concepts
%% Generate at 'http://dl.acm.org/ccs/ccs.cfm'.

%% End of generated code

%% Keywords
%% comma separated list
 %% \keywords are mandatory in final camera-ready submission

% definition of extendible "transition arrows"
\newcommand{\linefill}{\cleaders\hbox{$\smash{\mkern-2mu\mathord-\mkern-2mu}$}\hfill\vphantom{\lower1pt\hbox{$\rightarrow$}}}  
\newcommand{\Linefill}{\cleaders\hbox{$\smash{\mkern-2mu\mathord=\mkern-2mu}$}\hfill\vphantom{\hbox{$\Rightarrow$}}}  
\newcommand{\transi}[2]{\mathrel{\lower1pt\hbox{$\mathrel-_{\vphantom{#2}}\mkern-8mu\stackrel{#1}{\linefill_{\vphantom{#2}}}\mkern-11mu\rightarrow_{#2}$}}}
\newcommand{\trans}[1]{\transi{#1}{{}}}
\newcommand{\transo}{\mathord{\trans{~}}}
\newcommand{\ntransi}[2]{\mathrel{\lower1pt\hbox{$\mathrel-_{\vphantom{#2}}\mkern-8mu\stackrel{#1}{\linefill_{\vphantom{#2}}}\mkern-8mu\nrightarrow_{#2}$}}}
\newcommand{\ntrans}[1]{\ntransi{#1}{{}}}

\newcounter{sarrow}
\newcommand\strans[1]{%
  \mathrel{\raisebox{0.1em}{
    \stepcounter{sarrow}%
    \!\!\!\!
    \begin{tikzpicture}
      \node[inner sep=.5ex] (\thesarrow) {$\scriptstyle #1$};
      \path[draw,<-,decorate,line width=0.25mm,
      decoration={zigzag,amplitude=0.7pt,segment length=1.2mm,pre=lineto,pre length=4pt}] 
      (\thesarrow.south east) -- (\thesarrow.south west);
    \end{tikzpicture}%
  }}}

\newcommand\Strans[1]{%
\mathrel{\raisebox{0.1em}{
\!\!\begin{tikzpicture}
  \node[inner sep=0.6ex] (a) {$\scriptstyle #1$};
  \path[line width=0.2mm, draw,implies-,double distance between line
  centers=1.5pt,decorate, 
    decoration={zigzag,amplitude=0.7pt,segment length=1.2mm,pre=lineto,
    pre   length=4pt}] 
    (a.south east) -- (a.south west);
\end{tikzpicture}}%
}}

\newcommand{\legal}{{\it legal}}
\newcommand{\Ws}{{\it Ws}}
\newcommand{\Rs}{{\it Rs}}
\newcommand{\calR}{{\cal R}}

\newcommand{\bbT}{\mathbb{T}}

\newcommand{\calE}{{\cal E}}
\newcommand{\nat}{\mathbb{N}}
\newcommand{\noteq}{\neq}

% ???
\newcommand{\lt}{{\bf Less than}}

\newcommand{\ev}{\mathit{ev}}
\newcommand{\Events}{\mathit{Evt}}
\newcommand{\Inv}{\mathit{Inv}}
\newcommand{\Resp}{\mathit{Res}}
\newcommand{\his}{\mathit{his}}
\newcommand{\exec}{\mathit{exec}}
\newcommand{\complete}{\mathit{complete}}
\newcommand{\Var}{\mathsf{Var}}
\newcommand{\Val}{\mathsf{Val}}
\newcommand{\CASOp}{\mathit{CAS}} 
\newcommand{\WR}{\mathsf{Wr_R}}
\newcommand{\RA}{\mathsf{Rd_A}}
\newcommand{\R}{\mathsf{Rd}}
\newcommand{\RX}{\mathsf{Rd_X}}
\newcommand{\E}{\mathsf{E}}
\newcommand{\W}{\mathsf{Wr}}
\newcommand{\IW}{\mathsf{IWr}}
\newcommand{\WX}{\mathsf{Wr_X}}
\newcommand{\CRA}{\mathsf{CRA}}
\newcommand{\URA}{\mathsf{U}}

\newcommand{\HB}{{\sf hb}\xspace} 
\newcommand{\PO}{{\sf po}\xspace}
\newcommand{\MO}{{\sf mo}\xspace}
\newcommand{\SC}{{\sf sc}\xspace}
\newcommand{\RF}{{\sf rf}\xspace}
\newcommand{\SB}{{\sf sb}\xspace}

\newcommand{\ok}{{\tt ok}}

\newcommand{\start}[1]{\texttt{Begin}_{#1}}
\newcommand{\starto}[1]{(\start{#1},\ok)}
\newcommand{\starta}[1]{(\start{#1},\bot)}

\newcommand{\commitWriter}[1]{\texttt{CommitWriter}_{#1}}
\newcommand{\commitReadOnly}[1]{\texttt{CommitReadOnly}_{#1}}

\newcommand{\commit}[1]{\texttt{Commit}_{#1}}
\newcommand{\commito}[1]{(\commit{#1},\ok)}
\newcommand{\commita}[1]{(\commit{#1},\bot)}

\newcommand{\readcall}[1]{\texttt{Read}_{#1}}

\newcommand{\reado}[3]{(\readcall{#1}(#2), #3)}
\newcommand{\reada}[2]{(\readcall{#1}(#2),\bot)}

\newcommand{\writecall}[1]{\texttt{Write}_{#1}}
\newcommand{\writeo}[3]{(\writecall{#1}(#2,#3),\ok)}
\newcommand{\writeos}[3]{\writecall{#1}(#2,#3)}
\newcommand{\writea}[3]{(\writecall{#1}(#2,#3),\bot)}

\newcommand{\support}{\mathit{supp}}

\newcommand{\reads}{\mathit{Reads}}
\newcommand{\writes}{\mathit{Writes}}
\newcommand{\commits}{\mathit{Commits}}

\newcommand{\readas}{\mathit{Reads}^\bot}
\newcommand{\writeas}{\mathit{Writes}^\bot}

\newcommand{\idle}{\mathit{idle}}
\newcommand{\live}{\mathit{live}}

\newcommand{\fv}{\mathit{fv}}

\newcommand{\refeq}[1]{(\ref{#1})}
\newcommand{\fr}{{\sf fr}}
\newcommand{\ltsb}{{\sf sb}}
\newcommand{\lteco}{{\sf eco}}
\newcommand{\ltrf}{\mathord{\sf rf}}
\newcommand{\ltfr}{{\sf fr}}
\newcommand{\lthb}{{\sf hb}}
\newcommand{\ltsw}{{\sf sw}}
\newcommand{\sw}{{\sf sw}}
\newcommand{\ltmox}{{\sf mo}^x}
\newcommand{\ltmo}{{\sf mo}}
\newcommand{\PreExec}{{\it PE}}
\newcommand{\Approx}{{\it Pad}}
\newcommand{\Seq}{{\it Seq}}

\newcommand{\True}{{\it true}}
\newcommand{\False}{{\it false}}

\newcommand{\justified}{justified\xspace}
\newcommand{\notjustified}{unjustified\xspace}
\newcommand{\Justified}{Justified\xspace}
\newcommand{\Notjustified}{Unjustified\xspace}

\newcommand{\rdval}{{\it rdval}}
\newcommand{\wrval}{{\it wrval}}
\newcommand{\loc}{{\it var}}
\newcommand{\var}{{\it var}}

\newcommand{\imp}{\Rightarrow}
\newcommand{\expr}{\mathit{Exp}}
 
\newcommand{\restr}{\downharpoonright}
\newcommand{\ordupto}{\downarrow}
\newcommand{\ordabove}{\uparrow}

\newcommand{\observedWrites}{\mathit{E\!W}\!}
\newcommand{\OW}{\mathit{O\!W}\!}
\newcommand{\CW}{\mathit{C\!W}\!}
\newcommand{\val}{\mathit{val}}
\newcommand{\notT}{\hat{t}}
\newcommand{\pc}{\mathit{pc}}

\newcommand{\PaderStepRel}{\mathit{PaderbornStep}}

\newcommand*{\HardArrow}{\mathbin{\tikz
    [baseline=-0.15ex,-latex] \draw[->] (0pt,0.5ex) --
    (1.1em,0.5ex);}}%

\newcommand*{\SoftArrow}[1][]{\mathbin{\tikz
    [baseline=-0.15ex,-latex,densely dashed, ->] \draw[->] (0pt,0.5ex) --
    (1.1em,0.5ex);}}%

\newcommand*{\DoubleArrow}{\mathbin{\tikz
    [baseline=-0.25ex,-latex] \draw[->>] (0pt,0.5ex) --
    (1.1em,0.5ex);}}%

\newcommand{\seqarrow}{\DoubleArrow}
\newcommand{\seqtempord}{\DoubleArrow}
\newcommand{\seqcausord}{\prec}
\newcommand{\seqrestr}{\mathop{\downharpoonright}}

%% \maketitle
%% Note: \maketitle command must come after title commands, author
%% commands, abstract environment, Computing Classification System
%% environment and commands, and keywords command.
\begin{abstract}
  This paper develops an operational semantics for a release-acquire
  fragment of the C11 memory model with relaxed accesses. We show that
  the semantics is both sound and complete with respect to the
  axiomatic model. The semantics relies on a per-thread notion of
  observability, which allows one to reason about a weak memory C11
  program in program order. On top of this, we develop a proof
  calculus for invariant-based reasoning, which we use to verify the
  release-acquire version of Peterson's mutual exclusion algorithm.
\end{abstract}

\maketitle

\renewcommand{\shortauthors}{S. Doherty, B. Dongol, H. Wehrheim, and J. Derrick}

\newcommand{\eqrng}[2]{(\ref{#1}-\ref{#2})}
\newcommand{\refalg}[1]{Algorithm~\ref{#1}}
\newcommand{\refprop}[1]{Proposition~\ref{#1}}
\newcommand{\reffig}[1]{Figure~\ref{#1}}
\newcommand{\refthm}[1]{Theorem~\ref{#1}}
\newcommand{\reflem}[1]{Lem\-ma~\ref{#1}}
\newcommand{\refcor}[1]{Corollary~\ref{#1}}
\newcommand{\refsec}[1]{Section~\ref{#1}}
\newcommand{\refex}[1]{Example~\ref{#1}}
\newcommand{\refdef}[1]{Definition~\ref{#1}}
\newcommand{\reflst}[1]{Listing~\ref{#1}}
\newcommand{\refchap}[1]{Chapter~\ref{#1}}
\newcommand{\reftab}[1]{Table~\ref{#1}}

\newcommand{\LE}{\mathit{LE}}
 
\newcommand{\WSC}{{\sf WSC}}
\newcommand{\USC}{{\sf U}}
\newcommand{\RSC}{{\sf RSC}}

\newcommand{\AVar}{{\it AVar}}

\renewcommand{\SC}{{\sf SC}\xspace}
\newcommand{\ltsc}{{\sf sc}\xspace}
\newcommand{\PInv}{{\it PInv}}

\newcommand{\action}[3]{\ensuremath{
\begin{array}[t]{l@{~}l}
\multicolumn{2}{l}{#1}\\
\textsf{Pre:}&#2\\
\textsf{Eff:}&#3
\end{array}
}}

\renewcommand{\iff}{\Leftrightarrow}

\tikzset{
    mo/.style={dashed,->,>=stealth,thick,black!20!purple},
    hb/.style={solid,->,>=stealth,thick,blue},
    sw/.style={solid,->,>=stealth,thick,black!50!green},
    rf/.style={dashed,->,>=stealth,thick,black!50!green},
    fr/.style={dashed,->,>=stealth,thick,red}
 }

\newcommand{\alib}{{\cal M}}

\newcommand{\emptyseq}{\langle~\rangle}
\newcommand{\ord}{{\it ord}}

\newcommand{\restrict}{\downharpoonright}

\newcommand{\ES}{\mathsf{ES}}

\newcommand{\bbA}{\mathbb{A}}
\newcommand{\bbB}{\mathbb{B}}
\newcommand{\bbC}{\mathbb{C}}
\newcommand{\bbD}{\mathbb{D}}
\newcommand{\causalord}[1]{\prec_{#1}}
\newcommand{\causalC}{\causalord{\bbC}}
\newcommand{\causalLC}{\causalord{{\cal L}, \bbC}}

\newcommand{\po}{{\sf po}}

\newcommand{\libf}{\mathcal}
\newcommand{\Prog}{{\it Prog}}
\newcommand{\Comm}{{\it Com}}
\newcommand{\Exp}{{\it Exp}}
\newcommand{\PComm}{Comm_{\parallel}}
\newcommand{\whilestep}[1]{\stackrel{#1}{\longrightarrow}}
\newcommand{\parastep}[1]{\stackrel{#1}{\rightarrow_{\parallel}}}

\newcommand{\whileEnd}{\mathbf{end}}

\newcommand{\traceArrow}[1]{\stackrel{#1}{\longrightarrow}}
\newcommand{\ltsArrow}[1]{\stackrel{#1}{\Longrightarrow}}

\newcommand{\enter}{{\it enter}}
\newcommand{\critical}{{\it critical}}
\newcommand{\exit}{{\it exit}}

\newcommand{\Act}{{\sf Act}}

\newcommand{\kwtag}{{\it tag}}
\newcommand{\tid}{{\it tid}}
\newcommand{\act}{{\it act}}
\newcommand{\Op}{A}

%TODO: symbol for these.
\newcommand{\emptymap}{emptymap}
\newcommand{\emptyord}{empord}

\section{Introduction}

Intensive research on the correctness of shared-memory concurrent
programs over the last three decades has resulted in a variety of
tools and techniques. However, the vast majority of these
have been developed on the assumption of \emph{sequential
  consistency}~\cite{DBLP:journals/tc/Lamport79}. Programs running on modern
hardware execute using \emph{weak memory models}
\cite{DBLP:journals/computer/AdveG96}, requiring many of these
techniques to be reworked.

This paper is focused on the C11 memory model, which has been the
topic of several recent papers (e.g.,
\cite{DBLP:conf/pldi/LahavVKHD17,DBLP:conf/popl/KangHLVD17,DBLP:conf/popl/LidburyD17,DBLP:conf/ecoop/KaiserDDLV17,DBLP:conf/popl/BattyDW16,DBLP:conf/vmcai/DokoV16,DBLP:conf/oopsla/NienhuisMS16,DBLP:conf/pdp/HeVQF16,DBLP:conf/popl/LahavGV16,DBLP:conf/popl/BattyDG13,DBLP:conf/popl/BattyOSSW11}).
Typically the C11 memory model is described using an axiomatic
semantics
\cite{DBLP:conf/popl/BattyDG13,DBLP:conf/popl/BattyDW16,DBLP:conf/popl/BattyOSSW11,DBLP:conf/pldi/LahavVKHD17}
via a two-step procedure. (1)~Construct {\em candidate executions} of
a program comprising low-level (e.g., read/write operations) in which
reads may return an arbitrary value. (2)~Apply a number of axioms over
the memory model to rule out invalid candidate executions. Such axioms
may state, for instance, that every read is validated by a write that
has written the value read. Of particular interest are axioms that
exclude certain cycles from arising.  However precise, axiomatic
definitions are unsuitable for program verification (in particular,
those involving invariant-based reasoning), which requires one to
consider the step-wise execution of a program. There has therefore
been a substantial effort to develop an operational semantics: for
weak memory models in general
\cite{DBLP:conf/ecoop/KaiserDDLV17,DBLP:conf/popl/KangHLVD17,LahavV15}
and for C11
specifically~\cite{DBLP:conf/oopsla/NienhuisMS16,DBLP:conf/popl/LidburyD17}.

Our key goal in this paper is to develop an operational model that
supports verification of weak memory C11 programs.  Like many
programming languages, C11 has several advanced features, e.g.,
speculation, that contributes to the complexity of the logics for
reasoning about them. Some operational models
(e.g.,~\cite{DBLP:conf/oopsla/NienhuisMS16}) attempt to deal with the
full complexity of the language and its behaviour.  Other models focus
on a well-behaved and well-understood fragment (e.g.,~\cite{LahavV15,
  DBLP:conf/popl/KangHLVD17}).  In order to support an intuitive
verification method, we take the latter course.  We do \emph{not}
handle some-forms of speculation (thin-air reads), release sequences,
non-atomic accesses or sequentially consistent accesses.  This leaves
us with the so-called \emph{RAR
  fragment}~\cite{DBLP:conf/popl/BattyDW16} of C11 (see
\refsec{sec:rc11-memory-model}), where $\ltsb \cup \ltrf$ is acyclic,
and thus dependencies between operations are easier to manage. All
read/write/update operations are either \emph{relaxed} or
\emph{synchronised} via release-acquire annotations. Acyclicity of
$\ltsb \cup \ltrf$ precludes behaviours allowed by hardware
architectures such as Power \cite{DBLP:conf/fm/LahavV16}. Thus, to
ensure programs proved correct by our logic remain sound, one must
ensure adequate fencing of independent instructions during compilation
(see \cite{DBLP:conf/fm/LahavV16} for details).

This paper comprises three main contributions. The first contribution
is an operational semantics for the RAR fragment that we prove to be
both \textit{sound} and \textit{complete} with respect to the
axiomatic definition. Our semantics (like
\cite{PodkopaevSN16,DBLP:conf/popl/KangHLVD17}) allows each thread to
have its own (per-thread) \emph{observations} of memory. We build on
the recently proposed \emph{extended coherence
  order}~\cite{DBLP:conf/pldi/LahavVKHD17} (which is the transitive
closure of the \emph{communication relation}
in~\cite{AlglaveMT14}). The extended coherence order describes the
order of reads and writes to a variable (see \refex{ex:eco}), which in
turn enables one to define how events may be introduced in a valid C11
execution without violating validity of the axioms.

We combine the extended coherence order with the causality relation of
C11 (formalised by \emph{happens-before}) to define the set of writes
already \emph{encountered} by each thread. This set is in turn used to
define the writes observable by the thread (see
\refsec{sec:observ-event-semant}). Our operational semantics naturally
builds on observability: reads are validated on-the-fly (as opposed to
a post-hoc manner in the axiomatic semantics). Thus, each state
constructed using the transition relations of our operational
semantics is a valid C11 state (see
\refsec{sec:soundn-compl}). Moreover, we show that any candidate
execution that is valid according to the axiomatic semantics can be
generated by our operational semantics. %That is, our semantics is both
%sound and complete for validating C11 programs.

The second contribution is a verification technique that builds on the
operational semantics to enable inductive reasoning over the program
steps.  One difficulty in using an operational semantics of
weak-memory to support verification is the fact that the state spaces
of such operational models are far more complicated than the state
space that one would use for a verification over sequentially
consistent memory, where the shared store can be represented using a
simple mapping from variables to values. We address this issue by
developing a notation that builds on conventional reasoning (over
sequentially consistent memory).  For example, we include assertions
that ensure a thread will read a particular value in a C11 state and
assertions that ensure happens-before order between writes to
different variables. The former is analogous to equations on variables
and their values in the conventional setting; the latter has no direct
analogue in a sequentially consistent setting (the closest analogue is
the use of auxiliary variables \cite{DBLP:books/garland/Owicki75} to
record whether certain operations have already occurred).

Our third contribution is the demonstration of the utility of our verification method by proving the mutual exclusion
property of a C11 version of Peterson's algorithm \cite{PetersonBlog}.

\section{Command Language} % and Uninterpreted Semantics}
\label{sec:comm-lang-unint}
This section describes our command language and defines its
\emph{uninterpreted operational semantics}; namely, an operational
semantics that generate the read, write or update \emph{action} for
each step of the corresponding command. These actions are in turn used
to generate state transitions in \refsec{sec:paderb-semant-c11}, where
the reads and writes are interpreted in a C11 state. Such a decoupled
approach is inspired by the approach taken by Lahav et
al.~\cite{DBLP:conf/popl/LahavGV16}.

\subsection{Syntax}

The syntax of commands (for a single thread) is defined by the
following grammar, where $\Exp$ and $\Comm$ define expressions and
commands, respectively. We assume that $\ominus$ is a unary operator
(e.g., $\neg$), $\otimes$ is a binary operator (e.g., $\land$,
$\lor$), $B$ is an expression (of type $\Exp$) that evaluates to a
boolean, $x$ is a variable (of type $\Var$) and $n$ is a value (of
type $\Val$).  \medskip

\begin{tabular}[t]{r@{~}l}
  $\Exp$   ::= & ${\sf Val} \mid \Exp ^{\sf A}  \mid \ominus Exp \mid Exp \otimes Exp$ \\[1mm]
  $\Comm$ ::= & $\kwskip \mid x.\kwswap(n)^{\sf RA} \mid x := \expr \mid x :=^{\sf R} \expr % \mid
                % x :=^{\sf A} \expr
                \mid$ \\
               & $\Comm ; \Comm \mid \kwif~B\ \kwthen\ \Comm\ \kwelse\ \Comm \mid $\\
               & $\kwwhile\ B\ \kwdo\ \Comm$  
\end{tabular} \medskip

Commands have their standard meanings.  % \hw{say that this type of
% semantics is inspired by a Vaf paper} \bd{added a line}
The only exceptions are the synchronising annotations, \emph{release}
${\sf R}$, \emph{acquire} ${\sf A}$ and \emph{release-acquire}
${\sf RA}$ (which we describe in detail below), and the command
\kwswap, which generates a read-modify-write update event, atomically
swapping the variable $x$ with value $n$. Note that (for simplicity),
we only present a release-acquire version of the \kwswap\ operation,
but leave in the {\sf RA} annotation for emphasis. Furthermore, we
assume unannotated accesses are \emph{relaxed}, i.e., data races do
not give rise to undefined behaviour; however it is straightforward to
extend the semantics to incorporate non-atomic accesses (which
potentially generate undefined behaviour).

\newcommand{\defn}{\mathrel{\widehat{=}}}

\newcommand{\flag}{\mathit{flag}}
\newcommand{\turn}{\mathit{turn}}

% \section{Peterson's algorithm}

\begin{example}[Peterson's algorithm]
  The running example for this paper will be the classic Peterson's
  mutual exclusion algorithm for two threads (see
  \refalg{alg:petersons-ra}) implemented using release-acquire
  annotations (this algorithm is taken from \cite{PetersonBlog}).  As
  with Peterson's original algorithm, variable $\flag_i$ is used to
  indicate whether thread $i$ intends to enter its critical section
  and a shared variable $\turn$ is used to ``give way'' when both
  threads intend to enter their critical sections at the same time.

  The difference in the C11 implementation is with the synchronisation
  annotations. The flag variable is set to $\True$
  (line~\ref{set-flag}) using relaxed atomics (which does not induce
  any synchronisation), but is set to $\False$ (line \ref{unset-flag})
  using a release annotation. The intention of the latter is to
  synchronise this write to $\flag$ with the read of $\flag$ at
  line~\ref{busy-wait} in the other thread.  The value of $\turn$ is
  set using a $\kwswap$ command, which induces release-acquire
  synchronisation. % This means that all executions of $\kwswap$ are totally
% ordered. 
  Note that the read of $\turn$ within the guard of the busy wait loop
  (line~\ref{busy-wait}) is relaxed. However, as we shall see, the
  algorithm still satisfies the mutual exclusion property.

\begin{algorithm}[!t]
  \caption{Peterson's algorithm with release-acquire}\label{alg:petersons-ra}
  \small
  \raggedright $\textbf{Init:}\ \ \flag_1 = \False \land \flag_2 = \False \land \turn = 1$\\
  
\algrenewcommand\algorithmicindent{0.75em} 
\begin{varwidth}[t]{0.5\columnwidth}
  \begin{algorithmic}[1] \small
    \Thread{1}
    % \While {true}
    \State $\flag_1$ := $\True$ ; \label{set-flag}
    \State $\turn$.\kwswap(2)$^{\sf RA}$ ; \label{swap-turn} 
    \MyWhile {
      \begin{tabular}[t]{@{}l@{}}
        ($\flag_2 \!=\! \True)^{\sf A}$\\
        ${} \land \turn$ = 2
      \end{tabular}
    } \label{busy-wait}
    \Statex \qquad 
    \textsf{\textbf{do}} \Skip
    \EndMyWhile 
    \State Critical section ; \label{critical-section}
    \State $\flag_1$ :=$^{\sf R}$ \False ; \label{unset-flag}
    % \EndWhile
    \EndThread
  \end{algorithmic}
\end{varwidth}
\quad
\begin{varwidth}[t]{0.5\columnwidth}
  \begin{algorithmic}[1] \small
    \Thread{2}
    % \While {true}
    \State $\flag_2$ := \True ; 
    \State $\turn$.\kwswap(1)$^{\sf RA}$ ; 
    \MyWhile {
      \begin{tabular}[t]{@{}l@{}}
        ($\flag_1 \!= \!\True)^{\sf A}$\\
        {} ${}\land \turn$ = 1 
      \end{tabular}
    }
    \Statex \qquad 
    \textsf{\textbf{do}} \Skip
    \EndMyWhile
    \State Critical section ;
    \State $\flag_2$ :=$^{\sf R}$ \False ;
    % \EndWhile
    \EndThread
  \end{algorithmic}
\end{varwidth}
\end{algorithm}

\end{example}

\subsection{Uninterpreted semantics}
\label{sec:unint-semant}
The uninterpreted operational semantics of commands is given by a
relation
$\mathop{\longrightarrow} \subseteq \Comm \times \Act_\tau
\times \Comm$, where
\[
  \Act = \bigcup_{x \in \Var; m, n \in \Val} \{
  \begin{array}[t]{@{}l@{}}
    rd(x, n), rd^{\sf A}(x,
    n), wr(x, n), \quad \\
    \hfill wr^{\sf R}(x, n), upd^{\sf RA}(x, m, n)\}
  \end{array}
\]
$\tau \notin \Act$ is a \emph{silent action} and
$\Act_\tau = \Act \cup \{\tau\}$.  We write $C \whilestep{a} C'$ for
$(C, a, C') \in \mathop{\longrightarrow}$.

\begin{figure*}[t]
  \centering \small 
 
  $ 
  \inference{ 
    x \in \fv(E) \quad n \in \Val \\ a = 
    rd(x, n) }
  {eval(E,a,E[n/x])}
  $
  \ \
  $
  \inference{ 
    x \in \fv(E) \quad n \in \Val  
    \\ a = 
    rd^{\sf A}(x, n)}
  {eval(E^{\sf A},a,E[n/x])}
  $ \medskip
  \ \ 
  $
  \inference{\fv(E) \neq \emptyset \\ eval(E,a,E')}
    {eval(\ominus E, a, \ominus E')}
  $   \ \ 
  $
  \inference{
    \fv(E_1)  \neq \emptyset 
    \\  eval(E_1,a,E_1')}
   {eval(E_1 \otimes E_2, a, E_1' \otimes E_2)}
  $
  \ \ 
  $
  \inference{
    \fv(E_1)  = \emptyset 
    \\  eval(E_2,a,E_2')}
   {eval(E_1 \otimes E_2, a, E_1 \otimes E_2')}
  $ \vspace{-3mm}

  \caption{Expression evaluation}
  \label{fig:expr-eval}

\end{figure*}  
An expression evaluation step is formalised by a relation
$eval(E, a, E')$, where $E, E'$ are expressions and $a$ is a
\emph{read action} that is generated by the evaluation step (see
\reffig{fig:expr-eval}). We assume $\fv(E)$ returns the set of free
variables in $E$. Note that $eval(E, a, E')$ is only defined when
\mbox{$\fv(E) \neq \emptyset$}. Moreover, in the presence of a binary
operator, expression evaluation is assumed to take place from left to
right. The notation $E[n/x]$ stands for expression $E$ with variable $x$ replaced by value $n$. 

\begin{figure*}[t]
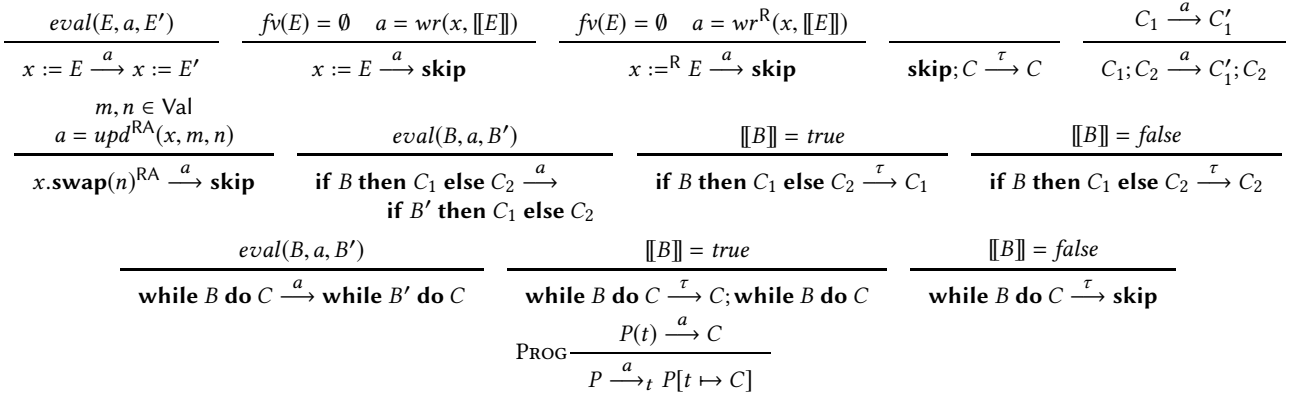

  \centering \small

  $ \inference{eval(E, a, E') }{x := E \whilestep{a} x:= E' } $ 
  \ \ 
  $ \inference{ \fv(E) = \emptyset \quad a = wr(x,[\![E]\!])}{x := E
    \whilestep{a} \kwskip } $
  \ \ 
  $ \inference{ \fv(E) = \emptyset \quad a = wr^{\sf R}(x,[\![E]\!])}{x
    :=^{\sf R} E \whilestep{a} \kwskip } $
  \ \ 
  $\inference{}{\kwskip ; C \whilestep{\tau} C}$
  \ \
  $\inference{C_1 \whilestep{a} C_1'}{C_1 ; C_2 \whilestep{a}
    C_1' ; C_2}$
  \medskip

  $\inference{m, n \in \Val \\ a = upd^{\sf RA}(x,m,n) }
  {x.\kwswap(n)^{\sf RA} \whilestep{a} \kwskip}$
  \ \ 
  $ \inference{eval(B,a,B')}{
    \begin{array}[t]{@{}l@{}}
      \kwif\ B\ \kwthen\ C_1\ \kwelse\ C_2
      \whilestep{a}\\
      \qquad\quad  \kwif\ B'\ \kwthen\ C_1\ \kwelse\ C_2
    \end{array}
  } $
  \ \ 
  $ \inference
  {[\![B]\!] = \True}
  {\kwif \ B\ \kwthen\
    C_1\ \kwelse\ C_2 \whilestep{\tau} C_1} $ 
  \ \
  $ \inference
  {[\![B]\!] = \False}
  {\kwif\ B\ \kwthen\
    C_1\ 
    \kwelse\ C_2 \whilestep{\tau} C_2} $

  \medskip
  
  $ \inference{eval(B,a,B')}{
    \begin{array}[t]{@{}l@{}}
      \kwwhile\ B\ \kwdo\ C
      \whilestep{a}  \kwwhile\ B'\ \kwdo\ C
    \end{array}
  } $
  \ \ 
  $ \inference{[\![B]\!] = \True}
  {
    \begin{array}[t]{@{}l@{}}
      \kwwhile\ B\ \kwdo\ C  \whilestep{\tau}  C ; \kwwhile\ B\ \kwdo\ C
    \end{array}
  }$ 
  \ \
  $ \inference{[\![B]\!] = \False}
  {\kwwhile\ B\ \kwdo\ C  \whilestep{\tau}  \kwskip}$ 
  \ \ \ $\text{\sc Prog}
  \inference{P(t) \whilestep{a} C} {P \whilestep{a}_t P[t
    \mapsto C]}$
    \vspace{-3mm}
    
  \caption{Uninterpreted operational semantics of
    commands and programs}
  \label{fig:comm-sem}
\end{figure*}

The uninterpreted operational semantics for commands is given in
\reffig{fig:comm-sem}.  Again, most of these rules are
straightforward. We assume $[\![E]\!]$ denotes the value of
(variable-free) expression $E$. An assignment $x := E$ generates a
read action whenever $\fv(E) \neq \emptyset$ and a write action
whenever $\fv(E) = \emptyset$. A $\kwswap$ command generates an update
action, and guard evaluation either generates a read or a silent
action.

Note that the uninterpreted operational semantics allows any value to
be read. Thus, we have the following property.
\begin{proposition} For all $m, m' \in \Val$, if
  $C \xrightarrow{rd(x,m)} C'$, then $C \xrightarrow{rd(x,m')} C'$; if
  $C \xrightarrow{rd^{\sf A}(x,m)} C'$, then
  $C \xrightarrow{rd^{\sf A}(x,m')} C'$; and if
  $C \xrightarrow{upd^{\sf RA}(x,m,n)} C'$, then
  $C \xrightarrow{upd^{\sf RA}(x,m',n)} C'$.
\end{proposition}
% We interpret read actions and ensure they correspond to a C11 state in
% \refsec{sec:paderb-semant-c11}.

For simplicity, we assume concurrency at the top level only. We let $T$ be
the set of all threads and use function of type $\Prog: T \to \Comm$
to model a program comprising multiple threads. The uninterpreted
operational semantics of a program is given by a relation
$\mathop{\longrightarrow} \subseteq \Prog \times T \times \Act_\tau
\times \Prog$ (using overloading). As before, we write
$P \whilestep{a}_t P'$ for
$(P, t, a, P') \in \mathop{\longrightarrow}$. An evaluation step of a
program $P$ is given by the rule {\sc Prog} (\reffig{fig:comm-sem}), which relies on the
uninterpreted operational semantics of a command to generate an action
$a$ and command $C$ from the command $P(t)$. The program after taking
a transition is the program $P$ but with $t$ mapped to the new command
$C$.

% \noindent 
% \begin{proposition} Whenever $C \whilestep{a_1}_{t_1} C_1$ and 
% $C \whilestep{a_2}_{t_2} C_2$, then there exists a $C'$ such that $C_2 \whilestep{a_1}_{t_1} C'$ 
% and $C_1 \whilestep{a_2}_{t_2} C'$. 
% \end{proposition} 
% \begin{proof}
%   Trivial. If $t_1 = t_2$, then $C_1 = C_2$. If $t_1 \neq t_2$, the
%   property holds since the two commands are independent.
% \end{proof}

% \fbox{Better to state as:}
Since threads execute independently in the uninterpreted semantics, all actions
commute.
\begin{proposition}
\label{prop:program-commute}
  If $P \whilestep{a_1}_{t_1} P_1$ and
  $P_1 \whilestep{a_2}_{t_2} P'$ and $t_1 \neq t_2$, then there exists
  a $P_2$ such that $P \whilestep{a_2}_{t_2} P_2$ and
  $P_2 \whilestep{a_1}_{t_1} P'$.
\end{proposition} 
% \begin{proof}
%   Trivial. If $t_1 = t_2$, then $C_1 = C_2$. If $t_1 \neq t_2$, the
%   property holds since the two commands are independent.
% \end{proof}

% \medskip 
% \noindent Operational semantics for command language given in Figure \ref{fig:comm-sem} 
% based on the evaluation of expressions defined in Figure \ref{fig:expr-eval}. 
% Definition of events see other file. 

%%% Local Variables:
%%% mode: latex
%%% TeX-master: "ppopp-full"
%%% End:

% !TeX root = hvc.tex
\newlist{mylist}{enumerate*}{1}
\setlist[mylist]{label=(\arabic*)}

\renewcommand{\tag}{{\it tag}}

\newcommand{\Pad}{{\it RA}}

\section{An Operational Semantics for RAR C11}
\label{sec:paderb-semant-c11}

We now extend the semantics from \refsec{sec:comm-lang-unint} and
interpret read, write and update actions in the C11 memory
model. % Unlike existing semantics, which describe the behaviours of C11
% programs axiomatically, w
We develop an operational semantics that takes inspiration from the
axiomatic descriptions
\cite{DBLP:conf/popl/BattyDW16,DBLP:conf/popl/BattyOSSW11,DBLP:conf/pldi/LahavVKHD17}. In
\refsec{sec:soundn-compl}, we show that the operational model is in
fact equivalent to a reformulation (inspired by \cite{DBLP:conf/pldi/LahavVKHD17}) of the \emph{RAR fragment} of the RC11
semantics~\cite{DBLP:conf/popl/BattyDW16}. 

% The main idea behind our semantics is a notion of \emph{observability},
% derived from the orders in a C11 state, that permits each thread to
% have its own view of the global order of writes\footnote{Others have
%   presented similar ideas, e.g., \cite{DBLP:conf/popl/KangHLVD17} who
%   also record a thread-specific view of the memory in their semantics
%   and that of \cite{DBLP:conf/popl/LidburyD17} who introduce store
%   buffers on which reads determine the writes that they may read
%   from.}. To this end, we keep track of the set of writes already
% \emph{encountered} by each thread in an execution. A write $w$ is then
% no longer observable to a thread if there exists an encountered write
% that follows $w$ in \emph{modification order}, which is an order
% defined for a C11 state.

We formalise C11 states in \refsec{sec:c11-stat-observ} and define an
operational \emph{event semantics} based on observability
(\refsec{sec:observ-event-semant}). This event semantics in turn gives
rise to an \emph{interpreted semantics} (\refsec{sec:interpr-semant}).

\subsection{C11 States and Basic Orders}
\label{sec:c11-stat-observ}
% We work towards a formalisation of a state in C11. 
The formalisation in this section follows the existing literature on axiomatic C11 semantics
\cite{DBLP:conf/popl/BattyOSSW11,DBLP:conf/pldi/LahavVKHD17}. First we give some
preliminary definitions.  \smallskip

\noindent {\bf Notation.}  % The semantics of the command language has already been
% explained, so we next turn to the memory model and start with
% defining states.
For an action $a \in \Act$, we let $\mathit{var}(a)\in \Var$ be the
variable read (or written to), $\rdval(a) \in \Val$ be the value read
and $\wrval(a)\in \Val$ be the value written. We extend actions to
\emph{events} of type $\Events = G \times \Act_\tau \times T$, where
$G$ is the set of \emph{tags} used to uniquely identify events in an
execution.  For an event $(g, a, t)$, where $g$ is a tag, $a$ is an
action, and $t$ is a thread identifier, we define $\tag(e) = g$,
$\act(e) = a$, $\tid(e) = t$, and (using lifting)
$\loc(e) = \loc(\act(e))$, $\wrval(e) = \wrval(\act(e))$,
$\rdval(e) = \rdval(\act(e))$. For a relation
$R \subseteq \Events \times \Events$, we let $R_{|t}$ and $R_{|v}$ be
the restriction of $R$ to events of thread $t$, and variable $v$,
respectively.

% For a set of events $D$, we
We let $\URA$ denote the RMW update events, and
distinguish the sets $\WR \supseteq \URA$ (write release),
$\RA \supseteq \URA$ (read acquire), $\WX$ (write relaxed) and $\RX$
(read relaxed). Finally, we define $\R = \RA \cup \RX$ (all reads) and
$\W = \WR \cup \WX$ (all writes).

\begin{definition}
  \label{def:c11-state}
  A {\em C11 state} is a triple
  $\bbD = ((D, \ltsb), \ltrf, \ltmo)$ comprising a set of events $D$
  paired with a \emph{sequenced-before} relation
  $\ltsb \subseteq D \times D$, a \emph{reads-from} relation
  $\ltrf \subseteq \W \times \R$ and a \emph{modification order}
  $\ltmo \subseteq \W \times \W$.
\end{definition}
We let $\Sigma$ denote the set of all C11 states. The three relations
in a C11 state $((D, \ltsb), \ltrf, \ltmo)$ reflect different
relationships between operations.  The sequenced-before relation
$\ltsb$ records the program order within one thread; $\ltsb_{|t}$ is a
strict total order for each thread $t$.  The reads-from relation
$\ltrf$ provides the justification for the values being read: every
read must have a corresponding action that writes the value being
read.  The modification order $\ltmo$ describes an ordering of the
writes on variables; $\ltmo_{|v}$ is a strict total order for each
variable $v \in \Var$. % The axioms that formalise validity of a C11
% state are given in \refsec{sec:rc11-memory-model}.

%We
% pose no apriori restrictions on the relations.

% A C11 state $\bbD = (D, \ltsb, \ltrf, \ltmo)$ is a \emph{valid} iff
% \begin{mylist}
% \item 
% \item , and
% \item .
% \end{mylist}

Weak memory models are often defined in terms of a
\emph{happens-before} order (denoted $\lthb$), which formalises a
notion of causality. In C11, an event occuring in a thread before
another event in the same thread induces \emph{sequenced-before} order
(denoted $\ltsb$), which in turn induces happens before order.
Moreover, reads-from edges induce happens-before order when the
corresponding actions in the edge are \emph{synchronising actions}
(i.e., a release and an acquire). This is formalised by an additional
\emph{synchronises-with} relation (denoted $\ltsw$). Formally, we
define
\begin{gather*}
  \ltsw = \ltrf \cap (\WR \times \RA) \qquad \qquad \lthb = (\ltsb
  \cup \ltsw)^+\ .
\end{gather*}
As is standard in the literature, we assume all variables are
initialised by a special thread $0 \in T$.
Define the set of {\em initialising writes} to be $\IW = \{w \in \W \mid \tid(w) = 0\}$.
The initial states of our operational model are those of the form
$\sigma_0 = ((I, \emptyset), \emptyset, \emptyset)$ where
$I \subseteq \IW$, and for each variable $x$, there is exactly one write
$w \in I$ such that $\var(w) = x$. For a state $\sigma = ((D, \_), \_,\_)$,
let $I_{\sigma} = D \cap \IW$.

The relation $\ltfr = (\ltrf^{-1} ; \ltmo) \backslash {\it Id}$ (where
$;$ is relational composition) is the ``from-read''
relation\footnote{$\ltfr$ is also referred to as ``reads-before''
  \cite{DBLP:conf/pldi/LahavVKHD17}.} % \hw{what about calling fr
  % ``reads-before'' like Vaf17 does?} \bd{todo} 
that relates each read to all writes that are $\ltmo$-after the write
the read has read from. We must subtract ${\it Id}$ (identity) edges
from $\ltrf^{-1} ; \ltmo$ to cope with update events, which have the
potential to induce reflexivity in $\ltfr$
\cite{DBLP:conf/popl/BattyDW16,DBLP:conf/pldi/LahavVKHD17}.
\begin{example}
  \label{ex:observ-event-semant-0}
  An example C11 state is given below, where threads 1-4 have executed
  some actions.   Since the actions are unique, we
  elide the tags from each event, and we identify the thread id with
  the action itself, e.g., $wr_1(y,1)$ is the action $wr(y,1)$
  executed by thread~$1$. %\vspace{-4mm}
  \begin{center}
  \scalebox{0.85}
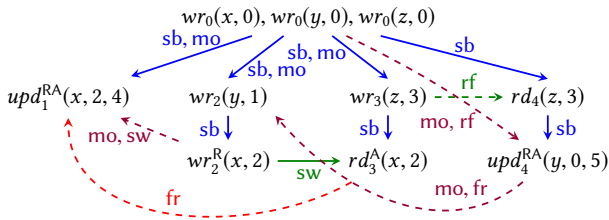
{
      \begin{tikzpicture}[node distance=.5cm]      
        \node (j) at (3.75,5.25) {$wr_0(x,0),wr_0(y,0),wr_0(z,0)$};
        \node (i) at (7.5,3) {$upd_4^{\sf RA}(y,0,5)$};
        \node (h) at (7.5,4) {$rd_4(z,3)$};
        \node (g) at (0,4) {$upd_1^{\sf RA}(x,2,4)$};
        \node (b) at (2.5,3) {$wr_2^{\sf R}(x,2)$};
        \node (a) at (2.5,4) {$wr_2(y,1)$};
        \node (d) at (5,3) {$rd_3^{\sf A}(x,2)$};
        \node (c) at (5,4) {$wr_3(z,3)$};
        \path 
        (a) edge[hb] node[left] {$\ltsb$} (b)
        (c) edge[hb] node[left] {$\ltsb$} (d)
        (h) edge[hb] node[right] {$\ltsb$} (i)
        (b) edge[sw] node[below] {$\ltsw$} (d)
        (c) edge[rf] node[above]{$\ltrf$} (h)
        (d) edge[fr,out=210,in=-90] node[above]{$\ltfr$} (g)
        (b) edge[mo] node[pos=0.1,left,xshift=-2mm] {$\ltmo, \ltsw$} (g)
        (i) edge[mo,out=-140,in=-50] node[pos=0.2,above,xshift=-2mm]
        {$\ltmo,\ltfr$} (a.south east)
        (j) edge[hb] node[pos=0.8,right] {$\ltsb, \ltmo$}   (a.north)
        (j) edge[hb] node[pos=0.2,right,xshift=5mm] {$\ltsb$}   (h.north)
        (j) edge[hb] node[pos=0.4,left] {$\ltsb,\ltmo$}   (c.north)
        (j) edge[hb] node[above] {$\ltsb, \ltmo$} (g)
        (j) edge[mo,bend left=5] node[pos=0.7,below,xshift=-3mm] {$\ltmo$, $\ltrf$} (i);
      \end{tikzpicture}
    } 
  \end{center} \vspace{-2mm}
    
  \noindent The initialising writes are $\ltsb$-before all thread actions, but
  are not ordered amongst themselves. Relation $\ltsb$ also describes
  the order for each thread. % For simplicity, we elide some of the
  % $\ltmo$, $\ltrf$ edges from the initialising writes. 
  Relation
  $\ltmo$ describes the order of modifications for each variable. The
  unsynchronised read $rd_4(z, 3)$ is justified by the $\ltrf$ from
  $wr_3(z, 3)$, whereas the synchronised read $rd_3^{\sf A}(x, 2)$ is
  justified by the $\ltsw$ from $wr_2^{\sf R}(x,2)$ and fixed before $upd_1^{\sf RA}(x,2,4)$ via the $\ltfr$ relation. Update events are
  related by both $\ltmo$ and $\ltrf$ to the immediately preceding write, and possibly related to later
  writes/updates by $\ltmo$ and $\ltfr$. If the write being read
   is releasing, then an update induces an $\ltsw$ (e.g., see $upd_1^{\sf RA}(x,2,4)$). \hfill\qed
\end{example}

In addition, our semantics uses the {\em extended coherence
  order}\footnote{The non-transitive version of this order is commonly
  referred to as the {\sf com} order
  \cite{AlglaveMT14}.}~\cite{DBLP:conf/pldi/LahavVKHD17}, denoted
$\lteco$, which is an order that fixes the order of reads and writes
to each variable (see \refex{ex:eco} below). Formally we define:
\[ \lteco = (\ltfr \cup \ltmo \cup \ltrf)^+\] 

% We need both the happens-before relation and the extended
% coherence order to determine the writes that a thread can read from.

\begin{example}
  \label{ex:eco}
  For executions of a C11 program, $\lteco$ over a single variable
  takes the following form, where $w_1,\dots, w_5$ are writes and
  $r_1$, $r_1'$ etc are reads and $u$ is an update.

  \begin{center} \scalebox{0.9}
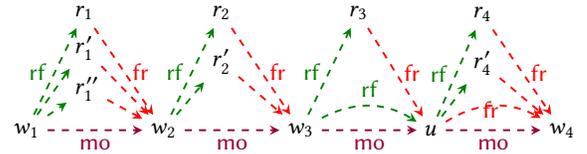
{
      \begin{tikzpicture}[node distance=1.4cm]      
        \node (w1) {$w_1$};
        \node[right=of w1] (w2) {$w_2$};
        \node[right=of w2] (w3) {$w_3$};
        \node[right=of w3] (u)    {$u$};
        \node[right=of u] (w4) {$w_4$};
        \node[above right=0.4cm of w1,yshift=1cm]   (r1)   {$r_1$};
        \node[above right=0.4cm of w1,yshift=0.4cm] (r1')  {$r_1'$};
        \node[above right=0.4cm of w1,yshift=-0.2cm](r1'') {$r_1''$};
        \node[above right=0.4cm of w2,yshift=1cm]   (r2)   {$r_2$};
        \node[above right=0.4cm of w2,yshift=0.2cm] (r2')  {$r_2'$};
        \node[above right=0.4cm of w3,yshift=1cm]   (r3)   {$r_3$};
        \node[above right=0.4cm of u,yshift=1cm]   (r4)   {$r_4$};
        \node[above right=0.4cm of u,yshift=0.2cm] (r4')  {$r_4'$};
        \path 
        (w1)   edge[mo] node[below] {$\ltmo$} (w2)
        (w2)   edge[mo] node[below] {$\ltmo$} (w3) 
        (u)   edge[mo] node[below] {$\ltmo$} (w4)
        
        (w1)   edge[rf] node[left] {$\ltrf$} (r1)
        (w1)   edge[rf] (r1')
        (w1)   edge[rf] (r1'')
        (r1)   edge[fr] node[right] {$\ltfr$} (w2)
        (r1')  edge[fr] (w2)
        (r1'') edge[fr] (w2)
        
        (w2)   edge[rf] node[left] {$\ltrf$} (r2)
        (w2)   edge[rf] (r2')
        (r2)   edge[fr] node[right] {$\ltfr$} (w3)
        (r2')  edge[fr] (w3)

        (w3)   edge[rf,bend left] node[above] {$\ltrf$}(u)
        (w3)   edge[mo] node[below] {$\ltmo$}(u)
        (u)    edge[fr,bend left] node[right,pos=0.3] {$\ltfr$} (w4)

        (w3)   edge[rf] node[left] {$\ltrf$} (r3)
        (r3)   edge[fr] node[right] {$\ltfr$} (u)
        
        (u)   edge[rf] node[left] {$\ltrf$} (r4)
        (u)   edge[rf] (r4')
        (r4)   edge[fr] node[right] {$\ltfr$} (w4)
        (r4')  edge[fr] (w4)
        ; 
      \end{tikzpicture} }
  \end{center}
  Reads $r_1$, $r_1'$ and $r_1''$ read from the write $w_1$, inducing
  from-read edges to $w_2$ (the write that immediately follows $w_1$
  in $\ltmo$). The update $u$ induces an $\ltrf$ from $w_3$ (the write
  event immediately before $u$ in $\ltmo$) and an $\ltfr$ to $w_4$
  (the write event immediately after $u$ in $\ltmo$). \hfill \qed
\end{example}

\subsection{Event Semantics and Observability}
\label{sec:observ-event-semant}

Recalling that $\Sigma$ denotes the set of all possible C11 states and
$\W$ is the set of all writes (including updates), each step of the
event semantics is formalised by the transition relation
$\mathord{\strans{\ \ }_{\Pad}} \subseteq \Sigma \times \W_\bot \times
\Events \times \Sigma$ (see \reffig{fig:c11-opsem}), where we have
$\W_\bot = \W \cup \{\bot\}$ and $\bot \notin \W$. Again, we write
$\sigma \strans{w, e}_{\Pad} \sigma'$ for
$(\sigma, w, e, \sigma') \in \mathord{\strans{\ \ }_{\Pad}}$.

For each rule $\sigma \strans{w, e}_{\Pad} \sigma'$, $w$ is the
\emph{write being observed} by the event $e$. Strictly speaking, the
event semantics could be defined without the $w$. However, making this
observed write explicit is useful for the verification
(\refsec{sect:verif}).

\begin{figure}[t]
  \centering
  \small
  $ \inference[{\sc Read}] {a \in \{rd(x, n),
    rd^\mathsf{A}(x, n) \}  \quad  \loc(w) = x \quad \wrval(w)=n   \\
    w \in \OW_\sigma(t) \quad \ltrf' = \ltrf \cup
    \{(w,e)\}  \quad \ltmo' = \ltmo } {((D,
    \ltsb),\ltrf,\ltmo) \strans{w,e}_{\Pad} ((D, \ltsb) + e,\ltrf',\ltmo')}$
  \bigskip

  $ \inference[{\sc Write}] {% e = (g, a, t) \quad g \notin tags(D) \quad
    a \in \{ wr(x,n),
    wr^{\sf R}(x,n)\} \quad w \in \OW_\sigma(t) \backslash
    \CW_\sigma \\ \loc(w) = x \quad  \ltrf' = \ltrf \quad
    \ltmo' = \ltmo[w,e] % loc(e)
  }{((D,
    \ltsb),\ltrf,\ltmo) \strans{w, e}_{\Pad} ((D, \ltsb) + e,\ltrf',\ltmo')}$ \bigskip

  $\inference[{\sc RMW}] {% e = (g, a, t) \quad g \notin tags(D) \quad
    a = upd^{\sf
      RA}(x,m,n) 
    \quad w \in \OW_\sigma(t) \backslash \CW_\sigma \quad \loc(w)
    = x \\ \wrval(w)=m \quad \ltrf' = \ltrf \cup \{(w,e)\} \quad
    \ltmo' = \ltmo[w,e]} {((D,
    \ltsb),\ltrf,\ltmo) \strans{w,e}_{\Pad} ((D, \ltsb) +
    e,\ltrf',\ltmo')}$

  \caption{Event semantics assuming
    $\sigma = ((D, \ltsb),\ltrf,\ltmo)$, $e = (g, a, t)$  and $g \notin tags(D)$}
  \label{fig:c11-opsem}
\end{figure}

We now describe each of the rules in \reffig{fig:c11-opsem}. Executing
each event $e$ updates $(D, \ltsb)$ to:
\begin{align*} 
  (D, \ltsb) + e
  & =
    \left(\begin{array}[c]{@{}l@{}}
      D \cup \{e\}, \\
      \ \ltsb \cup  (\{e' \in D  \mid \tid(e') \in \{\tid(e), 0\} \}
      \times \{e\} )
    \end{array}\right)
\end{align*} 
Thus, the initial writes are $\ltsb$-prior to every non-initialising
event.  Relations $\ltrf$ and $\ltmo$ are updated according to the
write events in $D$ that are observable to the thread executing the
given event. To this end, we must distinguish three sets of writes:
\emph{encountered writes} and \emph{observable writes}, which are
specific to each thread, and \emph{covered writes}, which are the set
of writes that are immediately followed, in reads-from order, by an
update event.

The set of \emph{encountered writes} are the writes that thread $t$ is
aware of (either directly or indirectly) in state
$\sigma=((D,\ltsb),\ltrf,\ltmo)$, and are given by:
\begin{align*}
  \observedWrites_\sigma(t) = \{w  \in \W \cap D \mid \exists e \in D.\
  \begin{array}[t]{@{}l@{}}
    \tid(e) = t \wedge {} \\
    (w,e) \in \lteco^?; \lthb^? \}
  \end{array}
\end{align*}
where $R^?$ is the reflexive closure of relation $R$.  Thus, for each
$w \in \observedWrites_\sigma(t)$, there must exist an event $e$ of
thread $t$ such that $w$ is either $\lteco$- or $\lthb$- or
$\lteco ; \lthb$-prior to $e$. Note that $\observedWrites_\sigma(t) = \emptyset$ if the thread $t$ has not executed any actions; as soon as the thread executes its first action, we have $I \subseteq \observedWrites_\sigma(t)$. 

From these, we determine the \emph{observable writes}, which are the
writes that thread $t$ can observe in its next read. These are defined
as:
\begin{align*}
  \OW_\sigma(t) = \{w  \in \W \cap D  \mid \forall w' \in \observedWrites_\sigma(t).\ (w,w') \notin \ltmo\}
\end{align*}
Thus, observable writes are not succeeded by any encountered write in
modification order, i.e., the thread has not seen another write
overwriting the value being read.

Finally, to guarantee {\em atomicity} of the update events, there
cannot be any write operations (in modification order) between the
write that an update reads from and the write of the update itself. We
therefore define the set of \emph{covered writes} as follows:
\begin{align*}
  \CW_\sigma = \{w \in \W \cap D \mid \exists u \in \URA.\ (w,u) \in \ltrf\}   
\end{align*}

\begin{example}
  \label{ex:observ-event-semant}
  Consider the C11 state $\sigma$ in \refex{ex:observ-event-semant-0}.
  Given that $I = \{wr_0(x,0),wr_0(y,0),wr_0(z,0)\}$ is the set of
  initialising writes, the encountered writes for each thread are as
  follows:
  \begin{align*}
    \observedWrites_\sigma(1)& = I \cup \{wr_2^{\sf R}(x,2), upd_1^{\sf RA}(x,2,4)\} \\
    \observedWrites_\sigma(2) &= I \cup \{
                                  wr_2(y,1),  
                                  wr_2^{\sf R}(x,2), % upd_1(x,2,4), 
                                  % wr_2(x,5), 
                                  upd_4^{\sf RA}(y,0,5)\}
\\  
    \observedWrites_\sigma(3)& = I \cup\{ wr_2(y,1), wr_2^{\sf R}(x,2), wr_3(z,3),upd_4^{\sf RA}(y,0,5)\} \\     
    \observedWrites_\sigma(4)& = I \cup\{ wr_3(z,3), upd_4^{\sf RA}(y,0,5)\} 
  \end{align*} 
  
  \noindent The observable writes are hence 
  \begin{align*}
    \OW_\sigma(1) &= \{
                    \begin{array}[t]{@{}l@{}}
                      wr_0(y,0), wr_0(z,0), 
                      wr_2(y,1),  wr_3(z,3), \\
                      upd_1^{\sf RA}(x,2,4), % wr_2(x,5), 
                      upd_4^{\sf RA}(y,0,5)\}
                    \end{array}\\
    \OW_\sigma(2) & = \{wr_0(z,0), wr_2(y,1), wr_3(z,3), upd_1^{\sf RA}(x,2,4)% , wr_2(x,5)
                    \} \\
    \OW_\sigma(3) & = \{wr_2(y,1),wr_2^{\sf R}(x,2), wr_3(z,3), upd_1^{\sf RA}(x,2,4)% , wr_2(x,5)
                    \} \\ 
    \OW_\sigma(4) & =\{
                    \begin{array}[t]{@{}l@{}}
                      wr_0(x,0), wr_2(y,1),wr_2^{\sf R}(x,2), wr_3(z,3), \\
                      upd_1^{\sf RA}(x,2,4), % wr_2(x,5),
                      upd_4^{\sf RA}(y,0,5)\}
                    \end{array}
  \end{align*}
  The covered writes are 
  $\CW_\sigma = \{wr_0(y,0), wr_2^{\sf R}(x,2)\}$.  \hfill \qed
\end{example}

Observable writes are used to resolve the read events in each
thread. Namely, a thread $t$ can read from any write event in
$\OW_\sigma(t)$. This is reflected in the {\sc Read} rule, where the
$\ltrf$ component is updated to record an $\ltrf$ from some observable
write $w$ to the read event $e$, provided $w$ writes to the variable
that $e$ reads and the value read matches the value written.

To explain the write and update semantics, we require some more formal
machinery.  The observable and covered writes together determine the
allowable updates to the $\ltmo$ relation after executing a write
event. Unlike SC, a write event to variable $x$ is not simply appended
to the end of $\ltmo_{|x}$. Instead we allow a thread $t$ that
performs a write $e$ (or update) to $x$ to insert $e$ after any
observable write $w$ in $\ltmo_{|x}$ that is not a covered write. This
condition is sufficient to ensure no cyclic dependencies arise as a
result of performing the write.

Given that $R[x]$ is the relational image of $x$ in $R$, we define
$R_{\Downarrow x} = \{x\} \cup R^{-1}[x]$ 
% \begin{align*}
%   R_{\Downarrow x} & = \{x\} \cup R^{-1}[x] % \{e' \mid (e', e) \in R\}
%   % & 
%   % R_{\uparrow x} & = R[x] % \{e' \mid (e, e') \in R\}
% \end{align*}
to be the set of all elements in $R$ that relate to $x$ (inclusive).
The insertion of a write event $e$ directly after a write $w$ in
$\ltmo$ is given by
\[ \textstyle\ltmo[w,e] =
  \begin{array}[t]{@{}l@{}}
    \ltmo  \cup %\ltmo_{\leq w} \cup \ltmo_{> w} \cup {} \\
    (\ltmo_{\Downarrow w} \times \{e\}) \cup (\{e\} \times
    \ltmo[w]) % _{\uparrow x})
  \end{array}
\]

The rules {\sc Write} and {\sc RMW} update $\ltmo$ in the same
way. For the write event $e$ executed by thread $t$, they pick some
$w$ that writes to the same variable as $e$, is observable to $t$ and
not covered by an update event, then insert $e$ immediately after $w$
in $\ltmo$.

% A read-modify-write event performs both a read and write in the same
% step.

\begin{example}
  For the execution in \refex{ex:observ-event-semant}, no thread may
  introduce a write between $wr_1^{\sf R}(x,2)$ and $upd_1^{\sf RA}(x, 4,5)$,
  or between $wr_0(y,0)$ and $upd_5^{\sf RA}(y,0,7)$.
\end{example}

\subsection{Interpreted Semantics}
\label{sec:interpr-semant}
We now combine the event semantics with the uninterpreted semantics to
give an interpreted semantics for the language in
\refsec{sec:comm-lang-unint} overall. We give two generic rules that
allows different memory models to be plugged in for the event
semantics.

To this end, we define a {\em configuration} to be a pair $(P,\sigma)$,
consisting of a program $P$ and a state $\sigma$ of the memory
model. The command part of a configuration triggers events that are
agnostic to values. However, the memory model will only allow certain
values in read events. This idea is captured by the following two
rules combining the uninterpreted program semantics (i.e., rule {\sc
  Prog}) from \refsec{sec:unint-semant} and an event semantics in some
memory model $M$:
\begin{center}
$
\inference{P \whilestep{\tau}_t P'}{(P,\sigma)
  \ltsArrow{\bot,\tau}_M (P',\sigma)}$ \medskip
\quad
$\inference{P
  \whilestep{a}_t P' \quad \sigma \strans{w,e}_M \sigma' \\
  \act(e) = a \quad \tid(e) = t} {(P,\sigma) \ltsArrow{w,e}_{M}
  (P',\sigma')}
$
\end{center}
The first rule describes a $\tau$-step and does not change the state.
The second states that a thread can execute action $a$ in the current
state $\sigma$ only if the event semantics of the memory model in
consideration permits it.

%%% Local Variables:
%%% mode: latex
%%% TeX-master: "ppopp-full"
%%% End:

% \subsection{A Peterson State}

% 

\begin{example} Consider the state of Peterson's algorithm
  (\refalg{alg:petersons-ra}) in RA C11 that results when thread 1 has
  reached the guard at line~\ref{busy-wait}, and thread 2 is about to
  execute line~\ref{swap-turn}. Execution of this step introduces the
  boxed event $upd_2^{\sf RA}(\turn,2,1)$. (We use the box for
  emphasis; it does not carry any special semantic meaning.)
  \begin{center} \scalebox{0.9}
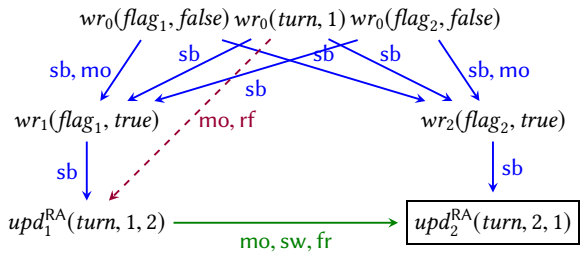
{
    \begin{tikzpicture}[node distance=3cm]      
      \node (i2)  {$wr_0(\turn,1)$};
      \node (i1) [left of=i2,xshift=1cm]{$wr_0(\flag_1,\False)$};
      \node (i3) [right of=i2,xshift=-1cm] {$wr_0(\flag_2,\False)$}; 
      \node (p0a) [below of=i1,yshift=1.5cm,xshift=-1cm] {$wr_1(\flag_1, \True)$};
      \node (p0b) [below of=p0a,yshift=1.5cm] {$upd_1^{\sf RA}(\turn,1,2)$}; 
      % \node (p0c) [below of=p0b,yshift=2.5cm] {$rd_1^{\sf A}(\flag_2,\False)$}; 
      \node (p1a) [below of=i3,yshift=1.5cm,xshift=1cm] {$wr_2(\flag_2,\True)$};
      \node (p1b) [below of=p1a,yshift=1.5cm] {\fbox{$upd_2^{\sf RA}(\turn,2,1)$}}; 
      %\node (p1c) [below of=p1b,yshift=2.5cm] {$rd_1(f0,1)$}; 
      \path (i1) edge[hb] node[left] {$\ltsb, \ltmo$} (p0a)
      (p0a) edge[hb] node[left] {$\ltsb$} (p0b)
      % (p0b) edge[hb] node[left] {$\ltsb$} (p0c)
      (p1a) edge[hb] node[right] {$\ltsb$} (p1b)
      (i1)  edge[hb] node[above] {$\ltsb$} (p1a)
      (i2)  edge[hb] node[above] {$\ltsb$} (p1a)
      (i2)  edge[hb] node[above] {$\ltsb$} (p0a)
      (i3)  edge[hb] node[right] {$\ltsb, \ltmo$} (p1a)
      (i3)  edge[hb] node[below] {$\ltsb$} (p0a)
      (i2) edge[mo] node[right] {$\ltmo, \ltrf$} (p0b)
      (p0b) edge[sw] node[below] {$\ltmo, \ltsw, \ltfr$}    (p1b); 
        % (i3) edge[rf] node[above] {$\ltrf$} (p0c)
        % (p0c) edge[fr] node[above] {$\ltfr$} (p1a); 
   \end{tikzpicture} }
\end{center}
In the state {\em without} the boxed event, thread 2 can read from 
$wr_0(\turn,1)$ via a read event, but it cannot do so via an update because $wr_0(\turn,1)$ is
covered by the existing update
$upd_1^{\sf RA}(\turn,1,2)$. Hence the update of thread~1 (when the
event in the box occurs) updates $turn$ from $2$ to $1$, which creates
% $\ltrf$,
$\ltmo$, $\ltsw$ and $\ltfr$ edges from $upd_1^{\sf RA}(\turn,1,2)$.

% \smallskip
% \noindent in the state {\em with} the event in the box, thread 1 has the following encountered writes
% \[ \observedWrites(1) = 
% \{ wr_i(f1,0), wr_1(f1,1), upd_0(turn,0,1), wr_i(turn,0), wr_0(f0,1), wr_i(f0,0)\} \] 
Now consider a continuation from the state with the boxed event, where
the threads read the values in their respective guards. Thread $2$ has
encountered $wr_1(\flag_1,\True)$, and hence, is no longer able to
observe $wr_0(\flag_1,\False)$. Similarly, since thread 2 has
encountered $upd_2^{\sf RA}(\turn,2,1)$ it is no longer able to
observe $wr_0(\turn, 1)$ or $upd_1^{\sf RA}(\turn,1,2)$. We therefore
conclude that thread $2$'s guard will evaluate to true, causing it to
spin at line~\ref{busy-wait}. 
In contrast, thread $1$ can read from either $wr_0(\flag_2, \False)$
or $wr_2(\flag_2, \True)$ since it has not yet encountered the event
$wr_2(\flag_2, \True)$. Similarly, since it has not yet encountered
$upd_2^{\sf RA}(\turn,2,1)$, it can read from both
$upd_1^{\sf RA}(\turn,1,2)$ and $upd_2^{\sf RA}(\turn,2,1)$. Thread 1
therefore could spin at line~\ref{busy-wait} or exit the busy
loop. Note that once it has read a new value for $\flag_2$ or $\turn$,
the previous value (in $\ltmo$-order) can no longer be read.
\end{example}

This example demonstrates how the basic synchronisation principle of
Peterson's algorithm is guaranteed by the release-acquire
annotations. Namely,
% \begin{itemize}
% \item
(1) the updates on $\turn$ are totally ordered via $\lthb$ due to the
release-acquire annotation on statement $\kwswap$, and (2) the thread
that is first to execute $\kwswap$, may miss to see that the other
thread has set its flag.

% % \bigskip
% \noindent A lemma we might want to have? this is basicallly the determinate value agreement lemma, but phrased without the 
% notation we use for the verification and might be useful to put in directly after definition the Paderborn semantics? 

% \begin{lemma}
%   Let $t_1,t_2 \in T, \sigma$ a C11-state. Then the following holds: 
% \begin{align*} 
%      \neg \exists w_1, w_2: \loc(w_1) & = \loc(w_2), w_1 \neq w_2: \\
%          & w_1 \in \OW_\sigma(t_1) \setminus \OW_\sigma(t_2) \wedge w_2 \in \OW_\sigma(t_2) \setminus \OW_\sigma(t_1) 
% \end{align*}
% \end{lemma}

% \noindent Corollary: if two threads both observe a single write to a location $x$, it is the same write

%%% Local Variables:
%%% mode: latex
%%% TeX-master: "ppopp-full"
%%% End:

% !TeX root = hvc.tex

\section{Validation of Operational Semantics}
\label{sec:valid-oper-semant}

We now justify our operational semantics by showing it to be sound and
complete with respect to an existing axiomatic version of the C11
memory model. There are several versions of the C11 axiomatic
semantics that might be regarded as both standard and complete
\cite{DBLP:conf/popl/BattyOSSW11, DBLP:conf/popl/BattyDW16,
  DBLP:conf/pldi/LahavVKHD17}.  Our semantics deals only with the
release, acquire and relaxed annotations on operations. We call this
the {\em RAR} fragment of C11.  The standard C11 semantics also
specifies the behaviour of operations carrying {\em sequentially
  consistent} and {\em non-atomic} annotations. We ignore these
annotations here.  Our semantics closely resembles the RAR fragment of
\cite{DBLP:conf/popl/BattyDW16} and \cite{DBLP:conf/pldi/LahavVKHD17}.
Like \cite{DBLP:conf/popl/BattyDW16}, we use the convention that
update operations are represented as a single event, rather than a
read/write pair. Like \cite{DBLP:conf/pldi/LahavVKHD17}, we adopt the
constraint that $\ltsb \cup \ltrf$ is acyclic, and make use of the
extended coherence order.\footnote{In the appendix,
  we prove that our axiomatic model is equivalent to a variant of the
  RAR fragment of \cite{DBLP:conf/popl/BattyDW16}. This proof is
  supported by a mechanisation in Memalloy
  \cite{DBLP:conf/popl/WickersonBSC17}, which shows our models is
  equivalent to the RAR fragment for models upto size 7.}
% well the ``corrected version'' for C11 of
% Lahav et al.~\cite{DBLP:conf/popl/KangHLVD17} called RC11. The
% corrections over the initial proposal for a C11 memory model of Batty
% et al.~\cite{DBLP:conf/popl/BattyOSSW11} proposed therein mainly
% concern the SC part of C11 which we altogether ignore here. 
% Still, we stick to the description given in
% \cite{DBLP:conf/popl/KangHLVD17} as it has introduced the extended
% coherence order, which turned out to be handy in defining our
% operational semantics\footnote{One additional difference between the
%   presentation here and in \cite{DBLP:conf/popl/KangHLVD17} is that we
%   assume RMW-operations to consists of a single event, whereas
%   \cite{DBLP:conf/popl/KangHLVD17} uses two events and imposes an
%   additional atomicity constraint on them.}.
The axiomatic semantics is given in
\refsec{sec:rc11-memory-model}. Soundness and completeness of the
memory model is presented in \refsec{sec:soundn-compl}.

\subsection{Background: RAR Fragment of RC11}
\label{sec:rc11-memory-model}

Axiomatic semantics start with {\em pre-executions}, which are
candidates for valid C11 executions. A number of axioms are used to
define which of these candidates are considered real executions.
Pre-executions only contain a set of events and program order (as
represented by the sequenced-before relation).  We call such a pair
$(D,\ltsb)$ a {\em pre-execution state}.  New events can be added to
pre-execution states using the $+$ operator in the same way as in
\reffig{fig:c11-opsem}. Thus, if
$(D,\ltsb)\strans{\bot, e}_\PreExec (D',\ltsb')$ then
$(D',\ltsb') = (D, \ltsb) + e$.  These pre-execution steps are
combined with the steps of a program as before, i.e., using the rules
in \refsec{sec:interpr-semant}, i.e., we replace $\strans{w,e}_{M}$ by
$\strans{\bot, e}_\PreExec$. Since the first event in
$\strans{\bot, e}_\PreExec$ is always $\bot$ (no write events are
observed), we write $\strans{\ e\ }_\PreExec$ for
$\strans{\bot, e}_\PreExec$.

\begin{proposition}
  \label{prop:PE-commute}
  If $\gamma \strans{e_1}_{\PreExec} \gamma_1$
  and $\gamma_1 \strans{e_2}_{\PreExec} \gamma'$ and
  $\tid(e_1) \neq \tid(e_2)$, then there exists a $\gamma_2$ s.t.
  $\gamma \strans{ e_2}_{\PreExec} \gamma_2$ and
  $\gamma_2 \strans{e_1}_{\PreExec} \gamma'$.
\end{proposition}

\noindent Once a candidate pre-execution $(D,\ltsb)$ is computed, it
is augmented with the relations $\ltrf$ and $\ltmo$.

% \noindent When all these conditions are met for some instantiations of
% the reads-from and modification orders, an execution is said to be
% valid.

\begin{definition} 
  \label{def:legal-execution}
  A C11 execution $((D, \ltsb), \ltrf, \ltmo)$ is {\em valid} iff
  each of the following axioms hold:
  
  \smallskip
  \noindent {\sc SB-Total.} Sequenced-before is a total order over the events
  of each (non-initialising) thread and orders all initialising writes
  before all other events.  Formally, for any $e, e' \in D$,
  $$
  \begin{array}[t]{@{}l@{}}
    ((e, e') \in \ltsb \imp \tid(e) = 0 \lor \tid(e) = \tid(e')) \land {} \\ 
    (\tid(e) = 0 \land \tid(e') \noteq 0 \imp (e,e') \in \ltsb) \land {} \\
    (\tid(e) \!\neq\! 0 \land \tid(e) \!=\! \tid(e') \land  e \!\neq\! e' \imp 
    (e, e') \in \ltsb \cup \ltsb^{-1})\ . \\%  \lor (e', e) \in \ltsb)  \\
  \end{array}
  $$
  
  \smallskip
  \noindent {\sc MO-Valid.} Modification order is a strict order on
  $\W \cap D$ consisting of a disjoint union of relations
  $\{\ltmo_{|x}\}_{x\in \Var}$ which are themselves total. That is,
  for any $w, w' \in \W \cap D$,
  $$
  \begin{array}[t]{@{}l@{}}
    ((w, w') \in \ltmo \imp \loc(e) = \loc(e')) \land {} \\
    (\tid(w) = 0 \land \tid(w') \neq 0 \land \loc(w) = \loc(w') \imp
    \quad \\
    \hfill (w,w') \in \ltmo) \land {} \\
     (\tid(w) \neq 0 \land \tid(w') \neq 0 \land {} \\
    \qquad \loc(w) = \loc(w') \land w \neq w'  
    \imp (w, w') \in \ltmo \cup \ltmo^{-1})\ . % \lor (w', w) \in \ltmo)
    \\
  \end{array}
  $$
  
  \smallskip
  \noindent {\sc RF-Complete.} Each read matches exactly one write in
  the execution, i.e., for every $e \in \R \cap D$ there is exactly
  one $w \in \W \cap D$ such that $(w, e) \in \ltrf$, and for every
  $(e, e') \in \ltrf$,
  \[
    e \in \W \land e' \in \R \land \loc(e) = \loc(e') \land \wrval(e) = \rdval(e').
  \]
  
  \smallskip \noindent {\sc No-Thin-Air.}
  The relation $\mathord\ltsb \cup \ltrf$ is acyclic. 
  
  \smallskip \noindent {\sc Coherence.}  The relations $\lthb; \lteco^?$
  and $\lteco$ are irreflexive.
\end{definition} 

\begin{definition}
  A pre-execution state $\gamma$ is \emph{justifiable} iff there exist
  relations $\ltrf$ and $\ltmo$ such that $(\gamma, \ltrf, \ltmo)$ is
  valid.
\end{definition}

\subsection{Soundness and Completeness}
\label{sec:soundn-compl}
Having defined a new operational semantics for C11, the next step is
now the comparison with the existing axiomatic semantics.  In the
following, we prove the before given operational and axiomatic
semantics to be equal.  We start by showing that the executions of the
operational semantics are all consistent.

\begin{theorem}% [Soundness]
  \label{thm:completeness}
  Let $\sigma=((D,\ltsb),\ltrf,\ltmo)$ be a C11 state reachable from
  $\sigma_0$ using relation\ \  $\strans{\ \ \ }_{\Pad}$. Then $\sigma$
  satisfies {\sc SB-Total}, {\sc MO-Valid}, {\sc RF-Complete}, {\sc
    No-Thin-Air} and {\sc Coherence}.
\end{theorem} 

We next show that all consistent executions of a program are reachable in our 
operational semantics. We do so in two steps. 
First, we consider the runs of a program on the memory model. 
Since the axiomatic semantics in its pre-execution allows for reads before the appropriate 
writes, not every sequence of events possible for pre-executions is also possible in the 
operational semantics. 

\begin{example}
  \label{ex:reordering}
Consider the following simple program with two threads.
\[ \kwthread\ 1\!\!: z:=x\ \qquad \qquad \kwthread\ 2\!\!: x:= 5 \]

\noindent We have mapping $P_0 = \{1 \mapsto z:=x, 2 \mapsto x:=5\}$.
The following pre-execution is possible:
\[ \delta_0 \xRightarrow{rd_1(x,5)}_\PreExec \delta_1
  \xRightarrow{wr_1(z,5)}_\PreExec \delta_2
  \xRightarrow{wr_2(x,5)}_\PreExec \delta_3
\]
where $\delta_i = (P_i, \gamma_i)$. The pre-execution state $\delta_3$
can be justified using the following C11 state

\begin{center} 
\begin{tikzpicture}[node distance=.5cm]      
      \node (a) at (0,1) {$wr_2(x,5)$};
      \node (b) at (2,1) {$rd_1(x,5)$};
      \node (c) at (4,1) {$wr_1(z,5)$}; 
      \path 
      (a) edge[rf] node[above] {$\ltrf$} (b)
      (b) edge[hb] node[above] {$\ltsb$} (c); 
\end{tikzpicture} 
\end{center} 

\noindent The sequence of events is however not possible in the
$\strans{\ \ \ }_{\Pad}$ semantics since we cannot have a read without
the prior write that it reads from, and hence the first transition
cannot be emulated.  Still, the operational semantics can reach the same
final C11 state by executing
\[
  \begin{array}[t]{@{}l@{}}
    (P_0, \sigma_0) \xRightarrow{wr_2(x,5)}_{\Pad}  (P_1',\sigma_1') \xRightarrow{rd_1(x,5)}_{\Pad} \qquad \qquad   \\
    \hfill (P_2', \sigma_2') \xRightarrow{wr_1(z,5)}_{\Pad}
    (P_3,\sigma_3)
  \end{array}
\] 
\noindent which is also a sequence of steps in $\xRightarrow{\ \ \ }_\PreExec$. \hfill $\Box$
\end{example} 

\noindent The ``reordering'' of events described in
\refex{ex:reordering} is always possible: for every sequence of steps
of pre-executions, we can find a corresponding permutation of these
steps in which reads are ordered after their writes (and the program
order within threads is preserved).

Putting together Propositions~\ref{prop:program-commute} and
\ref{prop:PE-commute}, we have the following result. % For any partial
% order $<$, we say $\prec$ is a linearization of $<$ iff $\prec$ is a
% total order and $\mathord{<} \subseteq \mathord{\prec}$. A sequence
% $E = e_1 \dots e_n$ \emph{corresponds} to a total order $\prec$ if
% $i < j \Leftrightarrow e_i \prec e_j$. Moreover, $E$ is a
% \emph{linearization} of partial order $<$ iff $E$ corresponds to a
% linearization $\prec$ of $<$.
\begin{proposition}
  \label{prop:program-PE-commute}
  If $(P, \gamma) \ltsArrow{e_1}_{\PreExec} (P_1, \gamma_1)$ and
  $(P_1, \gamma_1) \ltsArrow{e_2}_{\PreExec} (P', \gamma')$
  where $\tid(e_1) \neq \tid(e_2)$, then there exists a program $P_2$
  and a pre-execution state $\gamma_2$ such that
  $(P, \gamma) \ltsArrow{ e_2}_{\PreExec} (P_2, \gamma_2)$ and
  $(P_2, \gamma_2) \ltsArrow{ e_1}_{\PreExec} (P', \gamma')$.
\end{proposition}
% Given a sequence $E = e_1, \dots, e_n$ that linearizes a strict order
% $\prec$, we say a $F = f_1, \dots, f_n$ is a \emph{permutation of $E$
%   that preserves $\prec$} iff $F$ is a permutation of $E$ and there is
% no pair of elements $f_i \prec f_j$ with $j < i$. In other words, $F$
% is also a linearization of $\prec$.

This proposition is used to prove a permutation theorem for
independent elements. We say that sequence $e_1 e_2 \dots e_n$ is a
\emph{linearization} of a strict order $\prec$ iff
$\dom(\prec) \cup \ran(\prec) = \{e_1, e_2, \dots, e_n\}$ and for any $e_i$, $e_j$,
we have $e_i \prec e_j \imp i < j$.

\begin{lemma}
  \label{lem:permutation}
  Let
  $Q = (P_0,\gamma_0) \xRightarrow{e_1}_\PreExec (P_1,\gamma_1)
  \xRightarrow{e_2}_\PreExec \ldots \xRightarrow{e_k}_\PreExec
  (P_k,\gamma_k)$ and $\gamma_k = (D_k, \ltsb_k)$. Then for
  every % permutation
  linearization $f_1, \ldots , f_k$ of % $e_1, \ldots, e_k$ that
  % preserves
  $\ltsb_k$, there exist programs $P_1', \ldots, P_{n-1}'$
  and pre-execution states $\gamma_1', \ldots, \gamma_{n-1}'$ such
  that
  \[
    (P_0,\gamma_0) \xRightarrow{f_1}_\PreExec (P_1',
    \gamma_1') \xRightarrow{f_2}_\PreExec \ldots
    \xRightarrow{f_k}_\PreExec (P_k,\gamma_k)\ .
  \]
\end{lemma}
% \begin{proof}[Proof sketch.]  Let $F = f_1, \dots, f_k$ and
%   $E = e_1, \dots, e_k$, and let $\delta_i$ represent the pair
%   $(P_i, \gamma_i)$. Suppose $f_1$, the first element of $F$ is the
%   element $e_i$, the $i$th element of $E$. We show that $Q$ can be
%   transformed into a valid sequence of $\PreExec$ steps such that
%   $e_i$ is the first event considered. This is possible since we have
%   the property
%   $\forall j.\ 0 \le j \le i-1 \imp \tid(e_j) \neq tid(e_i)$; allowing
%   $e_i$ to percolate to the first position via repeated application of
%   \refprop{prop:program-PE-commute}. We then repeat the process for
%   $f_2$, for $f_3$ and so forth.
% \end{proof}

We now show that for every justifiable pre-execution there is an
execution of the C11 semantics that ends in the C11 state justifying
the pre-execution.  The theorem uses a notion that restricts
pre-executions and C11 executions to a set of events. For a set of
events $E \subseteq D$, we define:
\begin{align*}
  (D, \ltsb)_{\downarrow E} & = (E, \ltsb \cap (E \times E)) \\
  (\gamma,\ltrf,\ltmo)_{\downarrow E} & = (\gamma_{\downarrow E}, \ltrf \cap (E \times E), \ltmo \cap (E \times E))
\end{align*}
In the completeness proof, we assume that the given pre-execution
sequence 
$(P_0,\gamma_0) \xRightarrow{e_1}_\PreExec (P_1,\gamma_1)
\xRightarrow{e_2}_\PreExec \ldots \xRightarrow{e_k}_\PreExec
(P_k,\gamma_k)$ has been reordered such that $e_1 \dots e_k$ is a
linearization of $\ltsb_k \cup \ltrf_k$, where $\ltrf_k$ is the
reads-from relation used in the justification of $\gamma_k$. Such a
linearization is possible since $\ltsb_k \cup \ltrf_k$ is acyclic (axiom {\sc No-Thin-Air}). 

% The main theorem relies on the fact that for any $\sigma_k$ that
% justifies $\gamma_k$, restricting $\sigma_k$ to $\{e_1, \dots e_n\}$
% results in a C11-state that justifies
% $\gamma_k \downarrow \{e_1, \dots e_n\}$.
% \begin{lemma} 
%   Suppose
%   $(P_0,\gamma_0) \xRightarrow{e_1}_\PreExec (P_1,\gamma_1)
%   \xRightarrow{e_2}_\PreExec \ldots \xRightarrow{e_k}_\PreExec
%   (P_k,\gamma_k)$ such that $\gamma_k = (D_k, \ltsb_k)$ is justifiable with
%   justification $\sigma_k= (\gamma_k,\ltrf_k,\ltmo_k)$ and
%   $e_1, \dots, e_k$ is a linearization of $\ltsb_k \cup \ltrf_k$. Then, for
%   $0 \le i \le k$, each $\gamma_k \downarrow \{e_1, \ldots, e_i \}$ is
%   justifiable with justification
%   $\sigma_k \downarrow \{e_1, \ldots, e_i \}$.
% \end{lemma} 

\begin{theorem}
  \label{thm:soundness}
  Suppose
  $(P_0,\gamma_0) \xRightarrow{e_1}_\PreExec (P_1,\gamma_1)
  \xRightarrow{e_2}_\PreExec \ldots \xRightarrow{e_k}_\PreExec
  (P_k,\gamma_k)$ such that $\gamma_k = (D_k, \ltsb_k)$ is justifiable
  with justification $\sigma_k= (\gamma_k,\ltrf_k,\ltmo_k)$ and
  $e_1, \dots, e_k$ is a linearization of $\ltsb_k \cup \ltrf_k$. Then
 $ (P_0,\sigma_0) \xRightarrow{e_{1}}_{\Pad} (P_1, \sigma_1)
    \xRightarrow{e_{2}}_{\Pad} \ldots \xRightarrow{e_k}_{\Pad}
    (P_k,\sigma_k)$, where
  $\sigma_i = (\gamma_k,\ltrf_k,\ltmo_k)_{\downarrow \{e_1, \ldots, e_i
  \}}$, $0 < i < k$.
\end{theorem} 

\newcommand{\clib}{{\cal P}}
\newcommand{\lrestr}{\downharpoonright}
\newcommand{\hbc}{\mathit{hbc}}
\newcommand{\hbo}[1]{\stackrel{#1}{\rightarrow}}
\newcommand{\detval}[1]{\stackrel{#1}{=}}
\newcommand{\last}{\mathit{last}}
\newcommand{\SW}{\mathit{SW}}
\newcommand{\covered}{\mathit{covered}}

\newcommand{\OV}{\mathit{OV\_replace}}

\section{Verification}
\label{sect:verif}

We now describe our verification method
(\refsec{sec:verification-method}), building on the operational
semantics. In \refsec{sect:case-study}, we apply it to our case study,
Peterson's mutual exclusion algorithm.

\subsection{Verification Method}
\label{sec:verification-method}
Our verification method is built around two kinds of assertions for
describing states of the operational semantics. The first kind, {\em
  determinate-value} assertions, are used to describe the values that
a read operation might return. As such, these assertions are analogous
to equations that specify the values of variables in a conventional
(i.e., sequentially consistent) setting in which the state of an
algorithm can be represented as a store that maps variables to
values. The second kind of assertion, {\em variable-ordering}
assertions, has no direct analogue in the conventional
setting. Variable-ordering assertions provide a way to describe how
information about a variable propagates between
threads.  % We describe

\smallskip 
\noindent {\bf Determinate-values.}
In the following, we assume that $\sigma = ((D, \ltsb), \ltrf, \ltmo)$
is a valid C11 state. We let $\sigma.\last(x)$ be the write or update
to $x$ in $D$ that is not succeeded by another write or update in
$\ltmo_{|x}$. Note that $\sigma.\last(x)$ is well-defined in any valid
state $\sigma$. % , because $\ltmox$ is a total and nonempty order.
When $X$ is a set of operations and $x$ is a variable,
$X_{|x} = \{e \in X \mid \loc(e) = x\}$.  For the determinate value
assertions, consider some thread $t$ and variable $x$.  In some
states of the operational semantics, there is exactly {\em one} write
that $t$ can read-from when reading $x$. This is true precisely when
$\OW_\sigma(t)_{|x} = \{\sigma.\last(x)\}$ (recall that
$\sigma.\last(x)$ is never covered, and so $\sigma.\last(x)$ can
always be observed in a transition).
% Determinate-value assertions imply conditions of the form of
% $\OW_\sigma(t)_{|x} = \{\sigma.\last(x)\}$.
Under such a condition, the value returned by a read of $x$ in thread
$t$ must be $\wrval(\sigma.\last(x))$. This ultimately provides us
with a weak memory analogue of an equation asserting that a given
variable has a given value in a conventional sequentially consistent
setting. 

\begin{definition}
\label{def:det-val}
  Let $t$ be a thread, $\sigma$ a state and $v$ a
  value.  The \emph{determinate-value} assertion
  $x \detval{\sigma}_t v$ holds iff
  \begin{align}
    \label{dva:1}
    \OW_\sigma(t)_{|x} & = \{\sigma.\last(x)\} \\
    \label{dva:2}
    v & =\wrval(\sigma.\last(x))\\
    \label{dva:3}
    \sigma.\last(x) & \in I_{\sigma} \cup \{e \mid \exists e'.\  \tid(e) = t \wedge (e, e') \in \sigma.\lthb^?\} 
  \end{align}
\end{definition}
\noindent  Condition \refeq{dva:3} states that $\sigma.\last(x)$ is either an
operation of the initialising thread, an operation of $t$, or
happens-before an operation of $t$.

\begin{example}
  \label{ex:det-val-assertion}
  To illustrate the determinate-value assertion, consider the two
  states below. In each case, assume there are writes (not shown) to
  variable $x$ that are $\ltmo$-prior to the write to $x$. Also assume
  that each write is the last write in $\ltmo$ order.

  \noindent
  \scalebox{0.85}
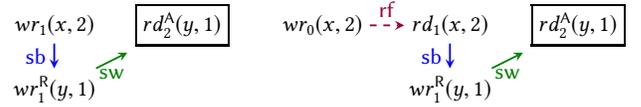
{
     \begin{tikzpicture}    
       \node (a) at (0,3) {$wr_1(x,2)$};  
       \node (b) at (0,2) {$wr_1^{\sf R}(y,1)$};
       \node (c) at (2,3) {\fbox{$rd_2^{\sf A}(y,1)$}};  
       \path 
        (a) edge[hb] node[left] {$\ltsb$} (b)
        (b) edge[sw] node[below] {$\ltsw$} (c); 
     \end{tikzpicture}}
        % \caption{Fig 1}
    %     \label{fig:1-1}
    % % \end{subfigure}
    \quad
    % \begin{subfigure}[b]{0.5\linewidth}
        %\centering
             \scalebox{0.85}{
     \begin{tikzpicture}
       \node (a) at (1,3) {$wr_0(x,2)$};  
       \node (b) at (3,3) {$rd_1(x,2)$};
       \node (c) at (3,2) {$wr_1^{\sf R}(y,1)$};
       \node (d) at (5,3) {\fbox{$rd_2^{\sf A}(y,1)$}};  
       \path 
        (a) edge[mo] node[above] {$\ltrf$} (b)
        (b) edge[hb] node[left] {$\ltsb$} (c)
        (c) edge[sw] node[below] {$\ltsw$} (d);

       % \node (a) at (0,3.5) {$wr_1(y,1)$};
       % \node (b) at (1,2.5) {\fbox{$wr_2(y,2)$}};
       % \node (c) at (2,5) {$wr_0(x,0)$};
       % \node (d) at (3,4) {$upd_3^{\sf RA}(x,0,1)$};
       % \path
       %   (a) edge[mo] node[left] {$\ltmo$} (b)
       %   (c) edge[rf] node[left] {$\ltmo$} node[right] {$\ \ltrf,\lthb$} (d)
       %   (c) edge[hb] node[left] {$\lthb$} (a) 
       %   (c) edge[hb] node[left] {$\lthb$} (b); 
   \end{tikzpicture}}
        % \caption{Fig 2}
    %     \label{fig:1-2}
    % \end{subfigure}
    % \caption{}
    % \label{fig:ex-for-assertionsB} 

 \noindent
 For the state on the left, after the boxed operation, thread 2
 satisfies \mbox{$x \detval{\sigma}_2 2$}, but for the state on the
 right, thread 2 does not. In each case, the only write to $x$ that
 thread 2 can observe is the illustrated write to $x$, but thread 2
 satisfies a corresponding determinate value assertion only on the
 left state. This is because on the left we have
 $(wr_1(x,2), rd_2(x,2)) \in \lthb$, but the unsychronised $\ltrf$
 edge on the right means that there is no analagous $\lthb$ edge.
\end{example}
% For example, 
% $\sigma$ 
% in the left of Figure \ref{fig:ex-for-assertions} without the boxed event.
% In $\sigma$, we have $x \detval{\sigma}_1 2$ 
% and $y \detval{\sigma}_1 1$, but notdeterminate value assertion for thread 2 
% (as $\OW_\sigma(2)$ contains two write events for $x$ and for $y$). 
In our verification, determinate-value assertions
support clean interaction with variable-ordering assertions,
which we describe shortly. Note that because our operational model
prevents update operations reading from covered writes (see Section
\ref{sec:paderb-semant-c11}), i.e., are more restricted than read
operations, an update operation on a variable $x$ may only be able to
read from the last write to $x$ even if $x \detval{\sigma} v$ is false
for all $v$. Below, we show how to handle important instances of this
situation.
% Note that in the case when the last write to $x$ is an update (with value $v$),
% a thread $t$ might only be able to read from this update
% and still $x \detval{\sigma}_t v$ does not hold (due to condition (a)).

% Note that in the case when the last write to $x$ is an
% update to value $v$, a thread $t$ might only be able to read from this
% update and still $x \detval{\sigma}_t v$ does not hold (due to
% condition \refeq{dva:1}).
%Note also that the last condition (c) enables a clean relationship
%between variable ordering and determinate value assertions.

The next two lemmas are immediate from the definition
of~$\detval{\sigma}_t$. \reflem{lem:det-val-read} below ensures that
the value returned by a reading transition using the semantics in
\reffig{fig:c11-opsem} is consistent with the determinate-value
assertion. \reflem{lem:ow-agreement} ensures that when a
determinate-value assertion holds for two threads reading from the
same variable, they return the same values for the variable.
\begin{lemma}[Determinate-Value Read]
\label{lem:det-val-read}
For any {\sc Read} or {\sc RMW} transition
$(P,\sigma) \ltsArrow{m, e}_{\Pad} (P',\sigma')$, if
$\loc(e) \detval{\sigma}_{\tid(e)} v$, then $\rdval(e) = v$.
\end{lemma}

\begin{lemma}[Determinate-Value Agreement]
\label{lem:ow-agreement}
For threads $t, t'$, variable $x$, and values $v, v'$,
if $x \detval{\sigma}_{t} v$ and $x \detval{\sigma}_{t'} v'$
then $v = v'$, and thus $t$ and $t'$ agree on the value of $x$.
\end{lemma}

\noindent 
Determinate-value assertions differ from their conventional
counterparts in that they are {\em relative to a particular thread}.
It is almost definitive of weak-memory systems that distinct threads
can have different views of the memory state.

\smallskip 
\noindent {\bf Variable-ordering.}
% Determinate-value assertions describe the state of shared memory. 
How can we ensure that distinct threads can agree on (or share)
sufficient determinate-value assertions to support a verification? We
address this problem by using another class of assertion: {\em
  variable-order} assertions, which orders two variables whenever the
last writes to the variables are causally (i.e., $\lthb$) ordered.

\begin{definition}
  Let $x, y$ be variables and $\sigma$ a state with $\lthb$ relation
  $\sigma.\lthb$.  The \emph{variable-order} assertion
  $x \hbo{\sigma} y$ holds iff $ (\sigma.\last(x), \sigma.\last(y)) \in \sigma.\lthb$ . 
\end{definition}

\noindent For example, the state $\sigma$ in the left of
Example~\ref{ex:det-val-assertion} without the boxed event satisfies
$x \hbo{\sigma} y$. When $x \hbo{\sigma} y$, a determinate-value
assertion $x \detval{\sigma}_t v$ can be ``copied'' to another thread
$t'$, whenever $t'$ performs an \emph{acquiring} read that reads-from
the last modification of $y$ and this write is \emph{releasing}.  It
is easy to see that in a state $\sigma'$ after such a synchronisation,
$\sigma'.\last(x)$ is happens-before an operation of $t'$, and thus
$x \detval{\sigma'}_{t'} v$.

\smallskip 
\noindent {\bf Inference rules.} Figure \ref{fig:det-val-ind} presents
a set of rules that precisely captures reasoning principles for
determinate-value and variable-order assertions. The ``copying'' of
determinate value assertions is captured in rule {\sf
  Transfer}\footnote{We show soundness of these proof rules in the
  appendix.}.  For the left state in
Example~\ref{ex:det-val-assertion} we can see this copying: when the
boxed event $rd_2^{\sf A}(y,1)$ occurs (leading to state $\sigma'$),
the determinate value assertion $x \detval{\sigma}_1 2$ is ``copied''
to thread~2 giving $x \detval{\sigma}_2 2$ by rule {\sf Transfer}.
Rule {\sf WOrd} shows how we introduce variable ordering assertions: a
variable ordering assertion can be introduced every time a thread
writes to one variable ($y$ in the rule), while having a determinate
value assertion on another variable ($x$ in the rule). Note that this
rule would not be sound, without Condition \refeq{dva:3} of Definition
\ref{def:det-val}: since the existence of an $\lthb$ edge from
$\sigma'.\last(x)$ to $\sigma'.\last(y)$.

\begin{figure*}[t]
  \centering
  $\inference[{\sc Init}]
  {\sigma_0 = ((I, \emptyset), \emptyset, \emptyset) \\ I \subseteq \IW}
  {x \detval{\sigma_0}_t \wrval(\sigma_0.\last(x))}$
  \ \  
  $ \inference[{\sc ModLast}] {x = \loc(e) \\  e \in \W_{|x} \\ m=\sigma.\last(x)} 
  {x \detval{\sigma'}_{\tid(e)} \wrval(e)}$
  \ \ 
  $\inference[{\sc Transfer}] {y = \var(e) \quad x \hbo{\sigma} y \\
  x \detval{\sigma}_{t} v \quad (m, e) \in \sw \\
  m = \sigma.\last(y)}
  {x \detval{\sigma'}_{\tid(e)} v}$
  \ \ 
      $\inference[{\sc UOrd}] {m \in \W_{{\sf R}|y} \\ e \in \URA_{|y} \\ x \hbo{\sigma} y }
  {x \hbo{\sigma'} y}$

  \bigskip 
  $\inference[{\sc NoMod}]
  {e \notin \W_{|x} \\ x \detval{\sigma}_t v}
  {x \detval{\sigma'}_t v}$
  \ \   
  $\inference[{\sc AcqRd}] {x = \loc(e) \quad e \in \RA_{|x} \\ m \in \WR_{|x} \quad
  m = \sigma.\last(x)}
  {x \detval{\sigma'}_{\tid(e)} \rdval(e)}$
  \ \ 
  $\inference[{\sc WOrd}] {x \neq y \quad e \in \W_{|y} \\
  x \detval{\sigma}_{\tid(e)} v \quad
  m = \sigma.\last(y)}
  {x \hbo{\sigma'} y}$
  \ \ 
  $\inference[{\sc NoModOrd}] {e \notin \W_{|\{x, y\}} % \cup \W_{|y} 
    \\ x \hbo{\sigma} y}
  {x \hbo{\sigma'} y}$
  \vspace{-3mm}
  
  \caption{Rules for determinate-value and variable-order assertions.
    We assume $\sigma, m, e, \sigma'$ satisfy
    $(\_,\sigma) \ltsArrow{m, e}_{\Pad} (\_,\sigma')$. }

  \label{fig:det-val-ind}
\end{figure*}

\smallskip
\noindent{\bf Last modification transitions.}
Observe that the rules in Figure \ref{fig:det-val-ind} are all conditioned on
the modification that is observed in the transition being the last modification
to the given variable. Thus, we must be able to prove that a given read or
update observes the last modification. There are several ways to do this.
It is easy to see that if $x \detval{\sigma}_t v$
for some thread $t$ in some state $\sigma$ then $t$ can only read the last
write to $x$. We formalise this claim in Lemma \ref{lem:last-mod-trans} below,
and in our case study we show how to use it in verification.

Update operations provide another way to guarantee
that a given operation observes the last modification at a given
variable. Given a C11 state $((D, \_), \_, \ltmo)$, an {\em update-only} variable is
any variable $x$ such that for all modifications $m \in D$ with $x = \loc(m)$,
either $m$ is an update or $m \in \IW$.
Note that initially, every variable is an update only variable.
In the operational semantics, update-only variables have the property that 
any new update or write can only be added to the \emph{end} of the modification order.
This is a consequence of the fact that for such a variable, any modification
but the last is covered. Thus, we have the following lemma.
%{\bf TODO: updates might only read from last write, but determinate value
%assertion is not true.}
\begin{lemma}[Last Modification Transition]
\label{lem:last-mod-trans}
 Let $t = \tid(e)$ and
  $x = \var(e)$ for some event $e$. For any reachable transition
  $(P,\sigma) \ltsArrow{m, e}_\Pad (P',\sigma')$,
  $m = \sigma.\last(x)$ if either $x \detval{\sigma}_{t} v$, for
  some value $v$, or $x$ is an update only variable in $\sigma$.
% \begin{enumerate}
% \item \label{last-mod-trans:detval} $x \detval{\sigma}_{t} v$, for
%   some value $v$, or 
% \item  \label{last-mod-trans:up-only}  $x$ is an update only variable in $\sigma$.
% \end{enumerate}
\end{lemma}
% These conditions depend on whether or not the observed modification of a transition
% is the last modification. The following lemma specifies conditions under which this
% is guaranteed to be true.

\noindent In other cases, other kinds of invariants can be used to guarantee
this last-modification property.

\begin{example}
Consider the following message-passing interaction between two
threads:

\begin{minipage}[t]{1.0\columnwidth}
  {\bf Init:} $f =0 \wedge d=0$
  
$\begin{array}{l@{\qquad \qquad}l} 
       \text{\bf thread}\ 1 & \text{\bf thread}\ 2\\
       1:\quad d :=5; & 1:\quad \text{\bf while } !f^{\sf A} \ \text{\bf do  skip};  \\
       2:\quad f :=^{\sf R} 1; & 2:\quad r := d; 
     \end{array}$
\end{minipage}

\smallskip \noindent Here, thread 1 sets the data variable $d$ to $5$,
and then indicates that the data is ready by setting the flag variable
$f$ to $1$. Thread 2 awaits this condition, and then consumes the
data. In order to show that this simple program is correct, we must be
able to prove that thread 2 always reads the correct value at line 2.

We sketch a proof that for any state $\sigma'$, where thread 2 is at
line 2, we have $d \detval{\sigma'}_2 5$. First note that this program
satisfies the invariant that for each write $w$ satisfying
$\var(w) = f$ and \mbox{$\wrval(w) = 1$}, $w$ is a releasing write of
thread 1 and \mbox{$\last(f) = w$}.  Using rules {\sf ModLast} and
{\sf WOrd}, after executing line 2 of thread 1, the resulting state
$\sigma$ satisfies $d \detval{\sigma}_1 5$ and $d \hbo {\sigma} f$.
This fact, along with the invariant above satisfy the premises of the
{\sf Transfer} rule where $x$ is $d$ and $y$ is $f$.  Thus, when the
loop exits, into state $\sigma'$, we have $d \detval{\sigma'}_2 5$, as
required.
\end{example}
Equipped with these techniques, we now
 show that Peterson's algorithm with the synchronisation annotations as 
given in Section \ref{sec:comm-lang-unint} guarantees mutual exclusion.

% !TeX root = ppopp.tex

\subsection{An Example Verification: Peterson's Algorithm}
\label{sect:case-study}

We turn now to the verification of the version of Peterson's Mutual Exclusion algorithm
given in Algorithm \ref{alg:petersons-ra}. Our verification consists
of proving a {\em mutual exclusion invariant} (Theorem \ref{theor:mutex}) stating that
there is no reachable state in which both processes are in their respective
critical sections.

To state our invariants, we make use of an auxiliary {\em program counter} function,
which for each thread, returns the line number of Algorithm \ref{alg:petersons-ra}
that the thread is currently executing.  
More precisely, for each configuration $(P, \sigma)$ of Peterson's algorithm, and $t$  a thread with $t \in \{1, 2\}$, the
expression $P.\pc_t$ returns $i$ when 
$P(t)$ is the part of the program starting on line $i$.

The mutual exclusion property for   \refalg{alg:petersons-ra} is proved in \refthm{theor:mutex}, which relies on the following invariants. \begin{eqnarray}
  &\text{$turn$ is an update-only location} \label{pete-prop:update-only}\\
  &\mathit{turn} \detval{\sigma}_1 2 \vee\mathit{turn} \detval{\sigma}_2 1  \label{pete-prop:turn}\\
  &P.\pc_t \in \{\ref{swap-turn}, \ref{busy-wait}, \ref{critical-section}, \ref{unset-flag}\} \implies \mathit{flag}_{t} \detval{\sigma}_t \mathit{true}   \label{pete-prop:flag}\\
  &P.\pc_t \in \{\ref{busy-wait}, \ref{critical-section}, \ref{unset-flag}\} \implies \mathit{flag}_t \hbo{\sigma}\mathit{turn}   \label{pete-prop:flag-turn-ord} \\
%   \end{gather}
%   \begin{gather}
&\begin{array}[t]{@{}l@{}}
  P.\pc_t \in \{\ref{busy-wait}, \ref{critical-section}, \ref{unset-flag}\} \wedge  P.\pc_{\notT} \in \{\ref{busy-wait}, \ref{critical-section}, \ref{unset-flag}\} \implies  \\ 
     \ \qquad \qquad \qquad \qquad \mathit{flag}_{\notT} \detval{\sigma}_t \mathit{true} \vee\mathit{turn} \detval{\sigma}_{\notT} t
\end{array}
   \label{pete-prop:cross}\\
  &P.\pc_t = \ref{critical-section} \wedge  P.\pc_{\notT} \in \{\ref{busy-wait}, \ref{critical-section}, \ref{unset-flag}\} \implies \mathit{turn} \detval{\sigma}_{\notT} t
  \label{pete-prop:crit}\\
  &P.\pc_t = \ref{set-flag} \implies \mathit{flag}_t \detval{\sigma} \mathit{false}
  \label{pete-prop:quies}
\end{eqnarray}
As in the classical (sequentially consistent) setting, we prove that these invariants hold for the initial configuration and for each transition of the algorithm. For space reasons we only provide details for one of these cases, i.e., where the first test at line \ref{busy-wait} is evaluated to false, causing it to enter the critical section.\footnote{The full proof is available in the extended version.}

\begin{proof}
We consider the first test at line \ref{busy-wait},
$\mathit{flag_t} = \mathit{false}$,
in the case where the test fails (the success
case is very simple).
% If this test fails, then nothing about the state changes except that $t$ moves to
% the second test in the condition. Because nothing about the state is changing,
% application of the rules {\sc NoMod} and {\sc NoModOrd} can be used to show that all the invariants are preserved in a standard way.
% Therefore, we only consider in detail the situation when
% the test succeeds.
Let $(P', \sigma')$ be the configuration after the step in question
Assume that $P.\pc_t = \ref{busy-wait}$, $P'.\pc_t = \ref{critical-section}$,
and $e = R_t(\mathit{flag}_{\notT}, \mathit{false})$.

Because $e$ is not a write and the value of $\pc_{\notT}$ does not change,
it is easy to use the {\sc NoMod} and {\sc NoModOrd} rules
to show that each invariant except for \refeq{pete-prop:crit}
is preserved. We now prove that
\refeq{pete-prop:crit} is preserved. We do so by proving that $\mathit{turn} \detval{\sigma'}_{\notT} t$
under the assumption that $P'.\pc_{\notT} \in \{\ref{busy-wait}, \ref{critical-section}, \ref{unset-flag}\}$.
Because $P.\pc_{\notT} = P'.\pc_{\notT}$, we have $P.\pc_{\notT} \in \{\ref{busy-wait}, \ref{critical-section}, \ref{unset-flag}\}$. Furthermore, by Lemma~\ref{lem:det-val-read} and
the fact that $e = R_t(\mathit{flag}_{\notT}, \mathit{false})$ and $m = \sigma.\last(\mathit{flag}_{\notT})$
the assertion $\mathit{flag}_{\notT} \detval{\sigma}_t \mathit{true}$ is false. Thus by Invariant \ref{pete-prop:cross},
we have $\mathit{turn} \detval{\sigma}_{\notT} t$.
Then, from rule {\sc NoMod}, and the fact that 
$e$ is not a write, we have $\mathit{turn} \detval{\sigma'}_{\notT}  t$,
as required.
\end{proof}

These invariants are sufficient to prove that Peterson's Algorithm satisfies the mutual exclusion property.
\begin{theorem}[Mutual exclusion]
\label{theor:mutex}
For each reachable configuration $(P, \sigma)$, $P.\pc_1 \neq 5 \vee P.pc_2 \neq 5$.
\end{theorem}
\begin{proof}
Assume %for a contradiction 
that $P.\pc_1 = 5$ and \mbox{$P.\pc_2 = 5$.} Then, by
Property \refeq{pete-prop:crit} above, we have $\mathit{turn} \detval{\sigma}_1  2$
and $\mathit{turn} \detval{\sigma}_2 1$.
But this is impossible by Lemma~\ref{lem:ow-agreement}. Thus,
$P.\pc_1 \neq 5$ or $P.\pc_2 \neq 5$. 
\end{proof}

\section{Conclusion and Related Work}

We have developed an operational semantics for the RAR fragment of the
C11 memory model, which has been shown to be both sound and complete with
respect to the axiomatic description. Thus, every state generated by
the operational semantics is guaranteed to be one allowed by the
axiomatic semantics. Moreover, any execution that is valid with
respect to the axiomatic semantics can be generated by the operational
semantics. Our semantics relies on a thread-local view of observability\footnote{Our notion of observability differs from those defined in \cite{DBLP:conf/ecoop/KaiserDDLV17,DBLP:conf/fm/FavaSS18}.}, which is defined 
in terms of $\lteco$ and $\lthb$ orders. We have developed a proof technique for our operational semantics with a notation that follows conventional proofs of  sequentially consistent memory as much as possible. Finally, we have applied this technique to an example verification. 

% \bd{expand/reword}

% Restricting attention to the RAR fragment of C11 has been a worthwhile
% pursuit as it has enabled us to explore the issues in weak memory
% verification (considering partial orders) in a simpler setting. 

There is a large body of related work; here, we provide a brief snapshot. There are several works aimed at
providing operational semantics for a larger subset of C11, including models that aim to address the so-called thin-air problem
(that we rule out by the {\sc No-Thin-Air} axiom), which invariably lead to more complex semantics.  \citet{DBLP:conf/oopsla/NienhuisMS16} provide a semantics that supports inductive reasoning, but they are forced to consider an order that does not include $\ltsb$. This complicates a verification technique that follows program order.
% Moreover, their model comprises independent program and concurrency transitions -- a program may generate an event that is only committed to memory at a later point in the execution. In contrast, we support  every read is justified at the point of execution.
\citet{DBLP:conf/popl/KangHLVD17} develop an operational model aimed at handling cycles in $\ltsb \cup \ltrf$.
Again, their sophisticated model handles a larger subset of the C11 language, but at the cost of a more complicated
state space and transition relation.  \citet{DBLP:conf/popl/LahavGV16} provide an operational model for a stronger 
release-acquire model, where $\ltsb \cup \ltrf \cup \ltmo$ is required to be acyclic.

\citet{DBLP:conf/popl/KangHLVD17} provide a basic program logic for
proving invariants; using their semantics in verification remains an open problem. 
% but to express invariants in this logic one would need to deal with the full complexity
% of their state space. 
\citet{DBLP:conf/esop/JagadeesanPR10} develop an operational semantics for
capable of coping with out-of-order executions for the Java memory model. However, their work aims to support Java compiler optimisations and they do not consider program verification.
One avenue for future work is to see how our notions of determinate-value and
variable-ordering assertions might be applied to a more sophisticated semantics~\cite{DBLP:conf/esop/JagadeesanPR10,DBLP:conf/popl/KangHLVD17}.

Concurrent separation logic (CSL) provides a different approach to verification,
and several frameworks have been developed for dealing with C11-style weak memory
\cite{DBLP:conf/ecoop/KaiserDDLV17,DBLP:conf/vmcai/DokoV16,DBLP:conf/esop/DokoV17,DBLP:conf/popl/JungSSSTBD15}.
These frameworks typically deal with a fragment of C11 containing release/acquire operations
but that is not comparable to the fragment of our model. Weak-memory CSL
has been a very active area of research for several years, and we refer the reader to
the introduction of \cite{DBLP:conf/ecoop/KaiserDDLV17} for an excellent review.

% OW proof system for release-acquire.
% Because it is only RA, $\ltsb \cup \ltrf \subseteq \lthb$,
% so it also has the acyclicity of $\ltsb \cup \ltrf$. Auxiliary variables
% are not sound, and . We can use them.

% None of the works above (as far as we are aware) have been applied to verify a real algorithm. 
% (Actually, thye use the taming semantics to verify RCA (s.d.))

%%% Local Variables:
%%% mode: latex
%%% TeX-master: "ppopp-full"
%%% End:

%% Acknowledgments
\begin{acks}                            
  This work has been supported by EPSRC grants
  \grantnum{EPSRC}{EP/R032556/1} and \grantnum{EPSCR}{EP/M017044/1},
  and DFG grant \grantnum{DFG}{WE 2290/12-1}. 
  We thank John Wickerson for helpful discussions. We
  also thank our anonymous reviewers and shepherd (Viktor Vafeiadis)
  for their comments, which have helped improve the paper.

%% acks environment is optional
                                        %% contents suppressed with 'anonymous'
  %% Commands \grantsponsor{<sponsorID>}{<name>}{<url>} and
  %% \grantnum[<url>]{<sponsorID>}{<number>} should be used to
  %% acknowledge financial support and will be used by metadata
  %% extraction tools.

  % This material is based upon work supported by the
  % \grantsponsor{GS100000001}{National Science
  %   Foundation}{http://dx.doi.org/10.13039/100000001} under Grant
  % No.~\grantnum{GS100000001}{nnnnnnn} and Grant
  % No.~\grantnum{GS100000001}{mmmmmmm}.  Any opinions, findings, and
  % conclusions or recommendations expressed in this material are those
  % of the author and do not necessarily reflect the views of the
  % National Science Foundation.
\end{acks}

\bibliography{references}

\clearpage
\appendix

\section{Proofs of \refsec{sec:soundn-compl}}

% \begin{theorem}% [Soundness]
%   \label{thm:completeness}
% \end{theorem} 
\noindent{\bf \refthm{thm:completeness}.} {\em 
  Let $\sigma=((D,\ltsb),\ltrf,\ltmo)$ be a C11 state reachable from
  $\sigma_0$ using relation $\strans{\ \ \ }_{\Pad}$. Then $\sigma$
  satisfies {\sc SB-Total}, {\sc MO-Valid}, {\sc RF-Complete}, {\sc
    No-Thin-Air} and {\sc Coherence}.
}

\begin{proof}[Proof of \refthm{thm:completeness}]
  
By induction on the number of steps executed to reach $\sigma$. 

\medskip \noindent {\bf Induction base.} The initial state $\sigma_0$ satisfies all conditions as all relations are 
empty and there are no read event in $\sigma_0$. % \bd{changed SB-TOTAL}
% \fbox{actually not true: we do not totally order the events of the thread $i$}
% \fbox{need to change definition of SB-TOTAL}

\medskip \noindent {\bf Induction step.} Let $\sigma_i$ be a C11 state
reachable in $i$ steps that fulfils the C11 consistency
conditions.  % {\sc SB-Total} to {\sc
                                              % Coherence}.
Let $\sigma_i \strans{\ e\ }_{C11} \sigma_{i+1}$.  We need to show that
$\sigma_{i+1}$ satisfies all conditions.

\smallskip\noindent
{\sc SB-Total}: Follows from definition of + and induction hypothesis.

\smallskip \noindent {\sc MO-Valid}: Follows from definition of $\ltmo[w,e]$ and induction hypothesis.

\smallskip \noindent {\sc RF-Complete}: Follows from rules {\sc Read} and {\sc RMW} and induction hypothesis. 

\smallskip \noindent {\sc No-Thin-Air}: Let
$\sigma_i = ((D_i,\ltsb_i),\ltrf_i,\ltmo_i)$, assume (by induction
hypothesis) that $\ltsb_i \cup \ltrf_i$ is acyclic and consider the
introduction of element $e$ to $\sigma_i$. For each rule in
\reffig{fig:c11-opsem}, $e$ is maximal in $\ltsb_{i+1}$. Thus
$\ltsb_{i+1} \cup \ltrf_i$ is acyclic. Moreover, if $e$ is a write
$\ltrf_{i+1} = \ltrf_i$, and if $e$ is a read or an update, $e$ is
maximal in $\ltrf_{i+1}$, and hence, $\ltsb_{i+1} \cup \ltrf_{i+1}$
is acyclic.

\smallskip \noindent {\sc Coherence}. Assume $\lthb_i ; \lteco_i^?$ is
irreflexive. Consider case distinction on the type of event
$e$.
\begin{itemize}
\item $e$ is a read event. This introduces edges
  $(e',e) \in \ltsb_{i+1}$, $(w,e) \in \ltrf_{i+1}$ and
  $(e,w') \in \ltfr_{i+1}$ for each $w'$ such that
  $(w,w') \in \ltmo_i$.  If we have a cycle in
  $\lthb_{i+1}; \lteco_{i+1}^?$, the cycle has to pass through $e$,
  and therefore also leave $e$ via some $(e, w') \in \ltfr_{i+1}$ edge
  (since these are the only outgoing edges from $e$). There are two
  cases:
 
  \begin{enumerate}
  \item There is a path with edges $(w', e'') \in \lteco_{i+1}^?$
    and $(e'', e') \in \lthb_{i+1}^?$ for some $e'' \in D_i$. 
    
  \item There is a path with edges $(w', e'') \in \lteco_{i+1}^?$ and
    $(e'', w) \in \lthb_{i+1}^?$ and $(w, e) \in \ltsw_{i+1}^?$ (due to
    the events $w$ and $e$ synchronising via {\sf R} and {\sf A}
    synchronisation). 
  \end{enumerate}
  Both cases potentially create a reflexive edge via the composition
  of edges $(e'', e) \in \lthb_{i+1}$ and $(e, e'') \in \lteco_{i+1}$.
  However, we now have $w' \in \observedWrites_{\sigma_i}(t)$ and
  since $(w,w') \in \ltmo$, we have $w \notin \OW_{\sigma_i}(t)$ and
  hence the edge $(w, e) \in \ltrf_{i+1}$ cannot exist, giving rise to
  a contradiction.

\item $e$ is a write event. This introduces edges
  $(e',e) \in \ltsb_{i+1}$, $(w,e) \in \ltmo_{i+1}$ and
  $(r,e) \in \ltfr_{i+1}$ for any read $r$ in $D_i$ that reads
  $\var(e)$. If $e$ is maximal in $\ltmo_{i+1}$ we are done as
  $\lthb_{i+1} ; \lteco_{i+1}^?$ cannot be reflexive. Otherwise, there
  must be an edge that leaves $e$ via an edge
  $(e, w') \in \ltmo_{i+1}$. We have two cases:
  \begin{enumerate}
  \item There is a path with edges $(w', e'') \in \lteco_{i+1}^?$ and
    $(e'', e') \in \lthb_{i+1}^?$. This potentially creates a
    reflexive edge, via the composition of edges
    $(e'', e) \in \lthb_{i+1}$ and $(e, e'') \in
    \lteco_{i+1}$. However, this would mean we have
    $w \notin \OW_{\sigma_i}(t)$, and hence the edge
    $(w, e) \in \ltmo_{i+1}$ cannot exist.
  \item There is a path with edges $(w', e'') \in \lteco_{i+1}^?$ and
    $(e'', w) \in \lthb_{i+1}$. This potentially creates a reflexive
    edge, via the composition of edges $(e'', w) \in \lthb_{i+1}$ and
    $(w, e'') \in \lteco_{i+1}$. However, this also means that we have
    edges $(w, w') \in \ltmo_i$, $(w', e'') \in \lteco_i^?$ and
    $(e'', w) \in \lthb_i$, i.e., $\lthb_i ; \lteco_i^?$ is reflexive,
    which is a contradiction.
  \item There is a path with edges $(w', e'') \in \lteco_{i+1}^?$ and
    $(e'', r) \in \lthb_{i+1}$. This case is similar to the one above.

  \end{enumerate}
\item $e$ is an update event. This introduces edges
  $(e', e) \in \ltsb_{i+1}$, $(w, e) \in \ltrf_{i+1}$,
  $(w, e) \in \ltmo_{i+1}$. The proof is similar to the proofs of the
  read and write cases.
\end{itemize} 
We now show $\lteco_{i+1}$ is irreflexive, assuming  $\lteco_{i}$ is
irreflexive. We perform case analysis on the type of event $e$.
\begin{itemize}
\item $e$ is a read event. This introduces $\lteco$ edges
  $(w,e) \in \ltrf_{i+1}$ and $(e,w') \in \ltfr_{i+1}$ for each $w'$
  such that $(w,w') \in \ltmo_{i}$. If $\lteco_{i+1}$ is reflexive, we
  must have an edge $(w', w) \in \lteco_{i+1}$. But this means we have
  edges $(w', w) \in \lteco_i$ and
  $(w, w') \in \ltmo_i \subseteq \lteco_{i+1}$, i.e., $\lteco_i$ is
  reflexive, which is a contradiction.
\item $e$ is a write event. This introduces $\lteco$ edges
  $(w,e) \in \ltmo_{i+1}$ and $(r,e) \in \ltfr_{i+1}$. If $e$ is
  maximal in $\ltmo_{i+1}$ we are done as $\lteco_{i+1}$ cannot be
  reflexive. Otherwise, there is an path that leaves $e$ via an edge
  $(e, w') \in \ltmo_{i+1}$. Then we either have a path with edge
  $(w', w) \in \lteco_{i+1}$ or a path with edge
  $(w', r) \in \lteco_{i+1}$. However, both contradict the assumption
  that $\lteco_i$ is irreflexive.
\item $e$ is an update event. This introduces $\lteco$ edges
  $(w,e) \in \ltrf_{i+1}$, $(w,e) \in \ltmo_{i+1}$, as well as
  $(e, w') \in \ltmo_{i+1}$ and $(e,w') \in \ltfr_{i+1}$ for each $w'$
  such that $(w,w') \in \ltmo_{i}$. If $\lteco_{i+1}$ is reflexive, we
  must have an edge $(w', w) \in \lteco_{i+1}$. But this means we have
  edges $(w', w) \in \lteco_i$ and
  $(w, w') \in \ltmo_i \subseteq \lteco_{i+1}$, i.e., $\lteco_i$ is
  reflexive, which is a contradiction. 
\end{itemize}
\end{proof}

\smallskip\noindent {\bf \reflem{lem:permutation}.} {\em 
  Let
  $Q = (P_0,\gamma_0) \xRightarrow{e_1}_\PreExec (P_1,\gamma_1)
  \xRightarrow{e_2}_\PreExec \ldots \xRightarrow{e_k}_\PreExec
  (P_k,\gamma_k)$ and $\gamma_k = (D_k, \ltsb_k)$. Then for
  every % permutation
  linearization $f_1, \ldots , f_k$ of % $e_1, \ldots, e_k$ that
  % preserves
  $\ltsb_k$, there exist programs $P_1', \ldots, P_{n-1}'$
  and pre-execution states $\gamma_1', \ldots, \gamma_{n-1}'$ such
  that
  \[
    (P_0,\gamma_0) \xRightarrow{f_1}_\PreExec (P_1',
    \gamma_1') \xRightarrow{f_2}_\PreExec \ldots
    \xRightarrow{f_k}_\PreExec (P_k,\gamma_k)\ .
  \]}
\begin{proof}[Proof of \reflem{lem:permutation}]
  Let $F = f_1, \dots, f_k$ and $E = e_1, \dots, e_k$, and let
  $\delta_i$ represent the pair $(P_i, \gamma_i)$. Suppose $f_1$, the
  first element of $F$ is the element $e_i$, the $i$th element of
  $E$. We show that $Q$ can be transformed into a valid sequence of
  $\PreExec$ steps such that $e_i$ is the first event considered.  By
  definition, we have that $Q$ is the sequence:
  \begin{align*}
    % Q_1  = 
    \delta_0 \xRightarrow{e_1}_\PreExec \dots \delta_{i-2}\xRightarrow{e_{i-1}}_\PreExec \delta_{i-1}
    \xRightarrow{e_i}_\PreExec \delta_{i} \ldots \xRightarrow{e_k}_\PreExec 
    \delta_k\ .
  \end{align*}
  Since both $E$ and $F$ are a linearizations of $\ltsb_k$, we have
  the property:
  \begin{align}
    \label{eq:tid-prop}
    \forall j.\ 0 \le j \le i-1 \imp \tid(e_j) \neq tid(e_i)
  \end{align}
  By \refeq{eq:tid-prop}, we have in particular that
  $\tid(e_{i-1}) \neq \tid(e_i)$. Thus by
  \refprop{prop:program-PE-commute} there must exist a
  $\delta_{i-1} '$ such that
  $\delta_{i-2}\xRightarrow{e_{i}}_\PreExec \delta_{i-1}'
  \xRightarrow{e_{i-1}}_\PreExec \delta_{i}$, and hence, a valid
  pre-execution sequence
  \begin{align*}
    % Q_2  = 
    \delta_0 \xRightarrow{e_1}_\PreExec \dots \delta_{i-2}\xRightarrow{e_{i}}_\PreExec \delta_{i-1}'
    \xRightarrow{e_{i-1}}_\PreExec \delta_{i} \ldots \xRightarrow{e_k}_\PreExec
    \delta_k\ .
  \end{align*}
  Again by property \refeq{eq:tid-prop}, we have that
  $\tid(e_{i-2}) \neq \tid(e_i)$ and the process above can be
  repeated so that we obtain:
  \begin{align*}
    % Q_3  = 
    \delta_0 \xRightarrow{e_1}_\PreExec \dots \xRightarrow{e_{i}}_\PreExec \delta_{i-2}'
    \xRightarrow{e_{i-2}}_\PreExec \delta_{i-1}'
    \xRightarrow{e_{i-1}}_\PreExec \delta_{i} \ldots \xRightarrow{e_k}_\PreExec
    \delta_k\ .
  \end{align*}
  Further repeating this process, we obtain:
  \begin{align*}
    % Q_i  =
    \begin{array}[t]{@{}l@{}}
      \delta_0 \xRightarrow{e_i}_\PreExec \delta_1' \xRightarrow{e_1}_\PreExec \delta_2'  \dots \delta_{i-2}'   
      \xRightarrow{e_{i-2}}_\PreExec \delta_{i-1}' 
      \xRightarrow{e_{i-1}}_\PreExec \qquad \\ \hfill \delta_{i} \xRightarrow{e_{i+1}}_\PreExec \ldots \xRightarrow{e_k}_\PreExec
      \delta_k\ .
    \end{array}
  \end{align*}
  which (since $e_i = f_1$) is equivalent to:
  \begin{align*}
    \begin{array}[t]{@{}l@{}}
    \delta_0 \xRightarrow{f_1}_\PreExec \delta_1' \xRightarrow{e_1}_\PreExec \delta_2' \dots \delta_{i-2}'    
    \xRightarrow{e_{i-2}}_\PreExec \delta_{i-1}'
      \xRightarrow{e_{i-1}}_\PreExec \qquad \\
      \hfill \delta_{i} \xRightarrow{e_{i+1}}_\PreExec  \ldots \xRightarrow{e_k}_\PreExec
    \delta_k\ .
    \end{array}
  \end{align*}
  We can now repeat the entire process for $f_2$ using the property
  and percolate the element it corresponds to in $E$ it to the correct
  position in $F$ since $f_2 \neq e_i$, i.e. $f_2$ corresponds to an
  element in $\{e_1, \dots e_k\} \backslash \{e_i\}$. Assuming that
  $f_2$ corresponds to position ${i'}$ in $E$, we have property
  $\forall j.\ 1 \le j \le i'-1 \imp \tid(e_j) \neq tid(e_{i'})$
  (analogous to \refeq{eq:tid-prop}), where the lower index is
  increased by $1$ and upper index is adjusted to $i' - 1$. Once $f_2$
  is in position, we can repeat for $f_3$ and so forth. \hfill \qed
\end{proof}

We now show that for every justifiable pre-execution there is an
execution of the C11 semantics that ends in the C11 state justifying
the pre-execution.  The theorem uses a notion that restricts
pre-executions and C11 executions to a set of events. For a set of
events $E \subseteq D$, we define:
\begin{align*}
  (D, \ltsb)_{\downarrow E} & = (E, \ltsb \cap (E \times E)) \\
  (\gamma,\ltrf,\ltmo)_{\downarrow E} & = (\gamma_{\downarrow E}, \ltrf \cap (E \times E), \ltmo \cap (E \times E))
\end{align*}
In the completeness proof, we assume that the given pre-execution
sequence
$(P_0,\gamma_0) \xRightarrow{e_1}_\PreExec (P_1,\gamma_1)
\xRightarrow{e_2}_\PreExec \ldots \xRightarrow{e_k}_\PreExec
(P_k,\gamma_k)$ has been reordered such that $e_1 \dots e_k$ is a
linearization of $\ltsb_k \cup \ltrf_k$, where $\ltrf_k$ is the
reads-from relation used in the justification of $\gamma_k$. Such a
linearization is possible since $\ltsb_k \cup \ltrf_k$ is acyclic (axiom {\sc No-Thin-Air}).

\smallskip \noindent {\bf \refthm{thm:soundness}.}{\em 
  Suppose
  $(P_0,\gamma_0) \xRightarrow{e_1}_\PreExec (P_1,\gamma_1)
  \xRightarrow{e_2}_\PreExec \ldots \xRightarrow{e_k}_\PreExec
  (P_k,\gamma_k)$ such that $\gamma_k = (D_k, \ltsb_k)$ is justifiable
  with justification $\sigma_k= (\gamma_k,\ltrf_k,\ltmo_k)$ and
  $e_1, \dots, e_k$ is a linearization of $\ltsb_k \cup \ltrf_k$. Then
  \[ (P_0,\sigma_0) \xRightarrow{e_{1}}_{\Pad} (P_1, \sigma_1)
    \xRightarrow{e_{2}}_{\Pad} \ldots \xRightarrow{e_k}_{\Pad}
    (P_k,\sigma_k) \] where
  $\sigma_i = (\gamma_k,\ltrf_k,\ltmo_k)_{\downarrow \{e_1, \ldots, e_i
  \}}$, $0 < i < k$.}
\begin{proof}[Proof of \refthm{thm:soundness}]
  By induction on the number of steps.
   %\fbox{current proof without covering} 

\noindent {\bf Base case.} The initial configurations agree and hence the claim holds for 0 steps.

\smallskip
\noindent {\bf Induction step.} Let the above claim hold for sequences
up to length $j$. We perform a case split on the type of event
$e_{j+1} \notin D_j$.
\begin{enumerate}
\item $e_{j+1}$ is a read event of a thread $t$, i.e.,
  $\tid(e_{j+1}) = t$. %  and either $\act(e_{j+1}) = rd^{\sf A}(x,n)$, or
  % $\act(e_{j+1}) = rd(x,n)$ for some value $n$.
  % \qquad
  Let $w$ be the write event that $e_{j+1}$ reads from, i.e.,
  $(w,e_{j+1}) \in \ltrf_k$. We known $w \in D_j$ since we consider
  elements in $\ltsb_k \cup \ltrf_k$ order. % Then
  % $\act(w) = wr^{\sf A}(x,n)$ or $\act(w) = wr(x,n)$.
  
  We need to show that $w \in \OW_{\sigma_j}(t)$.  The proof is by
  contradiction. Assume $w \notin \OW_{\sigma_j}(t)$, then there
  exists a $w' \in \observedWrites_{\sigma_j}(t)$ such that
  $(w,w') \in \ltmo_j$. Hence $(e_{j+1}, w') \in \ltfr_{k}$ and there
  exists some $e$ % with $\tid(e) = t$
  such that $(w', e) \in \lteco_{k}^?$ and $(e, e_{j+1}) \in \lthb_{k}^?$. There are three possibilities:
  \begin{itemize}
  \item $(w', e), (e, e_{j+1}) \in \emph{Id}$. This is an immediate contradiction since it implies $w' = e_{j+1}$. 
  \item $(w', e) \in \lteco_{k}$ and $(e, e_{j+1}) \in \emph{Id}$,
    i.e., $e = e_{j+1}$. This contradicts the assumption that
    $\lteco_{k}$ is irreflexive.

  \item $(w', e) \in \lteco_{k}^?$ and
    $(e, e_{j+1}) \in \lthb_{k}$. We then have
    $(e_{j+1}, e) \in \lteco_{k}$ resulting in a contradiction to
    the assumption that $\lthb_{k} ; \lteco_{k}^?$ is irreflexive. 
  \end{itemize}
  The contradictory scenario is depicted by the following diagram:
  
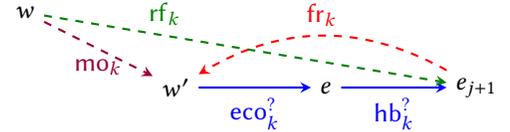
\begin{center}
     \begin{tikzpicture}% [node distance=.3cm]      
       \node (a) at (6,2) {$e_{j+1}$}; % =rd_t(x,n)$};
       \node (b) at (4,2) {$e$};
       \node (c) at (2,2) {$w'$};% = wr(x,\cdot)$}; 
       \node (d) at (0,3) {$w$}; % =(\_, wr(x,n), \_)$}; 
       \path 
       (d) edge[mo] node[below] {$\ltmo_{k}$} (c)
       (c) edge[hb] node[below] {$\lteco_{k}^?$} (b)
       (b) edge[hb] node[below] {$\lthb_{k}^?$} (a) 
       (d) edge[rf] node[pos=0.3,above] {$\ltrf_{k}$} (a)
       (a) edge[fr,bend right] node[above] {$\ltfr_{k}$} (c); 
     \end{tikzpicture}
   \end{center}

 \item Suppose $e_{j+1}$ is a write event and  %  $act(e_{j+1}) = wr^{\sf
     % R}(x,n)$ or $act(e_{j+1}) = wr(x,n)$ for some $n$.
   $w$ is the immediate predecessor of $e_{j+1}$ in
   $\ltmo_k$. Note that
   $w$ may either be a write or an update event. We must show that it
   is possible to take a
   $\strans{e_{j+1}}_{\Pad}$ step such that
   $e_{j+1}$ is placed immediately after
   $w$. To this end, we must show that $w \in \OW_{\sigma_j}(t)$.

   Suppose not, i.e., $w \notin \OW_{\sigma_j}(t)$.  Then there exists
   an event $w' \in \observedWrites_{\sigma_j}(t)$ such that
   $(w,w') \in \ltmo$ and an event $e$ such that
   $(w',e) \in \lteco_{k}^?$ and $(e,e_{j+1}) \in
   \lthb_{k}^?$. Since we have assumed $w$ is an immediate
   predecessor of $e_{j+1}$ in $\ltmo_{k}$ and that
   $(w, w') \in \ltmo_{k}$, we must have $(e_{j+1}, w')$.
   There are three possibilities:
   \begin{itemize}
   \item $(w', e), (e, e_{j+1}) \in \emph{Id}$. This is an immediate
     contradiction since it implies $w' = e_{j+1}$ and we have assumed
     $w' \in D_j$, $e_{j+1} \notin D_j$.
   \item $(w', e) \in \lteco_{k}$ and $(e, e_{j+1}) \in \emph{Id}$,
     i.e., $e = e_{j+1}$. This contradicts the assumption that
     $\lteco_{k}$ is irreflexive.
     
   \item $(w', e) \in \lteco_{k}^?$ and
     $(e, e_{j+1}) \in \lthb_{k}$. We then have
     $(e_{j+1}, e) \in \lteco_{k}$ resulting in a contradiction to
     the assumption that $\lthb_{k} ; \lteco_{k}^?$ is irreflexive. 
   \end{itemize}
   
  The contradictory scenario is depicted by the following diagram:

      
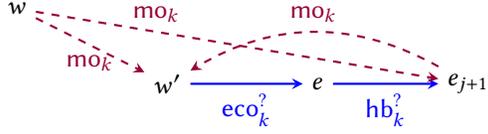
\begin{center}
        \begin{tikzpicture}
           \node (a) at (6,2) {$e_{j+1}$};
           \node (b) at (4,2) {$e$};
           \node (c) at (2,2) {$w'$}; 
           \node (d) at (0,3) {$w$}; 
           \path 
           (d) edge[mo] node[below] {$\ltmo_{k}$} (c)
           (c) edge[hb] node[below] {$\lteco_{k}^?$} (b)
           (b) edge[hb] node[below] {$\lthb_{k}^?$} (a) 
           (d) edge[mo] node[pos=0.3,above] {$\ltmo_{k}$} (a)
           (a) edge[mo,bend right] node[above] {$\ltmo_{k}$} (c); 
         \end{tikzpicture} 
       \end{center}
       We also need to show that $w$ selected is not covered. Again
       assume the contrary: there exists some update event $u$ such
       that $(w,u) \in \ltrf_k$. Then $(w,u) \in \ltmo_k$ as well.
       Hence there is an edge $(u,e_{j+1}) \in \ltfr_k$. Since the
       update $u$ and $e_{j+1}$ write to the same location, they need
       to be $\ltmo$-ordered. Here we have two cases: 
        \begin{itemize}
        \item If $(u,e_{j+1}) \in \ltmo_k$, then $w$ is not the immediate
          predecessor of $e_{j+1}$ in $\ltmo_k$.
        \item If $(e_{j+1},u) \in \ltmo_k$, then the $\ltfr_k$ edge
          and the $\ltmo_k$ edge together form a cyle, contradicting
          irreflexivity of $\lteco_k$.
        \end{itemize}
      \item Suppose $e_{j+1}$ is an update event  and $w$ is the
        immediate predecessor of $e_{j+1}$ in $\ltmo_k$. We must show
        that it is possible to take a\ $\strans{e_{j+1}}_{\Pad}$ step
        such that $e_{j+1}$ is placed immediately after $w$. This case
        is a combination of the read and write cases, namely if we
        assume $w \notin \OW_{\sigma_j}(t)$, then there must exist a
        $w'$ and $e$ as shown in the diagram below, which is a
        contradiction.
        
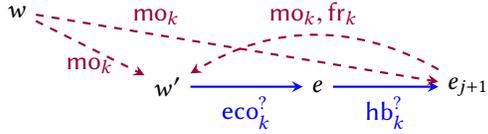
\begin{center}
         \begin{tikzpicture}
           \node (a) at (6,2) {$e_{j+1}$};
           \node (b) at (4,2) {$e$};
           \node (c) at (2,2) {$w'$}; 
           \node (d) at (0,3) {$w$}; 
           \path 
           (d) edge[mo] node[below] {$\ltmo_{k}$} (c)
           (c) edge[hb] node[below] {$\lteco_{k}^?$} (b)
           (b) edge[hb] node[below] {$\lthb_{k}^?$} (a) 
           (d) edge[mo] node[pos=0.3,above] {$\ltmo_{k}$} (a)
           (a) edge[mo,bend right] node[above] {$\ltmo_{k}, \ltfr_k$} (c); 
         \end{tikzpicture} 
       \end{center}
       % \hfill \qed
      \end{enumerate} 
 \end{proof}

%%% Local Variables:
%%% mode: latex
%%% TeX-master: "ppopp"
%%% End:

\section{Proofs for Section \ref{sect:verif}}
% \bd{what is $min_t$?}  \bd{Have we defined rechable state?}

% \begin{lemma}[Last Write]
% \label{lem:last-write-min}
% For any reachable state $\sigma$, thread $t$, location
% $x$ and value $v$, if $\sigma.x =_t v$ then
% $\sigma.\min_t(x) = \sigma.\last(x)$.
% \end{lemma}
% \begin{proof}
%   If $\sigma.\min_t(x)$ is ordered before $\sigma.\last(x)$ in $\ltmo$,
%   then we have $|\sigma.\OW(t, x)| > 1$.
% \end{proof}

\subsection{Proofs of lemmas}

\begin{figure*}[t]
  \centering
  $\inference[{\sc Init}]
  {\sigma_0 = ((I, \emptyset), \emptyset, \emptyset) \\ I \subseteq \IW}
  {x \detval{\sigma_0}_t \wrval(\sigma_0.\last(x))}$
  \ \  
  $ \inference[{\sc ModLast}] {x = \loc(e) \\  e \in \W_{|x} \\ m=\sigma.\last(x)} 
  {x \detval{\sigma'}_{\tid(e)} \wrval(e)}$
  \ \ 
  $\inference[{\sc Transfer}] {y = \var(e) \quad x \hbo{\sigma} y \\
  x \detval{\sigma}_{t} v \quad (m, e) \in \sw \\
  m = \sigma.\last(y)}
  {x \detval{\sigma'}_{\tid(e)} v}$
  \ \ 
      $\inference[{\sc UOrd}] {m \in \W_{{\sf R}|y} \\ e \in \URA_{|y} \\ x \hbo{\sigma} y }
  {x \hbo{\sigma'} y}$

  \bigskip 
  $\inference[{\sc NoMod}]
  {e \notin \W_{|x} \\ x \detval{\sigma}_t v}
  {x \detval{\sigma'}_t v}$
  \ \   
  $\inference[{\sc AcqRd}] {x = \loc(e) \quad e \in \RA_{|x} \\ m \in \WR_{|x} \quad
  m = \sigma.\last(x)}
  {x \detval{\sigma'}_{\tid(e)} \rdval(e)}$
  \ \ 
  $\inference[{\sc WOrd}] {x \neq y \quad e \in \W_{|y} \\
  x \detval{\sigma}_{\tid(e)} v \quad
  m = \sigma.\last(y)}
  {x \hbo{\sigma'} y}$
  \ \ 
  $\inference[{\sc NoModOrd}] {e \notin \W_{|\{x, y\}} % \cup \W_{|y} 
    \\ x \hbo{\sigma} y}
  {x \hbo{\sigma'} y}$
  \vspace{-3mm}
  
  \caption{Rules for determinate-value and variable-order assertions.
    We assume $\sigma, m, e, \sigma'$ satisfy
    $(\_,\sigma) \ltsArrow{m, e}_{\Pad} (\_,\sigma')$. }

  \label{fig:det-val-ind-appendix}
\end{figure*}

% \begin{figure*}[t]
%   \centering
%   $\inference[{\sc Init}]
%   {\sigma_0 = ((I, \emptyset), \emptyset, \emptyset) \\ I \subseteq \IW}
%   {x \detval{\sigma_0}_t \wrval(\sigma_0.\last(x))}$
%   \ \  
%   $\inference[{\sc NoMod}]
%   {e \notin \W_{|x} \\ x \detval{\sigma}_t v}
%   {x \detval{\sigma'}_t v}$
%   \ \  
%   $ \inference[{\sc ModLast}] {x = \loc(e) \\  e \in \W_{|x} \\ m=\sigma.\last(x)} 
%   {x \detval{\sigma'}_{\tid(e)} \wrval(e)}$
%   \ \ 
%   $\inference[{\sc Transfer}] {y = \var(e) \quad x \hbo{\sigma} y \\
%   x \detval{\sigma}_{t} v \quad (m, e) \in \sw \\
%   m = \sigma.\last(y)}
%   {x \detval{\sigma'}_{\tid(e)} v}$
%   \medskip 

%   $\inference[{\sc AcqRd}] {x = \loc(e) \quad e \in \RA_{|x} \\ m \in \WR_{|x} \quad
%   m = \sigma.\last(x)}
%   {x \detval{\sigma'}_{\tid(e)} \rdval(e)}$
%   \quad
%   $\inference[{\sc WriteOrd}] {x \neq y \quad e \in \W_{|y} \\
%   x \detval{\sigma}_{\tid(e)} v \quad
%   m = \sigma.\last(y)}
%   {x \hbo{\sigma'} y}$
%   \quad
%   $\inference[{\sc NoModOrd}] {e \notin \W_{|x} \cup \W_{|y} \\ x \hbo{\sigma} y}
%   {x \hbo{\sigma'} y}$

%   \caption{Rules for determinate-value and variable-order assertions.
%     We assume $\sigma, m, e, \sigma'$ satisfy
%     $(\_,\sigma) \ltsArrow{m, e}_{\Pad} (\_,\sigma')$. }

%   \label{fig:det-val-ind-appendix}
% \end{figure*}

\smallskip\noindent {\bf \reflem{lem:det-val-read}. }(Determinate-Value Read) {\em
For any {\sc Read} or {\sc RMW} transition
$(P,\sigma) \ltsArrow{m, e}_{\Pad} (P',\sigma')$, if
$\loc(e) \detval{\sigma}_{\tid(e)} v$, then $\rdval(e) = v$.
}
\begin{proof}
  By the definition of $\loc(e) \detval{\sigma}_{\tid(e)} v$, we know
  $m = \sigma.\last(x)$. Both the {\sc Read} or {\sc RMW} transitions
  stipulate that the value read is the $\rdval(e) = \wrval(m) = v$.
\end{proof}

\smallskip\noindent {\bf \reflem{lem:ow-agreement}. }(Determinate-Value Agreement) {\em
For any threads $t, t'$ location $x$, and values $v, v'$,
if $x \detval{\sigma}_{t} v$ and $x \detval{\sigma}_{t'} v'$
then $v = v'$, and thus $t$ and $t'$ agree on the value of $x$.
}
\begin{proof}
  By $x \detval{\sigma}_{t} v$ and $x \detval{\sigma}_{t'} v'$, we
  have $\OW_{\sigma}(t) = \OW_{\sigma}(t') = \{\sigma.\last(x)\}$, and
  thus $v = v'$.
\end{proof}

% \begin{lemma}[Last Modification Transition]
% \label{lem:last-mod-trans}
% For any
%   transition $(P,\sigma) \ltsArrow{m, e}_\Pad (P',\sigma')$,
%   $m = \sigma.\last(\loc(e))$ if either of the following conditions
%   hold: 
% \begin{enumerate}
% \item \label{last-mod-trans:detval} $\loc(e) \detval{\sigma}_{\tid(e)} v$, for
%   some value $v$, or 
% \item  \label{last-mod-trans:up-only}  $\loc(e)$ is an update only location in $\sigma$.
% \end{enumerate}
% \end{lemma}

\smallskip\noindent {\bf \reflem{lem:last-mod-trans}.} {\em Let $t = \tid(e)$ and
  $x = \var(e)$ for some event $e$. For any reachable transition
  $(P,\sigma) \ltsArrow{m, e}_\Pad (P',\sigma')$,
  $m = \sigma.\last(x)$ if either of the following conditions hold:
\begin{enumerate}
\item \label{last-mod-trans:detval} $x \detval{\sigma}_{t} v$, for
  some value $v$, or 
\item  \label{last-mod-trans:up-only}  $x$ is an update only location in $\sigma$.
\end{enumerate}}
\begin{proof} If property \ref{last-mod-trans:detval} holds, then
  $\OW_\sigma(t)_{|x} = \{\sigma.\last(x)\}$, and thus
  $m = \sigma.\last(x)$. If property \ref{last-mod-trans:up-only}
  holds, then every modification to $x$ is covered, except the
  last. Thus, because $m$ is not covered, $m = \sigma.\last(x)$.
\end{proof}

\subsection{Soundness of determinate-value and variable-order assertions}

For simplicity, we copy \reffig{fig:det-val-ind} as \reffig{fig:det-val-ind-appendix}. 
We refer to the set in condition \refeq{dva:3} of Definition~\ref{def:det-val} as the {\em happens-before cone} of $t$ in $\sigma$, and hence define: 
\[
 \sigma.\hbc(t) = I_\sigma \cup \{e \mid \exists e'.\  \tid(e) = t \wedge (e, e') \in \sigma.\lthb^?\}
\]

\begin{lemma} {\sc Init} is valid. 
\end{lemma}
\begin{proof} We have
  $\sigma_0 = ((I, \emptyset), \emptyset, \emptyset)$. We have three
  sub proofs
  \begin{enumerate}
  \item Since $\ltmo = \emptyset$, we have $\OW_{\sigma_0}(t) = I$,
    and also $|I_{|x}| = 1$ and hence
    $I_{|x} = \{\sigma_0.\last(x)\} = \OW_{\sigma_0}(t)_{|x}$.
  \item Trivial by definition.
  \item Immediate since $\sigma.\last(x) \in I$.
  \end{enumerate}
\end{proof}

\begin{lemma}[Establish Determinate-Values]For any reachable transition
  $(P,\sigma) \ltsArrow{m, e}_\Pad (P',\sigma')$, the rules {\sc NoMod}, {\sc ModLast}, {\sc Transfer} and {\sc AcqRead} from \reffig{fig:det-val-ind-appendix} hold. 
\end{lemma}
\begin{proof}\ 

\noindent {\sc NoMod}. This is easy to check since
    $\sigma.\last(x) = \sigma'.\last(x)$ and
    $\OW_{\sigma}(t)_{|x} = \OW_{\sigma'}(t)_{|x}$.

\medskip\noindent{\sc ModLast}. Since $m = \sigma.\last(x)$, the new
    modification $e$ is added to the end of $\ltmo_{|x}$, so that
    $\sigma'.\last(x) = e$. Because $e \in \OW_{\sigma'}(t)_{|x}$ and
    $e$ is $\ltmo$-after every other modification in
    $\sigma'.\ltmo_{|x}$, $\OW_{\sigma'}(t)_{|x}= \{e\}$. Finally, because
    $\tid(e) = t$, $e \in \{e \mid \exists e'.\  \tid(e) = t \wedge (e, e') \in \sigma.\lthb^?\}$.

\medskip\noindent{\sc Transfer}. 
First note that because $x \hbo{\sigma} y$ we have
    \begin{equation}
        (\sigma.\last(x), \sigma.\last(y)) \in \sigma.\lthb\ . \label{eq:transfer-1}
    \end{equation}
    Because $\lthb$ is irreflexive, we also have $x \neq y$, and thus by {\sc NoMod}: 
    \begin{equation}
        \sigma'.\last(x) = \sigma.\last(x)\ . \label{eq:transfer-prop}
    \end{equation}
    By \refeq{eq:transfer-1} and \refeq{eq:transfer-prop}, we have $(\sigma'.\last(x), \sigma.\last(y)) \in \sigma.\lthb$. Moreover, 
    \begin{align*}
        & (\sigma'.\last(x), \sigma.\last(y)) \in \sigma.\lthb 
        \\
      {}\Rightarrow{} & (\sigma'.\last(x), \sigma.\last(y)) \in \sigma'.\lthb
      && \text{because $\sigma.\lthb \subseteq \sigma'.\lthb$}
      \\
      {} \Rightarrow{} & 
        \begin{array}[t]{@{}l@{}}
        (\sigma.\last(y), e) \in \sigma'.\lthb \land {} \\
        (\sigma'.\last(x), e) \in \sigma'.\lthb 
      \end{array}   
      && 
      \text{\begin{tabular}[t]{@{}l@{}}
        assumption \\
        $(\sigma.\last(y), e) \in \ltsw$\end{tabular}}
    \end{align*}
Therefore, \begin{equation*}
      \sigma'.\last(x) \in \sigma'.\hbc(\tid(e))
    \end{equation*}
    This proves the third property of the determinate-value assertion.
    We prove the two remaining properties of the determinate-value
    assertion:
  \begin{itemize}
  \item  %$\sigma'.\OW(t)_{|x} = \{\sigma'.\last(x)\}$. 
    Since $\sigma'.\last(x) \in \sigma'.\hbc(\tid(e))$, we have
        $\sigma'.\last(x) \in \observedWrites_{\sigma}(t)$,
  and therefore $\sigma'.\OW(t)_{|x} = \{\sigma'.\last(x)\}$.
    
  \item %$\wrval(\sigma'.\last(x)) = v$. 
  By \refeq{eq:transfer-prop},
    $\wrval(\sigma'.\last(x)) = v$ iff $\wrval(\sigma.\last(x)) = v$,
    which is true because $x \detval{\sigma}_{t'} v$.
  \end{itemize}
  
\medskip\noindent{\sc AcqRd}. 
We know $\sigma'.\ltmo_{|x} = \sigma.\ltmo_{|x}$ and thus
  $\sigma'.\last(x) = \sigma.\last(x)$. Therefore by the assumption $m = \sigma.\last(x)$, we have $m = \sigma'.\last(x)$.
  Because $m \in \observedWrites_{\sigma'}(t)$ and $m$ is maximal in $\sigma'.\ltmo_{|x}$ we have
  $\sigma'.\OW(t)_{|x} = \{m\} = \{\sigma'.\last(x)\}$ by definition of $\OW_{\sigma'}$. The fact that $\rdval(e) = \wrval(m)$ follows from the premises of rules {\sc Read} and {\sc RMW}. Finally, we have $(m, e) \in \ltsw \subseteq \sigma'.\lthb$
  thus $\sigma'.\last(x) \in \sigma.\hbc(t)$. % in \{e \mid \exists e'.\  \tid(e) = t \wedge (e, e') \in \sigma'.\lthb^?\}$.
\end{proof}

\begin{lemma}[Establish Location-Order]
For any reachable transition
  $(P,\sigma) \ltsArrow{m, e}_\Pad (P',\sigma')$, the rules {\sc WriteOrd} and {\sc NoModOrd} hold. 
\end{lemma}

\begin{proof}\ 

\noindent  {\sc WriteOrd}. Note that $\sigma.\last(x) = \sigma'.\last(x)$ since $x \noteq y$ and $e \in \W_{|y}$. 
By $x \detval{\sigma}_{\tid(e)} v$, we have $\sigma.\last(x) \in \sigma.\hbc(\tid(e)) $. Expanding the definition of $\hbc$ and reformulating slightly, we see that
\[
\sigma.\hbc(\tid(e)) = 
\begin{array}[t]{@{}l@{}}
    I_\sigma \cup \{e' \mid \tid(e') = \tid(e)\} \cup {}\\
    \{e' \mid \exists e''.\  \tid(e'') = \tid(e) \wedge (e', e'') \in \sigma.\lthb\}
\end{array}
\]
Thus, there are three cases to consider. 
\noindent $1.$ 
$\sigma.\last(x) \in I_{\sigma}$. In this case, we get $x \hbo{\sigma'} y$ since we will have $$(\sigma'.\last(x), \sigma'.\last(y)) \in \sigma'.\ltsb \subseteq \sigma'.\lthb.$$ 

\smallskip\noindent $2.$ 
$\tid(\sigma.\last(x)) = \tid(e)$. By transition rules {\sc Write} and {\sc RMW} of our
operational semantics, $(\sigma.\last(x), e) \in \sigma'.\ltsb$ and hence
$(\sigma'.\last(x), e) \in \sigma'.\lthb$, or equivalently $x \hbo{\sigma'} y$.

\smallskip\noindent $3.$ There exists an $e'$ with $\tid(e') = \tid(e)$ and
$(\sigma.\last(x), e') \in \sigma.\lthb$. Because $\sigma'.\last(x) = \sigma.\last(x)$
and $\sigma.\lthb \subseteq \sigma'.\lthb$ we have $(\sigma'.\last(x), e') \in \sigma'.\lthb$.
By the modification transitions of our operational semantics we also have
$(e', e) \in \sigma'.\ltsb$, and thus (putting the two together) we have
$(\sigma'.\last(x), e) \in \sigma'.\lthb$ as required.

\medskip \noindent {\sc NoModOrd}. This is easy to check since because $e$ is not a modification
of $x$ or $y$ we have $\sigma'.\last(x) = \sigma.\last(x)$ and $\sigma'.\last(y) = \sigma.\last(y)$.
Now, because $(\sigma.\last(x), \sigma.\last(y)) \in \sigma.\lthb$
(the content of $x \hbo{\sigma} y$) and the fact that $\sigma.\lthb \subseteq \sigma'.\lthb$,
it follows that $(\sigma'.\last(x), \sigma'.\last(y)) \in \sigma'.\lthb$, or
equivalently $x \hbo{\sigma'} y$.

\medskip\noindent {\sc UpdOrd}. There are two cases to consider. First, assume
$m \neq \sigma.\last(y)$. In this case, $\sigma'.\last(y) = \sigma.\last(y)$.
Because $x \hbo{\sigma} y$, we have $(\sigma.\last(x), \sigma.\last(y)) \in \sigma.\lthb \subseteq \sigma'.\lthb$,
and thus $(\sigma.\last(x), \sigma'.\last(y)) \in \sigma'.\lthb$. Because $x \hbo{\sigma} y$ (by the irreflexivity
of $\lthb$), $x$ and $y$ are
distinct variables and thus $e \notin \W_{|x}$. Therefore, $\sigma'.\last(x) = \sigma.\last(x)$,
so $(\sigma.\last(x), \sigma'.\last(y)) \in \sigma'.\lthb$ as required.

Second, assume $m = \sigma.\last(y)$. In this case $\sigma'.\last(y) = e$. Furthermore, because
$m \in \W_{R|y}$ and $e$ is an update (which is acquiring) we have $(m, e) \in \ltsw$ and
therefore $(m, e) \in \sigma'.\lthb$. Because $x \hbo{\sigma} y$ we have
$(\sigma.\last(x), \sigma.\last(y)) \in \sigma.\lthb \subseteq \sigma'.\lthb$ and
thus, $(\sigma.\last(x), m) \in \sigma'.\lthb$ so $(\sigma.\last(x), e) \in \sigma'.\lthb$
by transitivity. Finally, because $\sigma'.\last(x) = \sigma.\last(x)$ and $\sigma'.\last(y) = e$
we have $(\sigma'.\last(x), \sigma'.\last(y)) \in \sigma'.\lthb$ as required.
\end{proof}

\newcommand{\irrefl}{{\it irrefl}}

\section{Relationship with Canonical C11}
\label{appendix-relationship}
In this section, we describe the relationship between
the version of the C11 semantics given in Section \ref{sec:valid-oper-semant}
and that of \cite{DBLP:conf/popl/BattyDW16}, on which it
is closely based. The semantics of \cite{DBLP:conf/popl/BattyDW16}
uses a notion of candidate execution as described below. We focus
on the relationship between our notion of {\em validity} (Definition \ref{def:legal-execution})
of this paper, and their notion of {\em consistency},
which we call {\em canonical consistency} in this appendix.
We prove that, for any candidate execution, validity
{\em without the {\sc NoThinAir} axiom} and a version of canonical
consistency (described below) are equivalent.

Batty et al \cite{DBLP:conf/popl/BattyDW16} use a notion of {\em candidate execution}
(Definition 7, \cite{DBLP:conf/popl/BattyDW16}),
which gives certain well-formedness conditions on executions.
For the purposes of this appendix, we define {\em candidate
executions} as follows:
\begin{definition}[Candidate Execution]
A tuple $((D, \ltsb), \ltrf, \ltmo)$ is a {\em candidate execution}
if it satisfies the conjunction of the conditions {\sc RF-Complete},
{\sc MO-Valid} and {\sc SB-Total} of Definition \ref{def:legal-execution}
in our current paper.
\end{definition}
Minor variations in presentation prevent us from claiming that
the definition just given is strictly equivalent to Definition 7 of
\cite{DBLP:conf/popl/BattyDW16}.
Principally, \cite{DBLP:conf/popl/BattyDW16} employs an equivalence
relation to determine when two operations are on the same thread,
whereas we index operations with a thread identifier.
Another difference is that Batty et al. \cite{DBLP:conf/popl/BattyDW16}
define the $\lthb$ relation such that initialising writes are
$\lthb$-prior to all other events, whereas we stipulate that
initialising writes are $\ltsb$-prior to all other events
(thus ensuring the $\lthb$-ordering indirectly).
With these caveats aside, the definition of candidate execution
given here is essentially the same as that of \cite{DBLP:conf/popl/BattyDW16}.

% {\sc RF-Complete} equiv WfRf in Def 7 of \cite{DBLP:conf/popl/BattyDW16}.
% {\sc MO-Valid} is equivalent (for non-atomic writes) to WfMo.
% {\sc SB-Total} is roughly equivalent to the last clause of Definition 5
% of \cite{DBLP:conf/popl/BattyDW16}.

Let $(D, \ltsb, \ltrf, \ltmo)$ be a candidate execution.
As is true for all versions of the C11 memory model,
canonical consistency is defined in terms of the happens-before
relation, which in turn is defined in terms
of the synchronises-with relation. The synchronises-with
relation of \cite{DBLP:conf/popl/BattyDW16}, which we 
call {\em canonical synchronises-with} and denote
by $\ltsw_C$ is slightly larger
than our definition
\[
\ltsw \subseteq \ltsw_C
\]
The extra edges in $\ltsw_C$ relate to the
so-called {\em release sequences}, which we have ignored
in our presentation. The effect of this relaxation is
that our version defines a weaker semantics, with more valid executions.

The happens-before relation in \cite{DBLP:conf/popl/BattyDW16},
which we call {\em canonical happens-before} and denote $\lthb_C$,
is defined as follows
\[
\lthb_C = (\ltsb \cup (I \times \neg I) \cup \ltsw_C)^+
\]
where $\neg I$ is the complement of the set of initialising
writes. In our version of the semantics, $I \times \neg I \subseteq \ltsb$,
and thus $\ltsb \cup (I \times \neg I) = \ltsb$ so
\[
\lthb_C = (\ltsb \cup \ltsw_C)^+
\] 
similar to our definition. Thus, because $\ltsw \subseteq \ltsw_C$,
$\lthb \subseteq \lthb_C$.

We now present the definition of consistency given in
\cite{DBLP:conf/popl/BattyDW16} as it relates to 
the RAR fragment.
\begin{definition}[Canonical RAR Consistency]
A candidate execution $\bbD = (D, \ltsb, \ltrf, \ltmo)$ is 
canonically consistent if all the following conditions hold
\begin{align*}
    & \irrefl(\lthb_C) \tag{\sc HB-C} \\
    & \irrefl((\ltrf^{-1})^? ; \ltmo ; \ltrf^? ; \lthb_C) \tag{\sc COH-C} \\
    & \irrefl(\ltrf ; \lthb_C) \tag{\sc RF-C} \\
    & \irrefl(\ltrf \cup (\ltmo ;\ltmo ; \ltrf^{-1}) \cup (\ltmo ; \ltrf)) \tag{\sc UPD-C} 
\end{align*}
where $\lthb_C$ is defined from $\bbD$ as above.
\end{definition}

To account for the fact that $\ltsw \subseteq \ltsw_C$, and thus
$\lthb \subseteq \lthb_C$, we give a slightly weaker notion of 
canonical consistency, called {\em weak canonical RAR consistency},
which we prove equivalent to our notion
of validity. This
weaker condition is obtained from the stronger by replacing
$\lthb_C$ by $\lthb$. Also, to simplify presentation of
the proof, we split the condition \mbox{{\sc RF-C}} into two
conditions: one called {\sc RF}, and one called {\sc RFI}
that explicitly requires the irreflexivity of the
$\ltrf$ relation. This second change is purely presentational,
and does not change the strength of the semantics.
\begin{definition}[Weak Canonical RAR Consistency]
\label{def:weak-canon}
A candidate execution $\bbD = (D, \ltsb, \ltrf, \ltmo)$ is 
canonically consistent if all the following conditions hold
\begin{align*}
    & \irrefl(\lthb) \tag{\sc HB} \\
    & \irrefl((\ltrf^{-1})^? ; \ltmo ; \ltrf^? ; \lthb) \tag{\sc COH} \\
    & \irrefl(\ltrf ; \lthb) \tag{\sc RF} \\
    & \irrefl(\ltrf) \tag{\sc RFI} \\
    & \irrefl((\ltmo ;\ltmo ; \ltrf^{-1}) \cup (\ltmo ; \ltrf)) \tag{\sc UPD} 
\end{align*}
where $\lthb$ is defined from $\bbD$ as usual.
\end{definition}
As we shall see, the validity condition $\irrefl(\lteco^? ; \lthb)$
(as used in the coherence condition of Definition \ref{def:legal-execution}
in our paper)
captures the collective effect
of conditions HB, COH and RF. The condition UPD, which we sometimes
call {\em update atomicity}, requires that each update
appears in $\ltmo$-order immediately after the
write that the update reads from. As we shall see, the validity
condition $\irrefl(\lteco)$ implies update atomicity, and for any
candidate execution, the update atomicity property implies
$\irrefl(\lteco)$.

The following lemma follows easily from the fact that $\lthb \subseteq \lthb_C$.
\begin{lemma}
For any candidate execution $$\bbD = ((D, \ltsb), \ltrf, \ltmo),$$ if $\bbD$
is canonical consistent, then it is weakly canonical consistent.
\end{lemma}
From now on, we consider only weak canonical consistency. Thus, when we refer
to properties {\sc HB}, {\sc COH}, {\sc RF}, and {\sc UPD} we mean those of {\em weak} canonical
consistency.

For the remainder of the section, we work towards proving the
following theorem
\begin{theorem}
For any candidate execution $$\bbD = ((D, \ltsb), \ltrf, \ltmo),$$ 
$\bbD$ is
weakly canonical consistent iff $\bbD$ satisfies
{\sc Coherence} of Definition \ref{def:legal-execution} on page \pageref{def:legal-execution}.
\end{theorem}
As we shall see, much of our proof is about reformulating the various relations and axioms
that make-up the canonical memory model.

% One minor challenge that we encounter is in working with the definition of
% {\em from-reads} $\ltrf = (\ltrf^{-1} ; \ltmo) - \mathit{Id}$, where the subtraction
% of the $\mathit{Id}$ makes working with $\ltrf$ somewhat more difficult.
% We would prefer to deal with the relation
% \begin{align}
%     \ltfr' &= \ltrf^{-1} ; \ltmo
% \end{align}
% (Note that now $\ltfr = \ltfr' - \mathit{Id}$.)
% The following lemma enables us to do so in important cases.
% \begin{lemma}[Simplifying $\ltfr$]
% For any candidate execution $\bbD = ((D, \ltsb), \ltrf, \ltmo)$
% $\ltfr ; \ltmo = \ltrf^{-1} ; \ltmo ; \ltmo = \ltfr' ; \ltmo$.
% \end{lemma}
% \begin{proof}
% $\ltfr ; \ltmo = (\ltfr' - \mathit{Id}) ; \ltmo$ so $\ltfr ; \ltmo \subseteq \ltfr' ; \ltmo$.
% It remains to show that $\ltfr' ; \ltmo \subseteq \ltfr ; \ltmo$. Let
% $(x, w) \in \ltrf^{-1}$, $(w, y) \in \ltmo$ and $(y, z) \in \ltmo$ so that $(x, z) \in \ltfr' ; \ltmo$.
% If $x \neq y$ then $(x, y) \in \ltfr$ and then $(x, z)\in \ltfr ; \ltmo$, as required.
% On the other hand, if $x = y$ then $x$ is an update. To see this observe that
% it has both an outgoing $\ltrf^{-1}$ edge, so it is a read, and an incoming $\ltmo$
% edge, so it is a write. Now, because $(w, x) \in \ltrf$, Lemma
% \ref{lem:supporting-rel-props}\ref{lem:supporting-rel-props:2} gives us $(w, x) \in \ltmo$,
% and so (by transitivity of $\ltmo$) $(w, z) \in \ltmo$.

% we have $(x, w) \in \ltrf^{-1}$ and $(x, z) \in \ltmo$

% $(x, x) \in \ltfr'$ and $(x, z) \in \ltmo$ so $(x, z)\in \ltfr ; \ltmo$
% as required.
% \end{proof}

The following lemma provides a more convenient form for the
{\sc UPD} property.
\begin{lemma}
\label{lem:upd-rewrite}
For any candidate execution $$\bbD = ((D, \ltsb), \ltrf, \ltmo),$$
the {\sc UPD} condition (that is,
$\irrefl((\ltmo ;\ltmo ; \ltrf^{-1}) \cup (\ltmo ; \ltrf))$) is equivalent to
$\irrefl(\ltfr ;\ltmo) \wedge \irrefl(\ltrf ; \ltmo)$.
\end{lemma}
\begin{proof}
First note that for any relations $r, s$ we have both
\begin{align*}
\irrefl(r ; s) & \iff \irrefl(s ; r) \\
\irrefl(r \cup s) & \iff \irrefl(r) \wedge \irrefl(s). 
\end{align*}
Using these equivalences,
{\sc UPD} is equivalent to $$\irrefl(\ltrf^{-1}; \ltmo ;\ltmo) \wedge \irrefl(\ltrf ; \ltmo).$$
It remains to show that $\irrefl(\ltrf^{-1}; \ltmo ;\ltmo)$
is equivalent to $\irrefl(\ltfr ;\ltmo)$. Because $\ltfr \subseteq \ltrf^{-1}; \ltmo$,
we have $$\ltfr ;\ltmo \subseteq \ltrf^{-1}; \ltmo ;\ltmo$$ and thus if 
$\irrefl(\ltrf^{-1}; \ltmo ;\ltmo)$ then $\irrefl(\ltfr ;\ltmo)$.

Finally, we show that if there is a cycle in $\ltrf^{-1}; \ltmo ;\ltmo$ then there
is also one in $\ltfr ;\ltmo$. Assume that $(x, x) \in \ltrf^{-1}; \ltmo ;\ltmo$.
Then there is some $y$ such that $(x, y) \in \ltrf^{-1}; \ltmo$
and $(y, x) \in \ltmo$. There are two cases to consider. In the first case,
$y = x$. But this is impossible because then we would have $(x, x) \in \ltmo$,
contrary to the irreflexivity of $\ltmo$. In the second case, $x \neq y$,
but then $(x, y) \in (\ltrf^{-1}; \ltmo) - \mathit{Id} = \ltfr$ so
$(x, x) \in \ltfr ;\ltmo$ and we are done.
\end{proof}

The first lemma says that each update operation can only read from its
immediate $\ltmo$ predecessor.
\begin{lemma}[Update orderings]
\label{lem:supporting-rel-props}
For any candidate execution $(D, \ltsb, \ltrf, \ltmo)$, satisfying {\sc UPD}
the following properties hold for any update $u \in D$ and event $x \in D$:
\begin{enumerate}[label=\roman*]
\item 
$$\label{lem:supporting-rel-props:1} (u, x) \in \ltfr \implies (u, x) \in \ltmo$$ % cycle in fr ; mo
\item 
$$\label{lem:supporting-rel-props:2} (x, u) \in \ltrf \implies (x, u) \in \ltmo$$ % cycle in rf ; mo
\end{enumerate}
\end{lemma}
\begin{proof}
Note first that $\ltmo$ must order $u$ and $x$
(in some direction). This is because $\var(u) = \var(x)$, $u$ is a modification,
$x$ is a modification (because it either has an incoming $\ltfr$ edge or an
outgoing $\ltrf$ edge) and $\ltmo$ is total
over modifications to the same location. Therefore, it is sufficient to derive a contradiction
from the assumption that $\ltmo$ orders the two operations the ``wrong'' way.

Assume for Property \ref{lem:supporting-rel-props:1} that
$(u, x) \in \ltfr$ and $(x, u) \in \ltmo$. But then $(u, u) \in \ltfr ; \ltmo$ contrary
to the {\sc UPD} property, as formulated in Lemma \ref{lem:upd-rewrite}.

Assume for Property \ref{lem:supporting-rel-props:2} that
$(x, u) \in \ltrf$ and $(u, x) \in \ltmo$. But then $(x, x) \in \ltrf ; \ltmo$ contrary
to the {\sc UPD} property, as formulated in Lemma \ref{lem:upd-rewrite}.
\end{proof}

We next need some properties about the structure of $\lteco$.
\begin{lemma}[Coherence inclusions]
\label{lem:rel-inclusions}
For any candidate execution $(D, \ltsb, \ltrf, \ltmo)$, that satisfies the {\sc UPD} property
the following inclusions hold:
\begin{enumerate}[label=\roman*]
\item \label{lem:rel-inclusions1} $\ltrf ; \ltfr \subseteq \ltmo$
\item \label{lem:rel-inclusions2} $\ltrf ; \ltmo \subseteq \ltmo$
\item \label{lem:rel-inclusions3} $\ltrf ; \ltrf \subseteq \ltmo ; \ltrf$
\item \label{lem:rel-inclusions4} $\ltmo ; \ltfr \subseteq \ltmo$
\item \label{lem:rel-inclusions6} $\ltfr ; \ltmo \subseteq \ltfr$
\item \label{lem:rel-inclusions7} $\ltfr ; \ltfr \subseteq \ltfr$
\end{enumerate}
\end{lemma}
\begin{proof}
\begin{itemize}
\item \ref{lem:rel-inclusions1}) Consider $(x, y) \in \ltrf$ and $(y, z) \in \ltfr$.
Because $\ltrf$ is one-to-many, $\ltrf^{-1}(y) = x$. Because $(y, z) \in \ltfr$,
$(\ltrf^{-1}(y), z) \in \ltmo$. Therefore, $(x, z) \in \ltmo$ as required.

\item \ref{lem:rel-inclusions2}) Consider $(x, y) \in \ltrf$ and $(y, z) \in \ltmo$.
Because $y$ has an incoming $\ltrf$ edge it is a read, because it
has an outgoing $\ltmo$ edge, it is a modification, and so $y$ is an update.
Thus, by Lemma \ref{lem:supporting-rel-props}\ref{lem:supporting-rel-props:2},
$(x, y) \in \ltmo$ and then the result follows by transitivity.

\item \ref{lem:rel-inclusions3}) Consider $(x, y) \in \ltrf$ and $(y, z) \in \ltrf$.
Because $y$ has both incoming and outgoing $\ltrf$ edges, it is an update.
Thus, by Lemma \ref{lem:supporting-rel-props}\ref{lem:supporting-rel-props:2},
$(x, y) \in \ltmo$ and then the result is immediate.

\item \ref{lem:rel-inclusions4}) Consider $(x, y) \in \ltmo$ and $(y, z) \in \ltfr$.
Because $y$ has an incoming $\ltmo$ edge it is a modification, because it
has an outgoing $\ltfr$ edge, it is a read, and thus $y$ is an update.
Thus, by Lemma \ref{lem:supporting-rel-props}\ref{lem:supporting-rel-props:1},
$(y, z) \in \ltmo$ and now $(x, z) \in \ltmo$ by transitivity.

\item \ref{lem:rel-inclusions6}) Let $(x, z) \in \ltfr ; \ltmo$. Thus, there
is some $y \neq x$ such that $(x, y) \in \ltrf^{-1} ; \ltmo$ and $(y, z) \in \ltmo$.
Let $w$ be unique such that $(w, x) \in \ltrf$ and so $(w, y) \in \ltmo$.
By transitivity of $\ltmo$ we have $(w, z) \in \ltmo$ and thus $(x, z) \in \ltrf^{-1} ; \ltmo$.
It remains to show that $x \neq z$. Assume otherwise. Then $x$ is an update
(because it has both an incoming $\ltrf$ edge and an incoming $\ltmo$ edge).
Now, by Lemma \ref{lem:supporting-rel-props}\ref{lem:supporting-rel-props:1},
because $(x, y) \in \ltfr$, $(x, y) \in \ltmo$ and
thus $(x, z) \in \ltmo$. But now $x \neq y$ by the irreflexivity of $\ltmo$.

\item \ref{lem:rel-inclusions7}) Consider $(x, y) \in \ltfr$ and $(y, z) \in \ltfr$.
Because $y$ has an incoming $\ltfr$ edge it is a modification, and because
$y$ has an outgoing $\ltfr$ edge it is a read. Thus, $y$ is an update
and by Lemma \ref{lem:supporting-rel-props}\ref{lem:supporting-rel-props:1},
$(y, z) \in \ltmo$. So now $(x, z) \in \ltfr ; \ltmo$ so by Property
\ref{lem:rel-inclusions6} of this Lemma, $(x, z) \in \ltfr$.
\end{itemize}
\end{proof}

The next lemma presents a ``closed-form'' for the $\lteco$ relation,
in which $\lteco$ is defined without use of a transitive closure, providing
a simple set of cases that must be considered when analysing
the relation. This is inspired by a similar expression in
\cite{DBLP:conf/pldi/LahavVKHD17}.
\begin{lemma}[$\lteco$ cases]
\label{lem:eco-cases}
For any semi-consistent execution $(D, \ltsb, \ltrf, \ltmo)$, with update atomicity
\[
\lteco = \ltrf \cup \ltmo \cup \ltfr \cup  (\ltmo;\ltrf) \cup (\ltfr;\ltrf)
\]
\end{lemma}
\begin{proof}
Let
\[
\lteco' = \ltrf \cup \ltmo \cup \ltfr \cup  (\ltmo;\ltrf) \cup (\ltfr;\ltrf)
\]
We show that $\lteco' = \lteco$.

Recall that $\lteco = (\ltfr \cup \ltmo \cup \ltrf)^+$. 
It is easy to see that $\lteco' \subseteq \lteco$
(as each option in the union defining $\lteco'$ is included in one or
two steps of $\lteco$).

We prove that $\lteco \subseteq \lteco'$ by induction.

Let $p$ be a
(nonempty) path through the transitive closure $\lteco$.
Thus for each $i$ such that $i+1 < |p|$ (where $|p|$ is the length of
$p$) $(p_i, p_{i+1}) \in \ltfr \cup \ltmo \cup \ltrf$ (indexing from $0$).
We prove by induction on the length of $p$ that
$(p_0, p_{|p|-1}) \in \lteco'$,
which is sufficient
to show that $\lteco \subseteq \lteco'$.
For the base case, $p$ contains two elements, $p_0$ and $p_1$, so we must
prove that $(p_0, p_1) \in \lteco'$.
But this follows from the fact that
$$
\ltfr \cup \ltmo \cup \ltrf \subseteq\ltrf \cup \ltmo \cup \ltfr \cup  (\ltmo;\ltrf) \cup (\ltfr;\ltrf)
$$
which is clear by inspection. For the induction, assume there is some $p'$ and
$x$ such that $p' = p ; \langle x \rangle$ and both
\begin{align*}
(p_0, p_{|p|-1}) &\in \lteco' \\
(p_{|p|-1}, x) &\in \ltfr \cup \ltmo \cup \ltrf
\end{align*}
We must prove that $(p_0, x) \in \lteco'$. It is sufficient to show that
\[
\lteco' ; (\ltfr \cup \ltmo \cup \ltrf) \subseteq \lteco'
\]
But by distributivity of $;$ over $\cup$, this is equivalent to
\[
(\lteco' ; \ltfr) \cup (\lteco' ; \ltmo) \cup (\lteco' ; \ltrf) \subseteq \lteco'
\]
Expanding the definition of $\lteco'$ and applying distributivity once again
we obtain 15 cases to check: five options in the union defining $\lteco'$
combined with each of the three relations $\ltfr, \ltmo, \ltrf$.
The cases, and their proofs are as follows:
\begin{itemize}
\item $\ltrf ; \ltfr \subseteq \lteco'$. But
\begin{align*}
\ltrf ; \ltfr &\subseteq \ltmo & \text{by Lemma \ref{lem:rel-inclusions}\ref{lem:rel-inclusions1}}\\
                  &\subseteq \lteco' & \text{by definition of $\lteco'$}
\end{align*}

\item $\ltrf ; \ltmo \subseteq \lteco'$. But
\begin{align*}
\ltrf ; \ltmo &\subseteq \ltmo & \text{by Lemma \ref{lem:rel-inclusions}\ref{lem:rel-inclusions2}}\\
                  &\subseteq \lteco' & \text{by definition of $\lteco'$}
\end{align*}

\item $\ltrf ; \ltrf \subseteq \lteco'$. But
\begin{align*}
\ltrf ; \ltrf &\subseteq \ltmo ; \ltrf& \text{by Lemma \ref{lem:rel-inclusions}\ref{lem:rel-inclusions3}}\\
                  &\subseteq \lteco' & \text{by definition of $\lteco'$}
\end{align*}

\item $\ltmo ; \ltfr \subseteq \lteco'$. But
\begin{align*}
\ltmo ; \ltfr &\subseteq \ltmo & \text{by Lemma \ref{lem:rel-inclusions}\ref{lem:rel-inclusions4}}\\
                  &\subseteq \lteco' & \text{by definition of $\lteco'$}
\end{align*}

\item $\ltmo ; \ltmo \subseteq \lteco'$. But
\begin{align*}
\ltmo ; \ltmo &\subseteq \ltmo & \text{by transitivity}\\
                  &\subseteq \lteco' & \text{by definition of $\lteco'$}
\end{align*}

\item $\ltmo ; \ltrf \subseteq \lteco'$. But this is true by definition of $\lteco'$.

\item $\ltfr ; \ltfr \subseteq \lteco'$. But
\begin{align*}
\ltfr ; \ltfr &\subseteq \ltfr & \text{by Lemma \ref{lem:rel-inclusions}\ref{lem:rel-inclusions7}}\\
                  &\subseteq \lteco' & \text{by definition of $\lteco'$}
\end{align*}

\item $\ltfr ; \ltmo \subseteq \lteco'$. But
\begin{align*}
\ltfr ; \ltmo &\subseteq \ltfr & \text{by Lemma \ref{lem:rel-inclusions}\ref{lem:rel-inclusions6}}\\
                  &\subseteq \lteco' & \text{by definition of $\lteco'$}
\end{align*}

\item $\ltfr ; \ltrf \subseteq \lteco'$. But this is true by definition of $\lteco'$.

\item $\ltmo;\ltrf ; \ltfr \subseteq \lteco'$. But
\begin{align*}
\ltmo;\ltrf ; \ltfr &\subseteq \ltmo ; \ltmo & \text{by Lemma \ref{lem:rel-inclusions}\ref{lem:rel-inclusions1}}\\
                  &\subseteq \ltmo & \text{by transitivity}\\
                  &\subseteq \lteco' & \text{by definition of $\lteco'$}
\end{align*}

\item $\ltmo;\ltrf ; \ltmo \subseteq \lteco'$. But
\begin{align*}
\ltmo;\ltrf ; \ltmo &\subseteq \ltmo ; \ltmo & \text{by Lemma \ref{lem:rel-inclusions}\ref{lem:rel-inclusions2}}\\
                  &\subseteq \ltmo & \text{by transitivity}\\
                  &\subseteq \lteco' & \text{by definition of $\lteco'$}
\end{align*}

\item $\ltmo;\ltrf ; \ltrf \subseteq \lteco'$. But
\begin{align*}
\ltmo;\ltrf ; \ltrf &\subseteq \ltmo ; \ltmo ; \ltrf& \text{by Lemma \ref{lem:rel-inclusions}\ref{lem:rel-inclusions3}}\\
                  &\subseteq \ltmo; \ltrf & \text{by transitivity}\\
                  &\subseteq \lteco' & \text{by definition of $\lteco'$}
\end{align*}

\item $\ltfr;\ltrf ; \ltfr \subseteq \lteco'$. But
\begin{align*}
\ltfr;\ltrf ; \ltfr &\subseteq \ltfr ; \ltmo & \text{by Lemma \ref{lem:rel-inclusions}\ref{lem:rel-inclusions1}}\\
                  &\subseteq \ltfr & \text{by Lemma \ref{lem:rel-inclusions}\ref{lem:rel-inclusions6}}\\
                  &\subseteq \lteco' & \text{by definition of $\lteco'$}
\end{align*}

\item $\ltfr;\ltrf ; \ltmo \subseteq \lteco'$. But
\begin{align*}
\ltfr;\ltrf ; \ltmo &\subseteq \ltfr ; \ltmo & \text{by Lemma \ref{lem:rel-inclusions}\ref{lem:rel-inclusions2}}\\
                  &\subseteq \ltfr & \text{by Lemma \ref{lem:rel-inclusions}\ref{lem:rel-inclusions6}}\\
                  &\subseteq \lteco' & \text{by definition of $\lteco'$}
\end{align*}

\item $\ltfr;\ltrf ; \ltrf \subseteq \lteco'$. But
\begin{align*}
\ltfr;\ltrf ; \ltrf &\subseteq \ltfr;\ltmo ; \ltrf & \text{by Lemma \ref{lem:rel-inclusions}\ref{lem:rel-inclusions3}}\\
                  &\subseteq \ltfr ; \ltrf& \text{by Lemma \ref{lem:rel-inclusions}\ref{lem:rel-inclusions6}}\\
                  &\subseteq \lteco' & \text{by definition of $\lteco'$}
\end{align*}
\end{itemize}
This completes our proof.
\end{proof}

% \begin{lemma}
% For any candidate execution $\bbD = ((D, \ltsb), \ltrf, \ltmo)$, $\bbD$ satisfies
% UPD iff it satisfies the following property
% \[
% \label{prop:update-atom-reform}
% (w, u) \in \ltrf \iff 
% \begin{array}[t]{@{}l@{}}
%     (w, u) \in \ltmo \wedge {} \\
%     \forall w'\in D. \neg ((w, w') \in \ltmo \wedge (w', u) \in \ltmo) \\ 
% \end{array}
% \]
% \end{lemma}
% \begin{proof}
% Assume UPD, and assume for a contradiction that \ref{prop:update-atom-reform}
% fails. Then there is some $(w, u) \in \ltrf$.
% \end{proof}

\begin{lemma}[Weak Canonical RAR Consistency implies $\lteco$-irreflexivity]
\label{lem:canon-coher-implies-eco-irrefl}
For a  candidate execution $$\bbD = ((D, \ltsb), \ltrf, \ltmo),$$ if $\bbD$
is weakly canonical consistent then $\bbD$ satisfies $\irrefl(\lteco)$.
\end{lemma}
\begin{proof}
Recall from Lemma \ref{lem:eco-cases} that
\[
\lteco = \ltrf \cup \ltmo \cup \ltfr \cup  (\ltmo;\ltrf) \cup (\ltfr;\ltrf)
\]
Assume for a contradiction that there is some $(x, x) \in \lteco$. There are
five cases to consider: one for each option of the union.
It cannot be that $(x, x) \in \ltrf$, $(x, x) \in \ltfr$, $(x, x) \in \ltmo$ edges,
because all these relations are irreflexive. Thus the pair $(x, x)$ must appear in one of the following:
$\ltmo;\ltrf$ or $\ltfr;\ltrf$. We treat each case separately.

In the first case, we have $(x, x) \in \ltmo;\ltrf$ for some $x \in D$.
The relation $\ltmo;\ltrf$ goes from modifying operations to reading
operations, so again $x$ must be an update. Let $w'$ be the modification
satisfying $(x, w') \in \ltmo$ and $(w', x) \in \ltrf$ (the existence of this
operation is guaranteed by the definition of relational composition). But now,
by Property Lemma \ref{lem:supporting-rel-props}\ref{lem:supporting-rel-props:1},
$(w', x) \in \ltmo$ and thus $(x, x) \in \ltmo$
contrary to the irreflexivity of $\ltmo$.

In the second case, we have $(x, x) \in \ltfr;\ltrf$. Let $w$ be the modification
satisfying $(x, w) \in \ltfr$ and $(w, x) \in \ltrf$ (again, the existence of this
modification is guaranteed by relational composition). Because
$$\ltfr = \ltrf^{-1};\ltmo \setminus Id$$
there is some modification $w'$
satisfying $(w', x) \in \ltrf$, $(w', w) \in \ltmo$ and $w' \neq w$.
But because $\ltrf$ is one-to-many $\ltrf^{-1}(x) = w$ and $\ltrf^{-1}(x) = w'$,
and thus $w=w'$, a contradiction.
This completes our proof.
\end{proof}

\begin{lemma}
\label{lem:rewrite-eco-hb}
For a candidate execution $$\bbD = ((D, \ltsb), \ltrf, \ltmo),$$
$\irrefl(\lteco ; \lthb)$ is equivalent to the conjunction of
{\sc COH} and {\sc RF} (defined in Definition \ref{def:weak-canon}).
\end{lemma}
\begin{proof}
Let $R = (\ltrf^{-1})? ; \ltmo ; \ltrf^?$ so that {\sc COH} is equivalent to
$\irrefl(R ; \lthb)$. Now, note that
\begin{align}
(\ltrf^{-1})? ; \ltmo ; \ltrf^? &= (\ltmo \cup \ltfr) ; (\ltrf \cup \mathit{Id}) \\
    & \qquad \qquad \text{def of $\ltfr$, ref. clos.} \nonumber \\
    &=(\ltmo \cup (\ltmo ; \ltrf) \cup \ltfr \cup (\ltfr ; \ltrf)) \\
    & \qquad \qquad\text{distrib of $;$ over $\cup$}\nonumber
\end{align}
Therefore, recalling from Lemma \ref{lem:eco-cases} that
\[
\lteco = \ltrf \cup \ltmo \cup \ltfr \cup  (\ltmo ; \ltrf) \cup (\ltfr;\ltrf)
\]
we have $\lteco = R \cup \ltrf$. Thus, because $;$ distributes over $\cup$,
$\lteco ; \lthb = (R ; \lthb) \cup (\ltrf ; \lthb)$, and so $\irrefl(\lteco ; \lthb)$ is
equivalent to $\irrefl(R ; \lthb \cup \ltrf ; \lthb)$. But, because
$\irrefl(r \cup s) \iff \irrefl(r) \wedge \irrefl(s)$ for any relations $r, s$, this is 
equivalent to the conjunction of {\sc COH} and {\sc RF}.
\end{proof}

\begin{lemma}[Weak Canonical RAR Consistency implies $\lteco ; \lthb$-irreflexivity]
\label{lem:canon-coher-implies-eco-hb-irrefl}
For a candidate execution $$\bbD = ((D, \ltsb), \ltrf, \ltmo),$$ whenever $\bbD$
is weakly canonical consistent, we have that $\bbD$ satisfies $\irrefl(\lteco^? ; \lthb)$.
\end{lemma}
\begin{proof}
First, note that Property {\sc HB} of weakly canonical consistency
ensures that $\irrefl(\lthb)$, so if there is a cycle in
$\lteco^? ; \lthb$ then there is a cycle in $\lteco ; \lthb$
(so we must actually take an $\lteco$ step). We prove that this later is
impossible. But by Lemma \ref{lem:rewrite-eco-hb}, because $\bbD$ satisfies {\sc COH}
and {\sc RF} we have $\irrefl(\lteco ; \lthb)$ as required.
\end{proof}

\begin{lemma}[Coherence implies canonical coherence]
\label{lem:coher-implies-canon-coher}
For a candidate execution $\bbD = ((D, \ltsb), \ltrf, \ltmo)$, if $\bbD$
satisfies $\irrefl(\lteco^? ; \lthb)$ then $\bbD$ satisfies all
of HB, COH, and RF above.
\end{lemma}
\begin{proof}
Assume  $\irrefl(\lteco^? ; \lthb)$. 
Because $\irrefl(\lteco^? ; \lthb)$ and $\lthb \subseteq \lteco^? ; \lthb$ we have $\irrefl(\lthb)$
as required for {\sc HB}.

Because $\irrefl(\lteco^? ; \lthb)$ and $\lteco ; \lthb \subseteq \lteco^? ; \lthb$ we have
$\irrefl(\lteco ;\lthb)$. By Lemma \ref{lem:rewrite-eco-hb}, this implies
the conjunction of {\sc COH} and {\sc RF}. This completes our proof.

\end{proof}

\begin{lemma}[Coherence implies Update Atomicity]
\label{lem:coher-implies-upd-atom}
For a candidate execution $\bbD = ((D, \ltsb), \ltrf, \ltmo)$, if $\bbD$
satisfies $\irrefl(\lteco^?)$ then $\bbD$ satisfies the
update atomicity property UPD.
\end{lemma}
\begin{proof}
By Lemma \ref{lem:upd-rewrite}, {\sc UPD} is equivalent to
$\irrefl(\ltfr ;\ltmo) \wedge \irrefl(\ltrf ; \ltmo)$.
But $\ltfr ;\ltmo \subseteq \lteco$, so because $\irrefl(\lteco)$
we have $\irrefl(\ltfr ;\ltmo)$. Likewise, $\ltrf ; \ltmo \subseteq \lteco$
so $\irrefl(\ltrf ; \ltmo)$. This is sufficient to prove {\sc UPD}.
\end{proof}

The four lemmas \ref{lem:canon-coher-implies-eco-irrefl},
\ref{lem:canon-coher-implies-eco-hb-irrefl},
\ref{lem:coher-implies-canon-coher} and \ref{lem:coher-implies-canon-coher}
together imply \ref{theor:canon-equiv} can now be used to prove our main
theorem.

\begin{theorem}
\label{theor:canon-equiv}
For any candidate execution $$\bbD = ((D, \ltsb), \ltrf, \ltmo),$$ $\bbD$ is
weakly canonical consistent iff $\bbD$ satisfies
{\SC Coherence} of Definition \ref{def:legal-execution} on page \pageref{def:legal-execution}.
\end{theorem}
\begin{proof}
For the left-to-right direction, see Lemmas \ref{lem:canon-coher-implies-eco-irrefl} and
\ref{lem:canon-coher-implies-eco-hb-irrefl}. For the right-to-left
direction, see Lemmas \ref{lem:coher-implies-canon-coher} and \ref{lem:coher-implies-canon-coher}.
\end{proof}

\newcommand{\vflag}{\mathit{flag}}
\newcommand{\vturn}{\mathit{turn}}
\newcommand{\false}{\mathit{false}}
\newcommand{\true}{\mathit{true}}

% \algrenewcommand\algorithmicindent{0.75em} 
% \begin{varwidth}[t]{0.5\columnwidth}
%   \begin{algorithmic}[1] \small
%     \Thread{1}
%     % \While {true}
%     \State $\flag_1$ := $\True$ ; \label{set-flag}
%     \State $\turn$.\kwswap(2)$^{\sf RA}$ ; \label{swap-turn} 
%     \MyWhile {
%       \begin{tabular}[t]{@{}l@{}}
%         ($\flag_2 \!=\! \True)^{\sf A}$\\
%         ${} \land \turn$ = 2
%       \end{tabular}
%     } \label{busy-wait}
%     \Statex \qquad 
%     \textsf{\textbf{do}} \Skip
%     \EndMyWhile 
%     \State Critical section ; \label{critical-section}
%     \State $\flag_1$ :=$^{\sf R}$ \False ; \label{unset-flag}
%     % \EndWhile
%     \EndThread
%   \end{algorithmic}
% \end{varwidth}
% \quad
% \begin{varwidth}[t]{0.5\columnwidth}
%   \begin{algorithmic}[1] \small
%     \Thread{2}
%     % \While {true}
%     \State $\flag_2$ := \True ; 
%     \State $\turn$.\kwswap(1)$^{\sf RA}$ ; 
%     \MyWhile {
%       \begin{tabular}[t]{@{}l@{}}
%         ($\flag_1 \!= \!\True)^{\sf A}$\\
%         {} ${}\land \turn$ = 1 
%       \end{tabular}
%     }
%     \Statex \qquad 
%     \textsf{\textbf{do}} \Skip
%     \EndMyWhile
%     \State Critical section ;
%     \State $\flag_2$ :=$^{\sf R}$ \False ;
%     % \EndWhile
%     \EndThread
%   \end{algorithmic}
% \end{varwidth}
% \end{algorithm}

\section{Proof of Peterson's algorithm}

A configuration $(P, \sigma)$ is an {\em initial configuration of
Peterson's algorithm} if $\sigma$ is an initial state of our semantics and
the following conditions hold: 
\begin{gather}
P.\pc_t = \ref{set-flag} \label{pete-init:pc}\\
\wrval(\sigma.\last(\vturn)) \in \{1, 2\} \label{pete-init:turn}\\
\wrval(\sigma.\last(\vflag_{t})) = \False \label{pete-init:flag}
\end{gather}
for each $t \in \{1, 2\}$. The last condition here is not
strictly necessary for our proof, but it is needed to ensure
that Peterson's algorithm makes progress.

\begin{lemma}[Peterson's C11 Invariants]
\label{lem:pet-invs}
% If $(P_0, \sigma_0)$ is an initial state of Peterson's algorithm, 
If $(P, \sigma)$ is a state reachable of Peterson's algorithm, then 
%from this initial state, 
$(P, \sigma)$
satisfies the following for each $t,\notT \in \{1, 2\}$.
\begin{gather}
  \text{$turn$ is an update-only location} \label{pete-prop:update-only-app}\\
  \vturn \detval{\sigma}_1 2 \vee\vturn \detval{\sigma}_2 1  \label{pete-prop:turn-app}\\
  P.\pc_t \in \{\ref{swap-turn}, \ref{busy-wait}, \ref{critical-section}, \ref{unset-flag}\} \implies \vflag_{t} \detval{\sigma}_t \true   \label{pete-prop:flag-app}\\
  P.\pc_t \in \{\ref{busy-wait}, \ref{critical-section}, \ref{unset-flag}\} \implies \vflag_t \hbo{\sigma}\vturn   \label{pete-prop:flag-turn-app}\\
\begin{array}[t]{@{}l@{}}
  P.\pc_t \in \{\ref{busy-wait}, \ref{critical-section}, \ref{unset-flag}\} \wedge  P.\pc_{\notT} \in \{\ref{busy-wait}, \ref{critical-section}, \ref{unset-flag}\} \implies  \\ 
     \ \qquad \qquad \qquad \qquad \vflag_{\notT} \detval{\sigma}_t \true \vee\vturn \detval{\sigma}_{\notT} t
\end{array}
   \label{pete-prop:cross-app}\\
  P.\pc_t = \ref{critical-section} \wedge  P.\pc_{\notT} \in \{\ref{busy-wait}, \ref{critical-section}, \ref{unset-flag}\} \implies \vturn \detval{\sigma}_{\notT} t
  \label{pete-prop:crit-app}\\
  P.\pc_t = \ref{set-flag} \implies \vflag_t \detval{\sigma} \False
  \label{pete-prop:quies-app}
\end{gather}
\end{lemma}
\begin{proof}
We first prove that each property holds in the initial configuration.
Let $(P,\sigma)$ be an initial state. Thus, for each $t$,
$\sigma.\pc_t = \ref{set-flag}$. This is sufficient to show that all of 
the invariants \ref{pete-prop:flag-app},
\ref{pete-prop:flag-turn-app}, \ref{pete-prop:cross-app}
and \ref{pete-prop:crit-app} are true initially, as these invariants
all assume that at least one thread $t$ has $P.\pc_t \neq \ref{set-flag}$.
We show that each remaining invariant
holds as follows:
\begin{description}
\item (\ref{pete-prop:update-only-app}) By definition, every location
is update-only in an initial state
\item (\ref{pete-prop:turn-app}) This follows from {\sc Init}, and the initial condition \ref{pete-init:turn}, $\vturn \detval{\sigma} 0$ or 
$\vturn \detval{\sigma} 1$.
\item (\ref{pete-prop:quies-app}) This follows from {\sc Init}, and the initial condition
$$\wrval(\sigma_0.\last(flag_t)) = \False.$$
\end{description}

We now prove, for each transition that each property is preserved.
Fix a transition $(P,\sigma) \ltsArrow{m, e}_{\Pad} (P',\sigma')$
with $(P,\sigma)$ satisfying the invariants of Figure \ref{lem:pet-invs}.
Also, fix a thread $t$ (thus fixing $\notT$), which is the
thread executing the operation represented by the transition.
For each transition, we prove that each invariant is preserved.
Where appropriate, we do so for both $t$ and $\notT$.
Invariants applied to $\notT$ are marked with a primed label.
We ignore the execution of line 5, as the critical section does not
modify the variables used in Peterson's algorithm.

{\bf Case 1:} $P.\pc_t = \ref{set-flag}$, and $P'.\pc_t = \ref{swap-turn}$ and
$e = W_t(\vflag_t, \true)$. It follows from Lemma \ref{lem:last-mod-trans}, and invariant \ref{pete-prop:quies-app}
that $$m = \sigma.\last(\vflag_t).$$
\begin{description}
\item (\ref{pete-prop:update-only-app}) [$\vturn$ is an update-only location in $\sigma'$].
This follows because $\vturn$ is an update-only location in $\sigma$, and $e \notin \W_{|\vturn}$.
\item (\ref{pete-prop:turn-app}) [$\vturn \detval{\sigma'}_1 2 \vee\vturn \detval{\sigma'}_2 1$].
From the rule {\sc NoMod}, and the fact that $e \notin \W_{|\vturn}$ it follows that
if $\vturn \detval{\sigma}_1 2$ (resp. $\vturn \detval{\sigma}_2 1$) then
$\vturn \detval{\sigma'}_1 2$ (resp. $\vturn \detval{\sigma'}_2 1$),
which is sufficient to prove that the invariant is preserved.
\item (\ref{pete-prop:flag-app}) [$P'.\pc_t \in \{\ref{swap-turn}, \ref{busy-wait}, \ref{critical-section}, \ref{unset-flag}\}
\implies \vflag_{t} \detval{\sigma'}_t \true$].
From rule {\sc ModLast}, the fact that $e \in \W_{|\vflag_t}$ and that $m = \sigma.\last(\vflag_t)$
it follows that $\vflag_t \detval{\sigma'} \wrval(e) = \true$.
\item (\ref{pete-prop:flag-app}') [$P'.\pc_{\notT} \in \{\ref{swap-turn}, \ref{busy-wait}, \ref{critical-section}, 
\ref{unset-flag}\}
\implies \vflag_{\notT} \detval{\sigma'}_{\notT} \true$].
From rule {\sc NoMod} and the fact that $e \notin \W_{|\vflag_{\notT}}$, it follows
that if $\vflag_{\notT} \detval{\sigma}_{\notT} \true$, then
$\vflag_{\notT} \detval{\sigma'}_{\notT} \true$, which is sufficient to prove that
the invariant is preserved.
\item (\ref{pete-prop:flag-turn-app}) [$P'.\pc_t \in \{\ref{busy-wait}, \ref{critical-section}, 
\ref{unset-flag}\} \implies \vflag_t \hbo{\sigma'}\vturn$].
Similar to the proof for Invariant \ref{pete-prop:flag-app}, $P'.\pc_t = 2 \notin \{\ref{busy-wait}, \ref{critical-section}, 
\ref{unset-flag}\}$.
\item (\ref{pete-prop:flag-turn-app}') [$P'.\pc_{\notT} \in \{\ref{busy-wait}, \ref{critical-section}, 
\ref{unset-flag}\} \implies \vflag_{\notT} \hbo{\sigma'}\vturn$].
From rule {\sc NoMod-Ord} and the fact that $e \notin \W_{|\vflag_{\notT}} \cup \W_{|\vturn}$,
it follows that if $\vflag_{\notT} \hbo{\sigma}\vturn$, then
$\vflag_{\notT} \hbo{\sigma'}\vturn$, which is sufficient to prove that
the invariant is preserved.
\item (\ref{pete-prop:cross-app}) [$P'.\pc_t \in \{\ref{busy-wait}, \ref{critical-section}, 
\ref{unset-flag}\} \wedge  P'.\pc_{\notT} \in \{\ref{busy-wait}, \ref{critical-section}, 
\ref{unset-flag}\} \implies 
     \vflag_{\notT} \detval{\sigma'}_t \true \vee\vturn \detval{\sigma'}_{\notT} t$].
     As before, it is sufficient that $P'.\pc_t \notin \{4, 5\}$.
\item (\ref{pete-prop:cross-app}') [$P'.\pc_{\notT} \in \{\ref{busy-wait}, \ref{critical-section}, 
\ref{unset-flag}\} \wedge  P'.\pc_t \in \{\ref{busy-wait}, \ref{critical-section}, 
\ref{unset-flag}\} \implies
  \vflag_t \detval{\sigma'}_{\notT} \true \vee\vturn \detval{\sigma'}_t t$].
It is again sufficient that $P'.\pc_t \notin \{\ref{busy-wait}, \ref{critical-section}, 
\ref{unset-flag}\}$.
\item (\ref{pete-prop:crit-app}) [$P'.\pc_t = \ref{critical-section} \wedge  P'.\pc_{\notT} \in \{\ref{busy-wait}, \ref{critical-section}, 
\ref{unset-flag}\} \implies 
  \vturn \detval{\sigma'}_{\notT} t$]. It is sufficient that $P'.\pc_t \neq 5$.
\item (\ref{pete-prop:crit-app}') [$P'.\pc_{\notT} = \ref{critical-section} \wedge  P'.\pc_t \in \{\ref{busy-wait}, \ref{critical-section}, 
\ref{unset-flag}\} \implies \vturn \detval{\sigma'}_t t$].
It is again sufficient that $P'.\pc_t \notin \{\ref{busy-wait}, \ref{critical-section}, 
\ref{unset-flag}\}$.
\item (\ref{pete-prop:quies-app}) [$P.\pc_t = \ref{set-flag} \implies \vflag_t \detval{\sigma} \False$].
It is sufficient that $P'.\pc_t = \ref{swap-turn}$.
\item (\ref{pete-prop:quies-app}') $P.\pc_{\notT} = \ref{set-flag} \implies \vflag_{\notT} \detval{\sigma} \False$. From rule {\sc NoMod} and the fact that $e \notin \W_{|\vflag_{\notT}}$,
it follows that if $\vflag_{\notT} \detval{\sigma} \False$, then
$\vflag_{\notT} \detval{\sigma'} \False$, which is sufficient to prove that
the invariant is preserved.
\end{description}
For the remaining cases, we do not explicitly state the invariant that we are proving.
The mapping from labels to invariants remains as above.

{\bf Case 2:} $P.\pc_t = \ref{swap-turn}$, and $P'.\pc_t = \ref{busy-wait}$ and $e = U_t(turn, v, \notT)$
for some $v$. By Lemma \ref{lem:last-mod-trans}, because $e$ is an update and
$\vturn$ is an update-only location, $m = \sigma.\last(\vturn)$.
\begin{description}
\item (\ref{pete-prop:update-only-app})
This follows because $e$ is an update.
\item (\ref{pete-prop:turn-app})
From the rule {\sc ModLast}, and that $e \in \W_{|\vturn}$ and $m = \sigma.\last(\vturn)$
it follows that $\vturn \detval{\sigma'}_t \wrval(e) = \notT$, which is sufficient.
\item (\ref{pete-prop:flag-app}) 
Note that by Invariant \ref{pete-prop:flag-app} applied to $(P, \sigma)$,
we have $\vflag_{t} \detval{\sigma}_t \true$.
Then, from the rule {\sc NoMod}, and the fact that $e \notin \W_{|\vflag_t}$ it
follows that $\vflag_{t} \detval{\sigma'}_t \true$ as required.
\item (\ref{pete-prop:flag-app}')
From rule {\sc NoMod} and the fact that $e \notin \W_{|\vflag_{\notT}}$, it follows
that if $\vflag_{\notT} \detval{\sigma}_{\notT} \true$, then
$\vflag_{\notT} \detval{\sigma'}_{\notT} \true$, which is sufficient to prove that
the invariant is preserved.
\item (\ref{pete-prop:flag-turn-app})
Note that by Invariant \ref{pete-prop:flag-app} applied to $(P, \sigma)$,
we have $\vflag_{t} \detval{\sigma}_t \true$.
Then, from the rule {\sc WriteOrd}, and the fact that $\vturn$ and $\vflag_t$
are distinct variables, $e \in \W_{|\vturn}$ and $m = \sigma.\last(\vturn)$,
we have $\vflag_t \hbo{\sigma'}\vturn$ as required.
\item (\ref{pete-prop:flag-turn-app}')
Note first that because $\vturn$ is update only, $m$ is an update
and thus $m \in \W_{R|\vturn}$. Then, from rule {\sc UpdOrd} and the fact that
$e \in \URA_{|\vturn}$ it follows that if $\vflag_{\notT} \hbo{\sigma}\vturn$, then
$\vflag_{\notT} \hbo{\sigma'}\vturn$, which is sufficient to prove that
the invariant is preserved.
\item (\ref{pete-prop:cross-app})
In the proof that this transition preserves Invariant \ref{pete-prop:turn-app}, we proved that
$\vturn \detval{\sigma'}_t \wrval(e) = \notT$, which is sufficient to prove that this current
invariant is preserved.
\item (\ref{pete-prop:cross-app}') Again, we know that $\vturn \detval{\sigma'}_t \wrval(e) = \notT$,
which is sufficient to prove that this invariant is preserved (bearing in mind that $t$ and $\notT$
are transposed in this invariant).
\item (\ref{pete-prop:crit-app}) It is sufficient that $P'.\pc_t \neq \ref{critical-section}$.
\item (\ref{pete-prop:crit-app}') Again, the fact that $\vturn \detval{\sigma'}_t \wrval(e) = \notT$
is enough.
\item (\ref{pete-prop:quies-app}) It is sufficient that $P'.\pc_t \neq \ref{set-flag}$.
\item (\ref{pete-prop:quies-app}') From rule {\sc NoMod} and the fact that $e \notin \W_{|\vflag_{\notT}}$,
it follows that if $\vflag_{\notT} \false$, then
$\vflag_{\notT} \detval{\sigma'}_{\notT} \false$, which is sufficient to prove that
the invariant is preserved.
\end{description}

{\bf Case 3:} In this case, we consider the first test at line \ref{busy-wait} $\vflag_t = \false$.
If this test returns $\true$, then nothing about the state changes except that $t$ moves to
the second test in the condition. Because nothing about the state is changing,
application of the rules {\sc NoMod} and {\sc NoModOrd} can be used to show that all the invariants
are preserved in a standard way. Therefore, we only consider in detail the situation when
the test returns $\false$. Thus, assume that $P.\pc_t = \ref{busy-wait}$, and $P'.\pc_t = \ref{critical-section}$,
and $e = R_t(\vflag_{\notT}, \false)$.

Because $e$ is not a write and the value of $\pc_{\notT}$ does not change,
it is straightforward to use the rules {\sc NoMod} and {\sc NoModOrd}
to show that each invariant except for \ref{pete-prop:crit-app} and \ref{pete-prop:crit-app}'
are preserved. Because it is simpler, we first prove that \ref{pete-prop:crit-app}' is preserved.
Briefly, $P'.\pc_{\notT} = P.\pc_{\notT}$, and because $e \notin \W_{|\vflag_t}$ we have
$\vflag_t \detval{\sigma}_{\notT} \true \implies \vflag_t \detval{\sigma'}_{\notT} \true$
(by rule {\sc NoMod}), and because $e \notin \W_{|\vturn}$ we have
$\vturn \detval{\sigma}_{\notT} \notT \implies \vturn \detval{\sigma'}_{\notT} \notT$.
These three properties are sufficient to show that \ref{pete-prop:crit-app}' is preserved.

We now prove that
\ref{pete-prop:crit-app} is preserved. We do so by proving that $\vturn \detval{\sigma'}_{\notT} t$
under the assumption that $P'.\pc_{\notT} \in \{\ref{busy-wait}, \ref{critical-section}, \ref{unset-flag}\}$.
Because $P.\pc_{\notT} = P'.\pc_{\notT}$, we have $P.\pc_{\notT} \in \{\ref{busy-wait}, \ref{critical-section}, \ref{unset-flag}\}$. Thus, because $P.\pc_t = \ref{busy-wait}$, Invariant \ref{pete-prop:cross-app}
guarantees that
\[
\vflag_{\notT} \detval{\sigma}_t \true \vee \vturn \detval{\sigma}_{\notT} = t
\]
But the disjunct $\vflag_{\notT} \detval{\sigma}_t$ must be false. If it were true,
Lemma \ref{lem:det-val-read} the read $e$ would have to return $\true$, contrary to
the hypothesis that $e = R_t(\vflag_{\notT}, \false)$. Thus $\vturn \detval{\sigma}_{\notT} = t$.
Then, from rule {\sc NoMod}, and the fact that 
$e$ is not a write, we have $\vturn \detval{\sigma'}_{\notT} t$ as required.

{\bf Case 4:} In this case, we consider the second test at line \ref{busy-wait} $\vturn = \notT$.
As before, if this test returns true, then all the invariants are straight-forwardly preserved.
So assume that $P.\pc_t = \ref{busy-wait}$, and $P'.\pc_t = \ref{critical-section}$,
and $e = R_t(\vturn, t)$.

Again, because $e$ is not a write and the value of $\pc_{\notT}$ does not change,
it is easy to show that each invariant except for \ref{pete-prop:crit-app} is preserved.
We show that Invariant \ref{pete-prop:crit-app} is preserved by proving that
$\vturn \detval{\sigma'}_{\notT} t$. By Lemma \ref{lem:det-val-read},
and the fact that $e = R_t(\vturn, t)$
the assertion $\vturn \detval{\sigma}_{t} \notT$ is false. Thus, by Invariant
\ref{pete-prop:turn-app}, $\vturn \detval{\sigma}_{\notT} t$. Then, from rule {\sc NoMod}, and the fact that 
$e$ is not a write, we have $\vturn \detval{\sigma'}_{\notT} t$
as required.

{\bf Case 5:} $P.\pc_t = \ref{unset-flag}$, and $P'.\pc_t = \ref{set-flag}$ and
$e = W_t(\vflag_t, \false)$. It follows from Lemma \ref{lem:last-mod-trans}, and Invariant \ref{pete-prop:flag-app}
that $$m = \sigma.\last(\vflag_t).$$
\begin{description}
\item (\ref{pete-prop:update-only-app})
This invariant is preserved because $e \notin \W_{|\vturn}$.
\item (\ref{pete-prop:turn-app}) This invariant is preserved by rule {\sc NoMod} and
the fact that $e \notin \W_{|\vturn}$.
\item (\ref{pete-prop:flag-app}) Note that $P'.\pc_t \notin \{\ref{swap-turn}, \ref{busy-wait}, \ref{critical-section},
\ref{unset-flag}\}$, which is sufficient to show that this invariant is preserved.
\item (\ref{pete-prop:flag-app}') This invariant is preserved by rule {\sc NoMod} and
the fact that $e \notin \W_{\vflag_{\notT}}$.
\item (\ref{pete-prop:flag-turn-app}) Note that $P'.\pc_t \notin \{\ref{busy-wait}, \ref{critical-section}\}$, which is sufficient to show that this invariant is preserved.
\item (\ref{pete-prop:flag-turn-app}') This invariant is preserved by rule {\sc NoMod} and
the fact that $e \notin \W_{|\vflag_{\notT}} \cup \W_{|\vturn}$.
\item (\ref{pete-prop:cross-app})  Note that $P'.\pc_t \notin \{\ref{busy-wait}, \ref{critical-section}, \ref{unset-flag}\}$, which is sufficient to show that this invariant is preserved.
\item (\ref{pete-prop:cross-app}') Again, it is sufficient to note that $P'.\pc_t \notin \{\ref{busy-wait}, \ref{critical-section}, \ref{unset-flag}\}$ (bearing in mind that $t$ and $\notT$ are transposed in \ref{pete-prop:cross-app}').
\item (\ref{pete-prop:crit-app}) This invariant is preserved by rule {\sc NoMod} and
the fact that $e \notin \W_{|\vturn}$.
\item (\ref{pete-prop:crit-app}') This invariant is preserved by rule {\sc NoMod} and
the fact that $e \notin \W_{|\vturn}$.
\item (\ref{pete-prop:quies-app})  From rule {\sc ModLast} and the fact that $e \in \W_{|\vflag_t}$,
$m = \sigma.\last(\vflag_t)$ and $\wrval(e) = \mathit{false}$, we have $\vflag_t \detval{\sigma'}_t \False$
as required.
\item (\ref{pete-prop:quies-app}') This invariant is preserved by rule {\sc NoMod} and
the fact that $e \notin \W_{|\vflag_{\notT}}$.
\end{description}
This completes our proof.
\end{proof}

\section{Mechanisation in MemAlloy}

The {\tt .cat} files for MemAlloy   \smallskip

\url{https://github.com/johnwickerson/memalloy} \smallskip

\noindent
are given below.

\begin{itemize}
\item {\tt c11\_rar.cat} contains our formalisation of the RAR fragment.
\item {\tt c11\_simp\_2.cat} is the same as {\tt c11\_simp.cat}, distributed with MemAlloy, but imports {\tt c11\_base\_rar.cat} instead of {\tt c11\_base.cat}.

\item {\tt c11\_base\_rar.cat} contains definitions common to both 
{\tt c11\_rar.cat} and {\tt c11\_simp\_2.cat}. It is essentially the file distributed with MemAlloy, but with a simplified $\ltsw$ relation  that ignores release sequences. It also omits events such as {\sf SC} events that are not part of our C11 model. 
\end{itemize}

No differences were found between {\tt c11\_rar.cat} and {\tt c11\_simp\_2.cat} for models up to size 7.

As a sanity check, both {\tt c11\_rar.cat} and {\tt c11\_simp\_2.cat} were compared against  {\tt c11\_lidbury.cat}, which formalises \citet{DBLP:conf/popl/LidburyD17}, revealing the exact same sets of counterexamples.

\subsection{File {\tt c11\_rar.cat}}

\begin{verbatim}
(* This file imports c11_rar_base.cat and 
rephrases the acyclicity axiom in terms of eco *)

"C"
include "c11_rar_base.cat"

let eco = (rf | co | fr)+

irreflexive hb as hb_irr
irreflexive hb ; eco as hb_eco_irr
irreflexive eco as eco_irr
\end{verbatim}

\subsection{File {\tt c11\_simp\_2.cat}}

\begin{verbatim}

(* This is the C11 file distributed with 
Memalloy, but imports c11_rar_base.cat  
instead of c11_base.cat *)

"C"
include "c11_rar_base.cat"

acyclic scp as Ssimp
\end{verbatim}

\subsection{File c11\_base\_rar.cat}

\begin{verbatim}
"C"
include "basic.cat"

(* Modifications to c11_base.cat *)

(* synchronises with (sw) simplified to 
elide release sequences *)
let sw = [REL]; rf; [ACQ]

empty [SC] as omitSC
empty [NAL] as omitNAL
empty [F] as omitF
empty [R & W] \ [REL & ACQ] as RAOnlyRMW

(* Definitions below are from c11_base.cat *)

let fsb = [F]; po
let sbf = po; [F]

(* release sequence *)
let rs = poloc*; rf*

(* happens before *)
let hb = (po | sw)+
let hbl = hb & loc

(* conflict *)
let cnf = (((W*M) | (M*W)) & loc) \ id
				  
(* data race *)
let dr = (cnf \ (A*A)) \ thd \ (hb | hb^-1)
undefined_unless empty dr as Dr 

(* unsequenced race *)
let ur = (cnf & thd) \ (po | po^-1)
undefined_unless empty ur as Ur

(* coherence, etc *)
acyclic hbl | rf | co | fr as HbCom

(* no "if(r==0)" *)
deadness_requires empty if_zero as No_If_Zero 

(* no unsequenced races *)
deadness_requires empty ur as Dead_Ur

(* coherence edges are forced *)
deadness_requires empty unforced_co as Forced_Co

(* external control dependency *)
let cde = ((rf \ thd) | ctrl)* ; ctrl
(* dependable release sequence *)
let drs = rs \ ([R]; !cde)
(* dependable synchronises-with *)
let dsw = 
sw & (((fsb?; [REL]; drs?) \ (!ctrl; !cde)) ; rf)		 
(* dependable happens-before *)
let dhb = po?; (dsw;ctrl)*
(* self-satisfying cycle *)	    
let ssc = id & cde
(* potential data race *)
let pdr = cnf \ (A*A)
(* reads-from on non-atomic location *)
let narf = rf & (NAL*NAL)

deadness_requires 
empty pdr \ (dhb | dhb^-1 | narf;ssc | ssc;narf^-1) 
as Dead_Pdr

let scb = fsb?; (co | fr | hb); sbf?
let scp = (scb & (SC * SC)) \ id

\end{verbatim}

\end{document}